\documentclass[12pt]{article}
\usepackage{hyperref}
\usepackage{amsmath}
\usepackage{graphicx}
\usepackage{url} 
\usepackage{float}
\usepackage{microtype}
\usepackage{booktabs} 
\usepackage{mathtools}
\usepackage{amsfonts}
\usepackage{amssymb}
\usepackage{bbold}
\usepackage{graphics}
\usepackage{amsthm}
\usepackage{blindtext}
\usepackage{hyperref}

\usepackage{algorithmicx}
\usepackage{algorithm}
\usepackage{algpseudocode}
\usepackage{natbib}

\usepackage[usestackEOL]{stackengine}
\usepackage{subfig}
\newtheorem{theorem}{Theorem}

\newtheorem{proposition}[theorem]{Proposition}

\newtheorem{definition}{Definition}[section]

\usepackage{hyperref}

\newtheoremstyle{example}{}{}{}{}{\bfseries}{\smallskip}{\newline}{}
\theoremstyle{example}
\newtheorem{example}{Example}
\usepackage{lipsum}

\DeclareMathOperator*{\argmax}{arg\,max}

\newcommand{\blind}{1}

\begin{document}

\def\spacingset#1{\renewcommand{\baselinestretch}%
{#1}\small\normalsize} \spacingset{1}


\if1\blind
{
  \title{\bf Improving Bridge estimators via $f$-GAN}
  \author{Hanwen Xing\thanks{
    The author would like to thank Prof. Geoff Nicholls and Prof. Kate Lee
for helpful and constructive discussions.}\hspace{.2cm}\\
    Department of Statistics, University of Oxford, UK\\}
  \maketitle
} \fi

\if0\blind
{
  \bigskip
  \bigskip
  \bigskip
  \begin{center}
    {\LARGE\bf Improving Bridge estimators via $f$-GAN}
\end{center}
  \medskip
} \fi

\bigskip
\begin{abstract}
Bridge sampling is a powerful Monte Carlo method for estimating ratios of normalizing constants. Various methods have been introduced to improve its efficiency. These methods aim to increase the overlap between the densities by applying appropriate transformations to them without changing their normalizing constants. In this paper, we first give a new estimator of the asymptotic relative mean square error (RMSE) of the optimal Bridge estimator by equivalently estimating an $f$-divergence between the two densities. We then utilize this framework and propose $f$-GAN-Bridge estimator ($f$-GB) based on a bijective transformation that maps one density to the other and minimizes the asymptotic RMSE of the optimal Bridge estimator with respect to the densities. This transformation is chosen by minimizing a specific $f$-divergence between the densities. We show $f$-GB is optimal in the sense that within any given set of candidate transformations, the $f$-GB estimator can  asymptotically achieve an RMSE lower than or equal to that achieved by Bridge estimators based on any other transformed densities. Numerical experiments show that $f$-GB outperforms existing methods in simulated and real-world examples. In addition, we discuss how Bridge estimators naturally arise from the problem of $f$-divergence estimation.

\end{abstract}


\noindent%
\vfill

\newpage
\spacingset{1.5} 

\section{Introduction} \label{sec:intro}

Estimating the normalizing constant of an unnormalized probability density, or the ratio of normalizing constants between two unnormalized densities is a challenging and important task. In Bayesian inference, such problems are closely related to estimating the marginal likelihood of a model or the Bayes factor between two competing models, and can arise from fields such as econometrics \citep{geweke1999using}, astronomy \citep{bridges2009bayesian}, phylogenetics \citep{fourment202019}, etc. Monte Carlo methods such as Bridge sampling \citep{bennett1976efficient,meng1996simulating}, path sampling \citep{gelman1998simulating}, reverse logistic regression \citep{geyer1994estimating}, nested sampling \cite{skilling2006nested} and reverse Annealed Importance Sampling \citep{burda2015accurate} have been proposed to address this problem. 
See \cite{friel2012estimating} for an overview of some popular algorithms. \cite{fourment202019} also compare the empirical performance of 19 algorithms for estimating normalizing constants in the context of phylogenetics.

Bridge sampling \citep{bennett1976efficient,meng1996simulating} is a powerful, easy-to-implement Monte Carlo method for estimating the ratio of normalizing constants. Let $\tilde q_i(\omega), \omega \in \Omega_i, \ i=1,2$ be two unnormalized probability densities with respect to a common measure $\mu$. Let $q_i(\omega) = \tilde q_i(\omega)/Z_i$ be the corresponding normalized density, where $Z_i$ is the normalizing constant. Bridge sampling estimates $r=Z_1/Z_2$ using samples from $q_1,q_2$ and the unnormalized density functions $\tilde q_1, \tilde q_2$. \cite{meng2002warp} point out that Bridge sampling is equally useful for estimating a single normalizing constant. The relative mean square error (RMSE) of a Bridge estimator depends on the overlap or ``distance" between $q_1,q_2$. The overlap can be quantified by some divergence between them.  When $q_1,q_2$ share little overlap, the corresponding Bridge estimator has large RMSE and therefore is unreliable. In order to improve the efficiency of Bridge estimators, various methods such as Warp Bridge sampling \citep{meng2002warp}, Warp-U Bridge sampling \citep{wang2020warp} and Gaussianized Bridge sampling \citep{jia2020normalizing} have been introduced. These methods first apply transformations $T_i$ to $q_i$ in a tractable way without changing the normalizing constant $Z_i$ for $i=1,2$, then compute Bridge estimators based on the transformed densities $q_i^{(T)}$ and the corresponding samples for $i=1,2$. If $q_1^{(T)}, q_2^{(T)}$ have greater overlap than the original ones, then the resulting Bridge estimator of $r$ based on $q_1^{(T)}, q_2^{(T)}$ would have a lower RMSE.

In this paper, we first demonstrate the connection between Bridge estimators and $f$-divergence \citep{ali1966general}. We show that one can estimate the asymptotic RMSE of the optimal Bridge estimator by equivalently estimating a specific $f$-divergence  between $q_1,q_2$. \cite{nguyen2010estimating} propose a general variational framework for $f$-divergence estimation. We apply this framework to our problem and obtain a new estimator of the asymptotic RMSE of the optimal Bridge estimator using the unnormalized densities $\tilde q_1, \tilde q_2$ and the corresponding samples. We also find a connection between Bridge estimators and the variational lower bound of $f$-divergence given by \cite{nguyen2010estimating}. In particular, we show that the problem of estimating an $f$-divergence between $q_1,q_2$ using this variational framework naturally leads to a Bridge estimator of $r=Z_1/Z_2$.  \cite{kong2003theory} observe that the optimal Bridge estimator is a maximum likelihood estimator under a semi-parametric formulation. We use this $f$-divergence estimation framework to extend this observation and show that many special cases of Bridge estimators such as the geometric Bridge estimator can also be interpreted as maximizers of some objective functions that are related to the $f$-divergence between $q_1,q_2$. This formulation also connects Bridge estimators and density ratio estimation problems: Since we can evaluate the unnormalized densities $\tilde q_1, \tilde q_2$, we know the true density ratio up to a multiplicative constant $r=Z_1/Z_2$. Hence estimating $r$ can be viewed as a problem of estimating the density ratio between $q_1,q_2$. A similar idea has been explored in e.g. Noise Contrastive Estimation \citep{gutmann2010noise}, where the authors treat the unknown normalizing constant as a model parameter, and cast the estimation problem as a classification problem. Similar ideas have also been discussed in e.g. \cite{geyer1994estimating} and \cite{uehara2016generative}.

We then utilize the connection between the asymptotic RMSE of the optimal Bridge estimator and a specific $f$-divergence between $q_1,q_2$, and propose $f$-GAN-Bridge estimator ($f$-GB), which improves the efficiency of the optimal Bridge estimator of $r$ by directly minimizing the first order approximation of its asymptotic RMSE with respect to the \emph{densities} using an $f$-GAN. $f$-GAN \citep{nowozin2016f} is a class of generative model that aims to approximate the target distribution by minimizing an $f$-divergence between the generative model and the target. Let $\mathcal{T}$ be a collection of transformations $T$ such that $\tilde q_1^{(T)}$, the transformed unnormalized density of $q_1$ is computationally tractable and have the same normalizing constant $Z_1$ as the original $\tilde q_1$. The $f$-GAN-Bridge estimator is obtained using a two-step procedure: We first use the $f$-GAN framework to find the transformation $T^*$ that minimizes a specific $f$-divergence between $q_1^{(T)}$ and $q_2$ with respect to $T \in \mathcal{T}$. Once $T^*$ and $q_1^{(T^*)}$ are chosen, we then compute the optimal Bridge estimator of $r$ based on $q_1^{(T^*)}$ and $q_2$ as the $f$-GAN-Bridge estimator. We show $T^*$ asymptotically minimizes the first order approximation of the asymptotic RMSE of the optimal Bridge estimator based on $q_1^{(T)}$ and $q_2$ with respect to $T$. In contrast, existing methods such as Warp Bridge sampling \citep{meng2002warp,wang2020warp} and Gaussianized Bridge sampling \citep{jia2020normalizing} do not offer such theoretical guarantee.
The transformed $q_1^{(T)}$ can be parameterized in any way as long as it is computationally tractable and preserves the normalizing constant $Z_1$. In this paper, we parameterize $q_1^{(T)}$ as a Normalizing flow \citep{rezende2015variational, dinh2016density} with base density $q_1$ because of its flexibility.

\subsection{Summary of our contributions}

The main contribution of our paper is that we give a computational framework to improve the optimal Bridge estimator by minimizing the first order approximation of its asymptotic RMSE with respect to the densities. We also give a new estimator of the asymptotic RMSE of the optimal Bridge estimator using the variational framework proposed by \cite{nguyen2010estimating}. This formulation allows us to cast the estimation problem as a 1-d optimization problem. We find the $f$-GAN-Bridge estimator outperforms existing methods significantly in both simulated and real-world examples. Numerical experiments show that the proposed method provides not only a reliable estimate of $r$, but also an accurate estimate of its RMSE. In addition, we also find a connection between Bridge estimators and the problem of $f$-divergence estimation, which allows us to view Bridge estimators from a different perspective.

This paper is structured as follows: In Section \ref{sec: bridge}, we briefly review Bridge sampling and existing improvement strategies. In Section \ref{sec: fdivandbridge}, we give a new estimator of the asymptotic RMSE of the optimal Bridge estimator using the variational framework for $f$-divergence estimation \citep{nguyen2010estimating}. We also demonstrate the connection between Bridge estimators and the problem of $f$-divergence estimation. In Section \ref{sec: fgan}, we introduce the $f$-GAN-Bridge estimator and give implementation details. We give both simulated and real-world examples in Section \ref{sec: mixture of rings}, \ref{sec: glmm}. Section \ref{sec: conclusion} concludes the paper with a discussion. A Python implementation of the proposed method alongside with examples can be found on Github. A Python implementation of the proposed method alongside with examples can be found in \url{https://github.com/hwxing3259/Bridge_sampling_and_fGAN}.

\section{Bridge sampling and related works} \label{sec: bridge}

Let $Q_1,Q_2$ be two probability distributions of interest. Let $q_i(\omega), \omega \in \Omega_i, i=1,2$ be the densities of $Q_1,Q_2$ with respect to a common measure $\mu$ defined on $\Omega_1 \cup \Omega_2$, where $\Omega_1$ and $\Omega_2$ are the corresponding supports. We use $\tilde q_i(\omega), i=1,2$ to denote the unnormalized densities and $Z_i, i=1,2$ to denote the corresponding normalizing constants, i.e. $q_i(\omega) = \tilde q_i(\omega) / Z_i$ for $i=1,2$. Suppose we have samples from $q_1,q_2$, but we are only able to evaluate the \emph{unnormalized} densities $\tilde q_i(\omega), i=1,2$. Our goal is to estimate the ratio of normalizing constants $r = Z_1/Z_2$ using only $\tilde q_i(\omega), i=1,2$ and samples from the two distributions. Bridge sampling \citep{bennett1976efficient,meng1996simulating} is a powerful method for this task. 
\begin{definition}[Bridge estimator]
Suppose $\mu(\Omega_1 \cap \Omega_2)>0$ and $\alpha : \Omega_1 \cap \Omega_2 \to \mathbb{R}$ satisfies
$0 < \left\vert\int_{\Omega_1 \cap \Omega_2} \alpha(\omega) q_1(\omega)q_2(\omega)d\mu(\omega)\right\vert < \infty$. Given samples $\{\omega_{ij}\}_{j=1}^{n_i} \sim q_i$ for $i=1,2$, the Bridge estimator $\hat r_\alpha$ of $r=Z_1/Z_2$ is defined as
\begin{equation}
    \hat r_\alpha = \frac{{n_2}^{-1}\sum_{j=1}^{n_2}\alpha(\omega_{2j}) \tilde q_1(\omega_{2j}) }{{n_1}^{-1}\sum_{j=1}^{n_1} \alpha(\omega_{1j})\tilde q_2(\omega_{1j}) } \label{bridgeestimate}
\end{equation}
\end{definition}

The choice of free function $\alpha$ directly affects the quality of $\hat r_\alpha$, which is quantified by the relative mean square error (RMSE) $E(\hat r_{\alpha} - r)^2/r^2$. Let $n = n_1+n_2$ and $s_i = n_i/n$ for $i=1,2$. Let $RE^2(\hat r_\alpha)$ denote the asymptotic RMSE of $\hat r_\alpha$ as $n_1,n_2 \rightarrow \infty$. Under the assumption that the samples from $q_1,q_2$ are i.i.d., \cite{meng1996simulating} show the optimal $\alpha$ which minimizes the first order approximation of $RE^2(\hat r_\alpha)$ takes the form 
\begin{equation}
    \alpha_{opt}(\omega)  \propto \frac{1}{s_1 \tilde q_1(\omega) + s_2 \tilde q_2(\omega) r}, \quad \omega \in \Omega_1 \cap \Omega_2 \label{alphaopt}
\end{equation}
The resulting $RE^2(\hat r_{\alpha_{opt}})$ with the optimal free function $\alpha_{opt}$ is 
\begin{equation}
    RE^2(\hat r_{\alpha_{opt}}) = \frac{1}{ns_1s_2} \left[ \left(\int_{\Omega_1 \cap \Omega_2} \frac{q_1(\omega)q_2(\omega)}{s_1 q_1(\omega) + s_2 q_2(\omega)} d\mu(\omega) \right)^{-1} -1\right] + o\left(\frac{1}{n}\right). \label{minimalRMSE}
\end{equation}

Note that the optimal $\alpha_{opt}$ is not directly usable as it depends on the unknown quantity $r$ we would like to estimate in the first place. To resolve this problem, \cite{meng1996simulating} give an iterative procedure
\begin{equation}
    \hat r^{(t+1)} = \frac{n_2^{-1}\sum_{j=1}^{n_2} \tilde q_1(\omega_{2j})/(s_1\tilde q_1(\omega_{2j}) + s_2 \tilde q_2(\omega_{2j}) \hat r^{(t)})}{n_1^{-1}\sum_{j=1}^{n_1}\tilde q_2(\omega_{1j})/(s_1\tilde q_1(\omega_{1j}) + s_2 \tilde q_2(\omega_{1j}) \hat r^{(t)})}, \quad t=0,1,2,... \label{iterativeoptimal}
\end{equation}
The authors show that for any initial value $\hat r^{(0)}$, $\hat r^{(t)}$ is a consistent estimator of $r$ for all $t \geq 1$, and the sequence $\{\hat r^{(t)}\}, \ t=0,1,2,...$ converges to the unique limit $\hat r_{opt}$. Let $MSE(\log \hat r_{opt})$ denote the asymptotic mean square error of $\log \hat r_{opt}$. 

Under the i.i.d. assumption, the authors also show
$RE^2(\hat r_{opt})$ and $MSE(\log \hat r_{opt})$ are asymptotically equivalent to $RE^2(\hat r_{\alpha_{opt}})$ in \eqref{minimalRMSE} up to the first order (i.e. they have the same leading term). Note that $\hat r_{opt}$ can be found numerically while $\hat r_{\alpha_{opt}}$ is not directly computable. We will focus on the asymptotically optimal Bridge estimator $\hat r_{opt}$ for the rest of the paper.

\subsection{Improving Bridge estimators via transformations}  \label{sec: transform}

From \eqref{minimalRMSE} and the fact that $RE^2(\hat r_{opt})$ and $RE^2(\hat r_{\alpha_{opt}})$ are asymptotically equivalent, we see $RE^2(\hat r_{opt})$ depends on the overlap between $q_1$ and $q_2$. Even when $\Omega_1 = \Omega_2$, if $q_1$ and $q_2$ put their probability mass on very different regions, the integral in \eqref{minimalRMSE} would be close to 0, leading to large RMSE and unreliable estimators. In order to improve the performance of $\hat r_{opt}$, one may apply transformations to $q_1,q_2$ (and to the corresponding samples) to increase their overlap while keeping the \emph{transformed} unnormalized densities computationally tractable and the normalizing constants unchanged. We assume that we are dealing with unconstrained, continuous random variables with a common support, i.e. $\Omega_1 = \Omega_2 = \mathbb{R}^d$. When the supports $\Omega_1, \Omega_2$ are constrained or different from each other, we can usually match them by applying simple invertible transformations to $q_1$, $q_2$. When $\Omega_1$,$\Omega_2$ have different dimensions, \cite{chen1997estimating} suggest matching the dimensions of $q_1,q_2$ by augmenting the lower dimensional distribution using some completely known random variables (See Appendix \ref{appendix:dimensionmatching} for details).

\cite{voter1985monte} gives a method to increase the overlap in the context of free energy estimation by shifting the samples from one distribution to the other and matching their modes. \cite{meng2002warp} extends this idea and consider more general mappings. Let $T_i: \mathbb{R}^d \to \mathbb{R}^d$, $i=1,2$ be two smooth and invertible transformations that aim to bring $q_1,q_2$ ``closer". For $\omega_i \sim q_i$, define $\omega_i^{(T)} = T_i(\omega_i)$, $\ i=1,2$. Then for $i=1,2$, the distribution of the transformed sample $\omega_i^{(T)}$ has density 
\begin{align}
    q_i^{(T)}(\omega_i^{(T)}) &= \tilde q_i(T_i^{-1}(\omega_i^{(T)}))\left\vert \det J_i(\omega_i^{(T)})\right\vert/Z_i \\
     &\equiv \tilde{q}_i^{(T)}(\omega_i^{(T)})/Z_i, \quad i=1,2\label{transformeddensity}
\end{align}
where $\tilde{q}_i^{(T)}$ is the unnormalized version of $ q_i^{(T)}$, $T_i^{-1}$ is the inverse transformation of $T_i$ and $J_i(\omega)$ is its Jacobian. One can then apply \eqref{bridgeestimate} to the \emph{transformed} samples and the corresponding unnormalized densities $\tilde{q}_1^{(T)},\tilde{q}_2^{(T)}$, and obtain a Bridge estimator
\begin{equation}
        \hat r^{(T)}_\alpha = \frac{n_2^{-1} \sum_{j=1}^{n_2} \tilde{q}_1^{(T)}(T_2(\omega_{2j})) \alpha(T_2(\omega_{2j}))}{n_1^{-1} \sum_{j=1}^{n_1} \tilde{q}_2^{(T)}(T_1(\omega_{1j})) \alpha(T_1(\omega_{1j}))} \label{transformedbridge}
\end{equation} without the need to sample from $\tilde{q}_1^{(T)}$ or $\tilde{q}_2^{(T)}$ separately. Let $\hat r_{opt}^{(T)}$ denote the asymptotically optimal Bridge estimator based on the transformed densities. We stress that the superscript of $\hat r^{(t)}$ in \eqref{iterativeoptimal} indicates the number of iterations, while the superscript in $\hat r_{opt}^{(T)}$ means it is based on the transformed densities. If the transformed $q_1^{(T)}, q_2^{(T)}$ have a greater overlap than the original $q_1,q_2$, then $\hat r^{(T)}_{opt}$ should be a more reliable estimator of $r$ with a lower RMSE.
\cite{meng2002warp} further extend this idea and propose the Warp transformation, which aims to increase the overlap by centering, scaling and symmetrizing the two densities $q_1,q_2$. One limitation of the Warp transformation is that it does not work well for multimodal distributions. \cite{wang2020warp} propose the Warp-U transformation to address this problem. The key idea of the Warp-U transformation is to first approximate $q_i$ by a mixture of Normal or $t$ distributions, then construct a coupling between them which allows
us to map $ q_i$ into a unimodal density in the same way as mapping the mixture back to a single standard Normal or $t$ distribution.

An alternative to the Warp transformation \citep{meng2002warp} is a Normalizing flow. A Normalizing flow (NF) \citep{rezende2015variational, dinh2016density,papamakarios2017masked} parameterizes a continuous probability distribution by mapping a simple base distribution (e.g. standard Normal) to the more complex target using a bijective transformation $T$, which is parameterized as a composition of a series of smooth and invertible mappings $f_1,...,f_K$ with easy-to-compute Jacobians. 
This $T$ is applied to the ``base" random variable $z_0 \sim p_0$, where $z_0 \in \mathbb{R}^d$ and $p_0$ is the known base density. Let
\begin{equation}
     z_k = f_k \circ f_{k-1} \circ ... \circ f_1(z_0), \quad k=1,...,K
\end{equation}
Since the transformation $T$ is smooth and invertible, by applying change of variable repeatedly, the final output $z_K$ has density
\begin{align}
    p_K(z_K) = p_0(z_0)\prod_{k=1}^K \left\vert \det J_k(z_{k-1}) \right\vert^{-1} \label{nfdensity}
\end{align}
where $J_k$ is the Jacobian of the mapping $f_k$. The final density $p_K$ can be used to approximate target distributions with complex structure, and one can sample from $p_K$ easily by applying $T = f_K \circ f_{K-1} \circ ... \circ f_1$ to $z_0 \sim p_0$. In order to evaluate $p_K$ efficiently, we are restricted to transformations $f_k$ whose $\det J_k(z)$ is easy to compute. For example, Real-NVP \citep{dinh2016density} uses the following transformation: For $m \in \mathbb{N}$ such that $1<m<d$, let  $z_{1:m}$ be the first $m$ entries of $z \in \mathbb{R}^d$, let $\times$ be element-wise multiplication and let $\mu_k, \sigma_k: \mathbb{R}^{m} \to \mathbb{R}^{d-m}$ be two mappings (usually parameterized by neural nets). The smooth and invertible transformation $y = f_k(z)$ for each step $k$ in Real-NVP is defined as
\begin{equation}
    y_{1:m} = z_{1:m}, \quad y_{m+1:d} = \mu_k(z_{1:m}) + \sigma_k(z_{1:m}) \times z_{m+1:d} \label{realNVP}
\end{equation}
This means $f_k$ keeps the first $m$ entries of input $z$, while shifting and scaling the remaining ones. The Jacobian $J_k$ of $f_k$ is lower triangular, hence $\det J_k(z) = \prod_{i=1}^{d-m} \sigma_{ik}(z_{1:m})$, where $\sigma_{ik}(z_{1:m})$ is the $i$th entry of $\sigma_k(z_{1:m})$. Each transformation $f_k$ is also called a coupling layer.  When composing a series of coupling layers $f_1,...,f_K$, the authors also swap the ordering of indices in \eqref{realNVP} so that the dimensions that are kept unchanged in one step $k$ are to be scaled and shifted in the next step. 
\cite{jia2020normalizing} utilize the idea of transforming $q_i$ using a Normalizing flow, and propose Gaussianzed Bridge sampling (GBS) for estimating a single normalizing constant. The authors set $q_1$ to be a completely known density, e.g. standard multivariate Normal, and aim to approximate the target $q_2$ using a Normalizing flow with base density $q_1$. The transformed density $q_1^{(T)}$ is estimated by matching the marginal distributions between $q_1^{(T)}$ and $q_2$. Once $q_1^{(T)}$ is chosen, the authors use \eqref{transformedbridge} and the iterative procedure \eqref{iterativeoptimal} to form the asymptotically optimal Bridge estimator of $Z_2$ based on the transformed $q_1^{(T)}$ and the original $q_2$. 

The idea of increasing overlap via transformations is also applicable to discrete random variables. For example, \cite{meng2002warp} suggest using swapping and permutation to increase the overlap between two discrete distributions. \cite{tran2019discrete} also give Normalizing flow models applicable to discrete random variables based on modulo operations. We give a toy example of increasing the overlap between two discrete distributions using Normalizing flows in Appendix \ref{appendix: discrete}. In the later sections, we will extend the idea of increasing overlap via transformations, and propose a new strategy to improve $\hat r^{(T)}_{opt}$ by directly minimizing the first order approximation of $RE^2(\hat r^{(T)}_{opt})$ with respect to the transformed densities.

\section{Bridge estimators and $f$-divergence estimation}\label{sec: fdivandbridge}

\cite{fruhwirth2004estimating} gives an MC estimator of $RE^2(\hat r_{opt})$. In this section, we introduce an alternative estimator of $RE^2(\hat r_{opt})$ and $MSE(\log \hat r_{opt})$ by equivalently estimating an $f$-divergence between $q_1,q_2$. This formulation allows us to utilize the variational lower bound of  $f$-divergence given by \cite{nguyen2010estimating}, and cast the problem of estimating $RE^2(\hat r_{opt})$ as a 1-d optimization problem. In the later section, we will also show how to use this new estimator to improve the efficiency of $\hat r^{(T)}_{opt}$. In addition, we find that estimating different choices of $f$ divergences under the variational framework proposed by \cite{nguyen2010estimating} naturally leads to Bridge estimators of $r$ with different choices of free function $\alpha(\omega)$.

\subsection{Estimating $RE^2(\hat r_{opt})$ via $f$-divergence estimation}
$f$-divergence \citep{ali1966general} is a broad class of divergences between two probability distributions. By choosing $f$ accordingly, one can recover common divergences between probability distributions such as KL divergence $KL(q_1,q_2)$, Squared Hellinger distance $H^2(q_1,q_2)$ and total variation distance $d_{TV}(q_1,q_2)$. 
\begin{definition}[$f$-divergence]
Suppose the two probability distributions $Q_1,Q_2$ have absolutely continuous density functions $q_1$ and $q_2$ with respect to a base measure $\mu$ on a common support $\Omega$. Let the generator function $f: \mathbb{R}^+ \rightarrow \mathbb{R}$ be a convex and lower semi-continuous function satisfying $f(1)=0$.  The $f$-divergence $D_f(q_1,q_2)$ defined by $f$ takes the form 
\begin{equation}
    D_f(q_1,q_2) = \int_{\Omega} f\left( \frac{q_1(\omega)}{q_2(\omega)} \right) q_2(\omega) d\mu(\omega)
\end{equation}
\end{definition}
Unless otherwise stated, we assume $\Omega = \mathbb{R}^d$ where $d \in \mathbb{N}$ i.e. both $q_1$ and $q_2$ are defined on $\mathbb{R}^d$. If the densities $q_1,q_2$ have different or disjoint supports $\Omega_1$, $\Omega_2$, then we apply appropriate transformations and augmentations discussed in the previous sections to ensure that the transformed and augmented densities (if necessary) are defined on the common support $\Omega=\mathbb{R}^d$. 
In this paper, we focus on a particular choice of $f$-divergence that is closely related to $RE^2(\hat r_{opt})$ in \eqref{minimalRMSE}.
\begin{definition}(Weighted harmonic divergence)
Let $q_1,q_2$ be continuous densities with respect to a base measure $\mu$ on the common support $\Omega$. The weighted harmonic divergence is defined as
\begin{equation}
    H_{\pi}(q_1,q_2) = 1 - \int_{\Omega} \left(\pi q_1^{-1}(\omega) + (1-\pi) q_2^{-1}(\omega)\right)^{-1}  d\mu(\omega) \label{harmonic_def}
\end{equation}
where $\pi \in (0,1)$ is the weight parameter.
\end{definition}
\cite{wang2020warp} observe that the weighted harmonic divergence $H_{\pi}(q_1,q_2)$ is an $f$-divergence with generator $f(u) = 1-\frac{u}{\pi+(1-\pi)u}$, and $RE^2(\hat r_{opt})$ can be rearranged as
\begin{equation}
    RE^2(\hat r_{opt}) = (s_1s_2n)^{-1} \left( (1 - H_{s_2}(q_1,q_2))^{-1} -1\right)+ o\left(\frac{1}{n}\right).\label{rmseinH}
\end{equation}
The same statement also holds for $MSE(\log \hat r_{opt})$ since $MSE(\log \hat r_{opt})$ is asymptotically equivalent to $RE^2(\hat r_{opt})$ \citep{meng1996simulating}. This means if we have an estimator of $H_{s_2}(q_1,q_2)$, then we can plug it into the leading term of the right hand side of \eqref{rmseinH} and obtain an estimator of the first order approximation of $RE^2(\hat r_{opt})$ and $MSE(\log \hat r_{opt})$. Before we give the estimator of $H_{s_2}(q_1,q_2)$, we first introduce the variational framework for $f$-divergence estimation proposed by \cite{nguyen2010estimating}.
Every convex, lower semi-continuous function $f: \mathbb{R}^+ \rightarrow \mathbb{R}$ has a convex conjugate $f^*$ which is defined as follows,
\begin{definition}(Convex conjugate)
Let $f: \mathbb{R}^+ \rightarrow \mathbb{R}$ be a convex and lower semi-continuous function. The convex conjugate of $f$ is defined as 
 \begin{equation}
     f^*(t) = \sup_{u \in \mathbb{R}^+} \{ut-f(u)\}
 \end{equation}
\end{definition}

\cite{nguyen2010estimating} show that any $f$-divergence $D_f(q_1,q_2)$ satisfies
\begin{align}
    D_f(q_1,q_2) &\geq \sup_{V \in \mathcal{V}} \Big( E_{q_1}[V(\omega)] - E_{q_2}[f^*(V(\omega))]\Big), \label{fdivlowerbound}
\end{align}
where $\mathcal{V}$ is an arbitrary class of functions $V: \Omega \to \mathbb{R}$, and $f^*(t)$ is the convex conjugate of the generator $f$ which characterizes the $f$-divergence $D_f(q_1,q_2)$. A table of common $f$-divergences with their generator $f$ and the corresponding convex conjugate $f^*$ can be found in \cite{nowozin2016f}. \cite{nguyen2010estimating} show that if $f$ is differentiable and strictly convex, then $D_f(q_1,q_2)$ is equal to $E_{q_1}[V(\omega)] - E_{q_2}[f^*(V(\omega))]$ in \eqref{fdivlowerbound} if and only if $V(\omega) = f'\left(\frac{q_1(\omega)}{q_2(\omega)}\right)$, the first order derivative of $f$ evaluated at $q_1(\omega)/q_2(\omega)$. The authors then give a new strategy of estimating the $f$-divergence $D_f(q_1,q_2)$ by finding the maximum of an empirical estimate of $E_{q_1}[V(\omega)] - E_{q_2}[f^*(V(\omega))]$ in \eqref{fdivlowerbound} with respect to the variational function $V \in \mathcal{V}$.
We now use this framework to give an estimator of $H_{\pi}(q_1,q_2)$. 
\begin{proposition}[Estimating $H_{\pi}(q_1,q_2)$]  \label{prop: estimating_harmonic_divergence}
Let $q_1,q_2$ be continuous densities with respect to a base measure $\mu$ on the common support $\Omega$. Let $\{\omega_{ij}\}_{j=1}^{n_i}$ be samples from $q_i$ for $i=1,2$. Let $\pi \in (0,1)$ be the weight parameter. Let $r$ be the true ratio of normalizing constants between $q_1,q_2$, and $C_2 > C_1 > 0$ be constants such that $r \in [C_1, C_2]$.
For $\tilde r \in [C_1,C_2]$, define 
\begin{multline} \label{harmoniclowerbound1}
        \resizebox{0.92\hsize}{!}{$G(\tilde r; \pi) = 1 - \frac{1}{\pi} E_{q_1}\left(\frac{\pi\tilde q_2(\omega) \tilde r}{(1-\pi) \tilde q_1(\omega) + \pi\tilde q_2(\omega) \tilde r} \right)^2 - \frac{1}{1-\pi} E_{q_2} \left( \frac{(1-\pi) \tilde q_1(\omega)}{(1-\pi)\tilde q_1(\omega) + \pi\tilde q_2(\omega) \tilde r }\right)^2$}
\end{multline} 
Then $H_{\pi}(q_1,q_2)$ satisfies 
\begin{equation} 
    H_{\pi}(q_1,q_2) \geq  G(\tilde r; \pi) \quad  \forall \tilde r \in [C_1,C_2] \label{harmoniclowerbound2},
\end{equation}
and equality holds if and only if $\tilde r=r$. In addition, let 

\begin{multline} \label{harmoniclowerbound_empirical}
        \hat G(\tilde r ; \pi, \{\omega_{ij}\}_{j=1}^{n_i}) = 1 - \frac{1}{\pi n_1} \sum_{j=1}^{n_1}\left(\frac{\pi\tilde q_2(\omega_{1j}) \tilde r}{(1-\pi) \tilde q_1(\omega_{1j}) + \pi\tilde q_2(\omega_{1j}) \tilde r} \right)^2 - \\ \frac{1}{(1-\pi)n_2} \sum_{j=1}^{n_2} \left( \frac{(1-\pi) \tilde q_1(\omega_{2j})}{(1-\pi)\tilde q_1(\omega_{2j}) + \pi\tilde q_2(\omega_{2j}) \tilde r }\right)^2
\end{multline} 
be the empirical estimate of $G(\tilde r; \pi)$ based on $\{\omega_{ij}\}_{j=1}^{n_i} \sim q_i$ for $i=1,2$.

If $\hat r_{\pi} = \argmax_{\tilde r \in [C_1,C_2]} \hat G(\tilde r ; \pi, \{\omega_{ij}\}_{j=1}^{n_i})$, then $\hat r_{\pi}$ is a consistent estimator of $r$, and $\hat G(\hat r_{\pi};\pi, \{\omega_{ij}\}_{j=1}^{n_i})$ is a consistent estimator of $H_{\pi}(q_1,q_2)$ as $n_1,n_2 \rightarrow \infty$.
\end{proposition}
\begin{proof}
 See Appendix \ref{appendix:proofs}.
\end{proof}
Note that \eqref{harmoniclowerbound2} is a special case of the variational lower bound \eqref{fdivlowerbound} with the $f$-divergence $D_f(q_1,q_2) = H_{\pi}(q_1,q_2)$, the corresponding generator $f(u) = 1-\frac{u}{\pi+(1-\pi)u}$ and variational function $V_{\tilde r}(\omega) = f'\left(\frac{\tilde q_1(\omega)}{\tilde q_2(\omega)\tilde r}\right)$ with $\mathcal{V} = \{V_{\tilde r}(\omega) \vert \tilde r \in [C_1,C_2]\}$, i.e. $\tilde r \in [C_1,C_2]$ is the is the sole parameter of $V_{\tilde r}(\omega)$. Note that $V_r(\omega)=f'\left(\frac{q_1(\omega)}{q_2(\omega)}\right)$ since $r$ is the ratio of normalizing constants between $q_1,q_2$. We parameterize the variational function in this specific form because we would like to take the advantage of knowing the unnormalized densities $\tilde q_1, \tilde q_2$ in our setup. Here we assume that $\tilde r \in [C_1,C_2]$ instead of $\tilde r \in \mathbb{R}^+$. This is not a strong assumption, since we can set $C_1$ $(C_2)$ to be arbitrarily small (large).  We take $\hat G(\hat r_{s_2} ; s_2, \{\omega_{ij}\}_{j=1}^{n_i})$ as an estimator of $H_{s_2}(q_1,q_2)$, and define our estimator of the first order approximation of $RE^2(\hat r_{opt})$ as follows:

\begin{definition}[Estimator of $RE^2(\hat r_{opt})$]
Let $\{\omega_{ij}\}_{j=1}^{n_i}$ be samples from $q_i$ for $i=1,2$. Define
\begin{equation}
    \widehat{RE}^2(\hat r_{opt}) = (s_1s_2n)^{-1} \left( (1 - \hat G(\hat r_{s_2} ; s_2, \{\omega_{ij}\}_{j=1}^{n_i}))^{-1} -1\right) \label{minimalrmseestimate}
\end{equation}
as an estimator of the first order approximation of both $RE^2(\hat r_{opt})$ and $MSE(\log \hat r_{opt})$ in \eqref{rmseinH}.
\end{definition}

Even though $\hat G(\hat r_{s_2} ; s_2, \{\omega_{ij}\}_{j=1}^{n_i})$ is a consistent estimator of $H_{s_2}(q_1,q_2)$, it suffers from a positive bias (See Appendix \ref{appendix:positive_bias} for details). We have not found a practical strategy to correct it so far. On the other hand, we believe this bias does not prevent our proposed error estimator $\widehat{RE}^2(\hat r_{opt})$ from being useful in practice. Since our estimator of $RE^2(\hat r_{opt})$ in \eqref{minimalrmseestimate} is a monotonically increasing function of $\hat G(\hat r_\pi; \pi, \{\omega_{ij}\}_{j=1}^{n_i})$ in Prop \ref{prop: estimating_harmonic_divergence}, the positive bias in $\hat G(\hat r_\pi; \pi, \{\omega_{ij}\}_{j=1}^{n_i})$ leads to a positive bias in $\widehat{RE}^2(\hat r_{opt})$. Therefore $\widehat{RE}^2(\hat r_{opt})$ will systemically overestimate the true error $RE^2(\hat r_{opt})$, which will lead to more conservative conclusions (e.g. wider error bars). This is certainly not ideal, but we believe in practice, it is less harmful than underestimating the variability in $\hat r_{opt}$. In addition, we see the proposed error estimator provides accurate estimates of $RE^2(\hat r_{opt})$ in both examples in Sec \ref{sec: mixture of rings} and \ref{sec: glmm}, indicating the effectiveness of it.

\subsection{$f$-divergence estimation and Bridge estimators}
In the last section, we focus on estimating $H_{\pi}(q_1,q_2)$. We now extend the estimation framework to other choices of $f$-divergence, and show how Bridge estimators naturally arise from this estimation problem. Let an $f$-divergence $D_f(q_1,q_2)$ with the corresponding generator $f(u)$ be given. Similar to Proposition \ref{prop: estimating_harmonic_divergence}, under our parameterization of the variational function $V_{\tilde r}$, the empirical estimate of $E_{q_1}[V(\omega)] - E_{q_2}[f^*(V(\omega))]$ in \eqref{fdivlowerbound} becomes 
\begin{align} \label{fdiv_empirical_general}
    \hat G_f(\tilde r; \{\omega_{ij}\}_{j=1}^{n_i}) &=  \frac{1}{n_1} \sum_{j=1}^{n_1} V_{\tilde r}(\omega_{1j}) - \frac{1}{n_2} \sum_{j=1}^{n_2} f^*(V_{\tilde r}(\omega_{2j}))  \\
    &= \frac{1}{n_1} \sum_{j=1}^{n_1} f'\left(\frac{\tilde q_1(\omega_{1j})}{\tilde q_2(\omega_{1j})\tilde r}\right) - \frac{1}{n_2} \sum_{j=1}^{n_2} f^* \circ f'\left(\frac{\tilde q_1(\omega_{2j})}{\tilde q_2(\omega_{2j})\tilde r}\right), \label{scalarobj}
\end{align}
where $\{\omega_{ij}\}_{j=1}^{n_i} \sim q_i$ for $i=1,2$. Let $\hat r^{(f)} = \argmax_{\tilde r \in \mathbb{R}^+} \hat G_f(\tilde r; \{\omega_{ij}\}_{j=1}^{n_i})$.
By \cite{nguyen2010estimating}, $V_{\hat r^{(f)}} = f'\left(\frac{\tilde q_1(\omega)}{\tilde q_2(\omega) \hat r^{(f)}}\right)$ is an estimator of $V_{r}(\omega) = f'\left(\frac{\tilde q_1(\omega)}{\tilde q_2(\omega) r}\right)$, and $\hat G_f(\hat r^{(f)}; \{\omega_{ij}\}_{j=1}^{n_i})$ is an estimator of $D_f(q_1,q_2)$. In Proposition \ref{prop: estimating_harmonic_divergence} we have shown that $\hat r^{(f)}$ and $\hat G_f(\hat r^{(f)}; \{\omega_{ij}\}_{j=1}^{n_i})$ are consistent estimators of $r$ and $D_f(q_1,q_2)$ when $D_f(q_1,q_2)$ is the weighted Harmonic divergence $H_\pi(q_1,q_2)$ \footnote{It is of interest to see if $\hat r^{(f)}$ and $\hat G_f(\hat r^{(f)}; \{\omega_{ij}\}_{j=1}^{n_i})$ are consistent for all generator functions $f$ and the corresponding $f$-divergences. We have not considered this general problem here.}. 
Here we show the connection between $\hat r^{(f)}$ and the Bridge estimators of $r$ with different choices of free function $\alpha(\omega)$.

\begin{proposition}[Connection between $\hat r^{(f)}$ and Bridge estimators] \label{prop: bridge_and_fdiv}
Suppose $f(u):\mathbb{R}^+ \to \mathbb{R}$ is strictly convex, twice differentiable and satisfies $f(1)=0$. Let $\{\omega_{ij}\}_{j=1}^{n_i}$ be samples from $q_i$ for $i=1,2$. If $\hat r^{(f)} = \argmax_{\tilde r \in \mathbb{R}^+} \hat G_f(\tilde r; \{\omega_{ij}\}_{j=1}^{n_i})$ is a stationary point of $\hat G_f(\tilde r; \{\omega_{ij}\}_{j=1}^{n_i})$ in \eqref{scalarobj}, then $\hat r^{(f)}$ satisfies the following equation
\begin{equation} \label{generalfixedpoint}
    \hat r^{(f)} = \frac{\frac{1}{n_2} \sum_{j=1}^{n_2} f''\left(\frac{\tilde q_1(\omega_{2j})}{\tilde q_2(\omega_{2j}) \hat r^{(f)}}\right)\frac{\tilde q_1(\omega_{2j})}{\tilde q_2(\omega_{2j})^2}\tilde q_1(\omega_{2j})}{\frac{1}{n_1} \sum_{j=1}^{n_1}  f''\left(\frac{\tilde q_1(\omega_{1j})}{\tilde q_2(\omega_{1j}) \hat r^{(f)}}\right) \frac{\tilde q_1(\omega_{1j})}{\tilde q_2(\omega_{1j})^2}\tilde q_2(\omega_{1j}) }
\end{equation}
where $f''$ is the second order derivative of $f$.
\end{proposition}
\begin{proof}
See Appendix \ref{appendix:proofs}.
\end{proof}

In Equation \eqref{generalfixedpoint}, $f''\left(\frac{\tilde q_1(\omega)}{\tilde q_2(\omega) \hat r^{(f)}}\right)\frac{\tilde q_1(\omega)}{\tilde q_2(\omega)^2}$ plays the role of the free function $\alpha(\omega)$ in a Bridge estimator \eqref{bridgeestimate}. Common Bridge estimators such as the asymptotically optimal Bridge estimator $\hat r_{opt}$ and the geometric Bridge estimator can be recovered by choosing $f$ accordingly (See Appendix \ref{appendix:fdiv_and_Bridge}). \cite{kong2003theory} observe that $\hat r_{opt}$ 
can be viewed as a semi-parametric maximum likelihood estimator. Proposition \ref{prop: bridge_and_fdiv} extends this observation and show that in addition to $\hat r_{opt}$, a large class of Bridge estimators can also be viewed as maximizers of some objective functions that are related to the variational lower bound of some $f$-divergences. In the next section, we will show how to use this variational framework to minimize the first order approximation of $RE^2(\hat r^{(T)}_{opt})$ with respect to the transformed  densities. 

\section{Improving $\hat r_{opt}$ via an $f$-GAN} \label{sec: fgan}
From Sec \ref{sec: transform}, we see that one can improve $\hat r_{opt}$ and reduce its RMSE by first transforming $q_1,q_2$ appropriately, then computing $\hat r^{(T)}_{opt}$ using the transformed densities and samples. From Sec \ref{sec: fdivandbridge} we also see the first order approximation of $RE^2(\hat r_{opt})$ is a monotonic function of $H_{s_2}(q_1,q_2)$. In this section, we utilize this observation and introduce the $f$-GAN-Bridge estimator ($f$-GB) that aims to improve $\hat r^{(T)}_{opt}$ by minimizing the first order approximation of $RE^2(\hat r^{(T)}_{opt})$ with respect to the transformed densities. We show it is equivalent to minimizing $H_{s_2}(q_1^{(T)},q_2)$ with respect to $q_1^{(T)}$ using the variational lower bound of $H_{\pi}(q_1,q_2)$ \eqref{harmoniclowerbound2} and $f$-GAN \citep{nowozin2016f}.

\subsection{The $f$-GAN framework}

We start by introducing the GAN and $f$-GAN models.
A Generative Adversarial Network (GAN) \citep{goodfellow2014generative} is an expressive class of generative models. Let $p_{tar}$ be the target distribution of interest. In the original GAN, \cite{goodfellow2014generative} estimate a generative model $p_\phi$ parameterized by a real vector $\phi$ by approximately minimizing the Jensen-Shannon divergence between $p_\phi$ and $p_{tar}$. The key idea of the original GAN is to introduce a separate discriminator which tries to distinguish between ``true samples" from $p_{tar}$ and artificially generated samples from $p_\phi$. This discriminator is then optimized alongside with the generative model $p_\phi$ in the training process. See \cite{creswell2018generative} for an overview of GAN models.

$f$-GAN \citep{nowozin2016f} extends the original GAN model using the variational lower bound of $f$-divergence \eqref{fdivlowerbound}, and introduces a GAN-type framework that generalizes to minimizing any $f$-divergence between $p_{tar}$ and $p_\phi$.  Let an $f$-divergence with the generator $f$ be given. \cite{nowozin2016f} parameterize the variational function $V_\xi$ and the generative model $p_\phi$ as two neural nets with parameters $\xi$ and $\phi$ respectively, and propose 
\begin{equation}
    G(\phi, \xi) = E_{p_{tar}}(V_\xi(\omega)) - E_{p_\phi}(f^*(V_\xi(\omega))) \label{fganobj}
\end{equation}
as the objective function of the $f$-GAN model, where $f^*$ is the convex conjugate of the generator $f$ of the chosen $f$-divergence. Recall that $G(\phi, \xi)$ is in the form of the variational lower bound \eqref{fdivlowerbound} of $D_f(p_\phi,p_{tar})$. \cite{nowozin2016f} show that $D_f(p_\phi,p_{tar})$ can be minimized by solving $\min_\phi \max_\xi G(\phi, \xi)$. Intuitively, we can view $\max_{\xi}G(\phi, \xi)$ as an estimate of $D_f(p_\phi,p_{tar})$ \citep{nguyen2010estimating}. This means minimizing $\max_\xi G(\phi, \xi)$ with respect to $\phi$ can be interpreted as minimizing an estimate of $D_f(p_\phi,p_{tar})$.

Now we show how to use the $f$-GAN framework to construct the $f$-GAN-Bridge estimator ($f$-GB). Suppose $q_1,q_2$ are defined on a common support $\Omega = \mathbb{R}^d$. Let $T_\phi: \Omega \rightarrow \Omega$ be a transformation parameterized by a real vector $\phi \in \mathbb{R}^l$ that aims to map $q_1$ to $q_2$. Let $q_1^{(\phi)}$ be the transformed density obtained by applying $T_\phi$ to $q_1$, and $\tilde q_1^{(\phi)}$ be the corresponding unnormalized density. We also require $\tilde q_1^{(\phi)}$ to be computationally tractable, and $\tilde q_1^{(\phi)} = q_1^{(\phi)}Z_1$, i.e. $\tilde q_1^{(\phi)}$ and $\tilde q_1$ have the same normalizing constant $Z_1$. Let $\mathcal{T} = \{T_\phi :  \phi \in \mathbb{R}^l\}$ be a collection of such transformations. Define $\hat r^{(\phi)}_{opt}$ to be the asymptotically optimal Bridge estimator of $r$ based on the unnormalized densities $\tilde q_1^{(\phi)}, \tilde q_2$ and corresponding samples $\{T_{\phi}(\omega_{1j})\}_{j=1}^{n_1}, \{\omega_{2j}\}_{j=1}^{n_2}$. Let $\pi \in (0,1)$. Define
\begin{multline} \label{harmoniclowerbound_transformed}
        \resizebox{0.98\hsize}{!}{$G(\phi, \tilde r; \pi) = 1 - \frac{1}{\pi} E_{q_1^{(\phi)}}\left(\frac{\pi\tilde q_2(\omega) \tilde r}{(1-\pi) \tilde q_1^{(\phi)}(\omega) + \pi\tilde q_2(\omega) \tilde r} \right)^2 - \frac{1}{1-\pi} E_{q_2} \left( \frac{(1-\pi) \tilde q_1^{(\phi)}(\omega)}{(1-\pi)\tilde q_1^{(\phi)}(\omega) + \pi\tilde q_2(\omega) \tilde r }\right)^2$}.
\end{multline} 
By Proposition \ref{prop: estimating_harmonic_divergence}, $G(\phi, \tilde r; \pi)$ is the variational lower bound of $H_{\pi}(q_1^{(\phi)}, q_2)$. In order to illustrate our idea, we first give an idealized Algorithm \ref{algo_idealized} to find the $f$-GAN-Bridge estimator. A practical version will be given in the next section.

\begin{algorithm}[H]
\caption{$f$-GAN-Bridge estimator (Idealized version)}
\label{algo_idealized}
\begin{algorithmic} 
\Require Samples $\{\omega_{ij}\}_{j=1}^{n_i} \sim q_i$ for $i=1,2$; Candidate transformations $ T_\phi \in \mathcal{T}$ parameterized by $\phi \in \mathbb{R}^l$.

\State Set $n=n_1+n_2$, $s_i = n_i/n$ for $i=1,2$.
\State Find $(\phi^*, \tilde r^*)$, a solution of $\min_{\phi \in \mathbb{R}^l} \max_{\tilde r\in \mathbb{R}^+} G(\phi, \tilde r; s_2)$ defined in \eqref{harmoniclowerbound_transformed}.
\State Use the iterative procedure in \eqref{iterativeoptimal} to compute the asymptotically optimal Bridge estimator $\hat r^{(\phi^*)}_{opt}$ based on $\tilde q_1^{(\phi^*)}$, $\tilde q_2$ and the samples $\{T_{\phi^*}(\omega_{1j})\}_{j=1}^{n_1}$,$\{\omega_{2j}\}_{j=1}^{n_2}$.
\State Compute $\widehat{RE}^2(\hat r^{(\phi^*)}_{opt})= (s_1s_2n)^{-1} \left((1- G(\phi^*, \tilde r^*; s_2))^{-1}-1\right)$. \\

\Return $\hat r^{(\phi^*)}_{opt}$ as the $f$-GAN-Bridge estimate of $r$, $\widehat{RE}^2(\hat r^{(\phi^*)}_{opt})$ as an estimate of $RE^2(\hat r^{(\phi^*)}_{opt})$ and $MSE(\log \hat r^{(\phi^*)}_{opt})$.
\end{algorithmic}
\end{algorithm}
Since $\tilde q_1^{(\phi)}$ and $\tilde q_1$ have the same normalizing constant by \eqref{transformeddensity}, $\hat r^{(\phi)}_{opt}$ is an asymptotically optimal Bridge estimator of $r$ for any transformation $T_\phi \in \mathcal{T}$. We show that within the given family of transformations $\mathcal{T}$, Algorithm \ref{algo_idealized} is able to find $T_{\phi^*}$ that minimizes the first order approximation of $RE^2(\hat r^{(\phi)}_{opt})$ with respect to $T_\phi \in \mathcal{T}$ under the i.i.d. assumption.

\begin{proposition}[Minimizing $RE^2(\hat r^{(\phi)}_{opt})$ using Algorithm \ref{algo_idealized}] \label{prop: minimizing_optimal_bridge_RMSE}
If $(\phi^*, \tilde r^*)$ is a solution of $\min_{\phi \in \mathbb{R}^l} \max_{\tilde r\in \mathbb{R}^+} G(\phi, \tilde r; s_2)$ defined in Algorithm \ref{algo_idealized}, then $G(\phi, \tilde r^*; s_2) = H_{s_2}(q_1^{(\phi)},q_2) $ for all $\phi \in \mathbb{R}^l$, $T_{\phi^*}$ minimizes  $H_{s_2}(q_1^{(\phi)},q_2)$ with respect to $T_\phi \in \mathcal{T}$. If the samples $\{\omega_{ij}\}_{j=1}^{n_i} \overset{i.i.d.}{\sim} q_i$ for $i=1,2$, then $T_{\phi^*}$ also minimizes $RE^2(\hat r^{(\phi)}_{opt})$ with respect to $T_\phi \in \mathcal{T}$ up to the first order. 
\end{proposition}
\begin{proof}
 See Appendix \ref{appendix:proofs}.
\end{proof}

From Proposition \ref{prop: minimizing_optimal_bridge_RMSE} we see that under the i.i.d. assumption, $T_{\phi^*}$ and the corresponding $f$-GAN-Bridge estimator $\hat r^{(\phi^*)}_{opt}$ are optimal in the sense that $\hat r^{(\phi^*)}_{opt}$ attains the minimal RMSE (up to the first order) among all possible transformations $T_\phi \in \mathcal{T}$ and their corresponding $\hat r^{(\phi)}_{opt}$. 
Since $G(\phi^*, \tilde r^*; s_2) = H_{s_2}(q_1^{(\phi^*)},q_2)$,  $\widehat{RE}^2(\hat r^{(\phi^*)}_{opt})$ in Algorithm \ref{algo_idealized} is exactly the leading term of $RE^2(\hat r_{opt}^{(\phi^*)})$ in the form of  \eqref{rmseinH}. Note that by Proposition \ref{prop: estimating_harmonic_divergence}, $\tilde r^*$ is equal the true ratio of normalizing constants $r$.
This means if we have $(\phi^*, \tilde r^*)$ in the idealized Algorithm \ref{algo_idealized}, it seems there is no need to carry out the following Bridge sampling step. However, $(\phi^*, \tilde r^*)$ is not computable in practice as $G(\phi, \tilde r; s_2)$ depends on the unknown normalizing constants $Z_1,Z_2$. Therefore $G(\phi, \tilde r; s_2)$ has to be approximated by an empirical estimate, and its corresponding optimizer w.r.t. $\tilde r$ is no longer equal to $r$. In the next section, we will give a practical implementation of Algorithm \ref{algo_idealized} and discuss the role of $\tilde r^*$ when $G(\phi, \tilde r; s_2)$ is replaced by an empirical estimate of it.

In Algorithm \ref{algo_idealized}, we use the $f$-GAN framework to minimize $H_{s_2}(q_1^{(\phi)},q_2)$ with respect to $T_\phi \in \mathcal{T}$. We can also apply this $f$-GAN framework to minimizing other choices of $f$-divergences such as KL divergence, Squared Hellinger distance and weighted Jensen-Shannon divergence. However, these choices of $f$-divergence are less efficient compared to the weighted Harmonic divergence $H_{s_2}(q_1^{(\phi)},q_2)$ if our goal is to improve the efficiency of $\hat r_{opt}^{(\phi)}$, as we can show that minimizing these choices of $f$-divergence between $q_1^{(\phi)}$ and $q_2$ can be viewed as minimizing some \emph{upper bounds} of the first order approximation of $RE^2(\hat r^{(\phi)}_{opt})$ (See Appendix \ref{appendix:other_fdiv}).

\subsection{Implementation and numerical stability} \label{sec: implementation details}

In this section, we give a practical implementation of the idealized Algorithm \ref{algo_idealized} based on an alternative objective function. We first describe the practical version of Algorithm \ref{algo_idealized} in Sec \ref{sec: practical}, then justify the choice of this alternative objective in Sec \ref{sec: choice of objective}.

\subsubsection{A practical implementation of Algorithm \ref{algo_idealized}} \label{sec: practical}
In this paper, we parameterize $q_1^{(\phi)}$ as a Normalizing flow. In particular, we parameterize $q_1^{(\phi)}$ as a Real-NVP \citep{dinh2016density} with base density $q_1$ and a smooth, invertible transformation $T_\phi$, where $T_\phi$ is parameterized by a real vector $\phi \in \mathbb{R}^l$. See Sec \ref{sec: transform} for a brief description of Real-NVP.
Given samples $\{\omega_{ij}\}_{j=1}^{n_i} \sim q_i$ for $i=1,2$, define
\begin{multline} \label{harmoniclowerbound_transformed_empirical}
        \hat G(\phi, \tilde r; \pi, \{\omega_{ij}\}_{j=1}^{n_i}) = 1 - \frac{1}{\pi n_1} \sum_{j=1}^{n_1}\left(\frac{\pi\tilde q_2(T_\phi(\omega_{1j})) \tilde r}{(1-\pi) \tilde q_1^{(\phi)}(T_\phi(\omega_{1j})) + \pi\tilde q_2(T_\phi(\omega_{1j})) \tilde r} \right)^2 \\ - \frac{1}{(1-\pi)n_2} \sum_{j=1}^{n_2} \left( \frac{(1-\pi) \tilde q_1^{(\phi)}(\omega_{2j})}{(1-\pi)\tilde q_1^{(\phi)}(\omega_{2j}) + \pi\tilde q_2(\omega_{2j}) \tilde r }\right)^2
\end{multline} 
to be the empirical estimate of $G(\phi, \tilde r; \pi)$ in \eqref{harmoniclowerbound_transformed}. Unlike Algorithm \ref{algo_idealized}, we do not aim to solve $\min_{\phi \in \mathbb{R}^l} \max_{\tilde r \in \mathbb{R}^+} \hat G(\phi, \tilde r; \pi, \{\omega_{ij}\}_{j=1}^{n_i})$ directly. Instead, we define our objective function as
\begin{multline} \label{hybridobjempirical}
        L_{\lambda_1, \lambda_2}(\phi, \tilde r; \pi, \{\omega_{ij}\}_{j=1}^{n_i}) = -\log(1-\hat G(\phi, \tilde r; \pi, \{\omega_{ij}\}_{j=1}^{n_i})) \\ - \frac{\lambda_1}{n_1} \sum_{j=1}^{n_1} \left(\log \tilde q_2 (T_\phi(\omega_{1j})) - \log \tilde q_1^{(\phi)}(T_\phi(\omega_{1j}))\right) -  \frac{\lambda_2}{n_2} \sum_{j=1}^{n_2} \log \tilde q_1^{(\phi)}(\omega_{2j}),
\end{multline}
where $\lambda_1, \lambda_2\geq0$ are two hyperparameters. We first give Algorithm \ref{algo_practical}, a practical implementation of Algorithm \ref{algo_idealized}, then justify the choice of the objective function \eqref{hybridobjempirical} in the following section. See Appendix \ref{appendix:algo_practical} for implementation details.

\begin{algorithm}[H]
\caption{$f$-GAN-Bridge estimator (Practical version)}
\label{algo_practical}
\begin{algorithmic} 
\Require Training samples $\{\omega_{ij}\}_{j=1}^{n_i}$ and estimating samples $\{\omega'_{ij}\}_{j=1}^{n'_i}$ from $q_i$ for $i=1,2$; Initial parameters $\phi_0 \in \mathbb{R}^l$, $\tilde r_0 >0$; Learning rate $\eta_\phi, \eta_{\tilde r}>0$; Tolerance level $\epsilon_1, \epsilon_2 > 0$; Hyperparameters $\lambda_1, \lambda_2 \geq 0$.
\State Transform and augment $q_1,q_2$ appropriately so that both densities are on a common support.
\State Set $t=0, n'=n'_1+n'_2, s_i = n'_i/n'$ for $i=1,2$
\While{$\left\vert L_{\lambda_1, \lambda_2}(\phi_t, \tilde r_t; s_2, \{\omega_{ij}\}_{j=1}^{n_i}) -  L_{\lambda_1, \lambda_2}(\phi_{t-1},\tilde r_{t-1}; s_2, \{\omega_{ij}\}_{j=1}^{n_i})\right\vert > \epsilon_1$ or $\left\vert\tilde r_t - \tilde r_{t-1}\right\vert>\epsilon_2$ or $t=0$}
\State Update $\phi_{t+1} = \phi_t - \eta_\phi \nabla_\phi L_{\lambda_1, \lambda_2}(\phi_t, \tilde r_t; s_2, \{\omega_{ij}\}_{j=1}^{n_i})$
\State Update $\tilde r_{t+1} = \tilde r_t + \eta_{\tilde r} \nabla_{\tilde r} L_{\lambda_1, \lambda_2}(\phi_t, \tilde r_t; s_2, \{\omega_{ij}\}_{j=1}^{n_i})$
\State Update $t=t+1$
\EndWhile
\State Use $\tilde r_t$ as the initial value of the iterative procedure in \eqref{iterativeoptimal}, compute $\hat r'^{(\phi_t)}_{opt}$ based on $\tilde q_1^{(\phi_t)}, \tilde q_2$ and the estimating samples $\{T_{\phi_t}(\omega'_{1j})\}_{j=1}^{n'_1}$,  $\{\omega'_{2j}\}_{j=1}^{n'_2}$.
\State Compute $\widehat{RE}^2(\hat r'^{(\phi_t)}_{opt}) = \max_{\tilde r \in \mathbb{R}^+} (s_1s_2n')^{-1} \left( (1 -  \hat G(\phi_t, \tilde r; s_2, \{\omega'_{ij}\}_{j=1}^{n'_i}))^{-1} -1\right)$.
\Return $\hat r'^{(\phi_t)}_{opt}$ as the $f$-GAN-Bridge estimate of $r$; $\widehat{RE}^2(\hat r'^{(\phi_t)}_{opt})$ as an estimate of $RE^2(\hat r'^{(\phi_t)}_{opt})$ and $MSE(\log \hat r'^{(\phi_t)}_{opt})$.
\end{algorithmic}
\end{algorithm}

In Algorithm \ref{algo_practical}, most of the computational cost is spent on estimating $q_1^{(\phi)}$. Since we parameterize $q_1^{(\phi)}$ as a Real-NVP in this paper, we leverage the GPU computing framework for neural networks. In particular, we implement Algorithm \ref{algo_practical} using PyTorch \citep{paszke2017automatic} and CUDA \citep{cuda}. As a result, most of the computation of Algorithm \ref{algo_practical} is parallelized and carried out on the GPU. This greatly accelerates the training process in Algorithm \ref{algo_practical}. We will further compare the computational cost of Algorithm \ref{algo_practical} to existing improvement strategies for Bridge sampling  \citep{meng2002warp,jia2020normalizing,wang2020warp} in Section \ref{sec: mixture of rings} and \ref{sec: glmm}.

\subsubsection{Choosing the objective function} \label{sec: choice of objective}
Note that the original empirical estimate $\hat G(\phi, \tilde r; s_2, \{\omega_{ij}\}_{j=1}^{n_i})$ can be extremely close to 1 when $q_1^{(\phi)}$ and $q_2$ share little overlap. In order to improve its numerical stability, we first transform $\hat G(\phi, \tilde r; s_2, \{\omega_{ij}\}_{j=1}^{n_i})$ to log scale using a monotonic function $h(x)=-\log(1-x)$, then apply the log-sum-exp trick on the transformed $-\log(1-\hat G(\phi, \tilde r; s_2, \{\omega_{ij}\}_{j=1}^{n_i}))$. Since $h(x)$ is monotonically increasing on $(-\infty,1)$, applying this transformation does not change the optimizers of $\hat G(\phi, \tilde r; s_2, \{\omega_{ij}\}_{j=1}^{n_i})$. 

In addition, GAN-type models can be difficult to train in practice \citep{arjovsky2017towards}. 
\cite{grover2018flowgan} suggest one can stabilize the adversarial training process of GAN-type models by incorporating a log likelihood term into the original objective function when the generative model $q_1^{(\phi)}$ is a Normalizing flow. Since both $\tilde q_1^{(\phi)}$ and $\tilde q_2$ are computationally tractable in our setup, we are able to extend this idea and stabilize the alternating training process by incorporating two ``likelihood" terms that are asymptotically equivalent to $\lambda_1 KL(q_1^{(\phi)}, q_2), \lambda_2 KL(q_2,q_1^{(\phi)})$ up to additive constants into the transformed $f$-GAN objective $-\log(1-\hat G(\phi, \tilde r; s_2, \{\omega_{ij}\}_{j=1}^{n_i}))$. Our proposed objective function $L_{\lambda_1, \lambda_2}(\phi, \tilde r; s_2, \{\omega_{ij}\}_{j=1}^{n_i})$ is then a weighted combination of $-\log(1-\hat G(\phi, \tilde r; s_2, \{\omega_{ij}\}_{j=1}^{n_i}))$ and the two ``likelihood" terms, where the hyper parameters $\lambda_1, \lambda_2 \geq 0$ control the contribution of the ``likelihood" terms.

Similar to Algorithm \ref{algo_idealized}, let $(\phi^*_L,\tilde r^*_L)$ be a solution of the min-max problem $\min_{\phi \in \mathbb{R}^l} \max_{\tilde r \in \mathbb{R}^+} L_{\lambda_1, \lambda_2}(\phi, \tilde r; s_2, \{\omega_{ij}\}_{j=1}^{n_i})$. Note that regardless of the choice of $\lambda_1, \lambda_2$, the scalar parameter $\tilde r$ only depends on $L_{\lambda_1, \lambda_2}(\phi, \tilde r; s_2, \{\omega_{ij}\}_{j=1}^{n_i})$ through $\hat G(\phi, \tilde r; s_2, \{\omega_{ij}\}_{j=1}^{n_i})$. Therefore by Proposition \ref{prop: bridge_and_fdiv}, if $\tilde r^*_L$ is a stationary point of $L_{\lambda_1, \lambda_2}(\phi^*_L, \tilde r; s_2, \{\omega_{ij}\}_{j=1}^{n_i})$ w.r.t. $\tilde r \in \mathbb{R}^+$, then $\tilde r^*_L$ can be viewed as a Bridge estimator of $r$ based on the transformed $\tilde q_1^{(\phi^*_L)}$ and the original $\tilde q_2$ with a specific choice of the free function $\alpha(\omega)$. However, $\tilde r^*_L$ is sub-optimal since the free function $\alpha(\omega)$ it uses is different from the optimal $\alpha_{opt}(\omega)$ in \eqref{alphaopt}. This means $\tilde r^*_L$ will have greater asymptotic error than the asymptotically optimal Bridge estimator. In addition, $\tilde r^*_L$ suffers from an adaptive bias \citep{wang2020warp}. Such bias arises from the fact that the estimated transformed density $q_1^{(\phi_t)}$ in Algorithm \ref{algo_practical} is chosen based on the training samples $\{\omega_{ij}\}_{j=1}^{n_i}$ for $i=1,2$. This means the density of the distribution of the transformed training samples $\{T_{\phi_t}(\omega_{1j})\}_{j=1}^{n_1}$ is no longer proportional to $\tilde q_1^{(\phi_t)}(T_{\phi_t}(\omega_{1j}))$ for $j=1,...,n_1$, as $\phi_t$ can be viewed as a function of  $\{\omega_{ij}\}_{j=1}^{n_i}$ (See Appendix \ref{appendix:algo_practical} for more discussions). Hence we do not use $\tilde r^*_L$ as our final estimator of $r$. Instead, once we have obtained $\tilde r^*_L$, we use it as a sensible initial value of the iterative procedure in \eqref{iterativeoptimal}, and compute the asymptotically optimal Bridge estimator $\hat r'^{(\phi^*_L)}_{opt}$ using a separate set of estimating samples $\{\omega'_{ij}\}_{j=1}^{n'_i}$, $i=1,2$. The resulting $\hat r'^{(\phi^*_L)}_{opt}$ does not suffer from the adaptive bias as the estimating samples are independent to the transformation $q_1^{(\phi_t)}$. When $n'_i = n_i$ for $i=1,2$, $\hat r'^{(\phi^*_L)}_{opt}$ is also statistically more efficient than $\tilde r^*_L$.

On the other hand, if $\phi_L^*$ is a minimizer of $L_{\lambda_1, \lambda_2}(\phi,\tilde r^*_L; s_2, \{\omega_{ij}\}_{j=1}^{n_i})$ with respect to $\phi$, then it asymptotically minimizes a mixture of $-\log(1-H_{s_2}(q_1^{(\phi)},q_2))$, $KL(q_1^{(\phi)}, q_2)$ and $KL(q_2,q_1^{(\phi)})$. Recall that as $n_1,n_2 \to \infty$, the additional log likelihood terms in \eqref{hybridobjempirical} is asymptotically equivalent to $\lambda_1 KL(q_1^{(\phi)}, q_2), \lambda_2 KL(q_2,q_1^{(\phi)})$ up to additive constants. We have demonstrated that minimizing $-\log(1-H_{s_2}(q_1^{(\phi)},q_2))$ with respect to $\phi$ is equivalent to minimizing the first order approximation of $RE^2(\hat r^{(\phi)}_{opt})$ under the i.i.d. assumption. We can also show that minimizing $KL(q_1^{(\phi)}, q_2), KL(q_2,q_1^{(\phi)})$ correspond to minimizing upper bounds of the first order approximation of $RE^2(\hat r^{(\phi)}_{opt})$ w.r.t. $\phi$ under the same assumption (See Appendix \ref{appendix:other_fdiv}). Note that when $\lambda_1, \lambda_2 \neq 0$, Proposition \ref{prop: minimizing_optimal_bridge_RMSE} no longer holds for this hybrid objective asymptotically, i.e. $T_{\phi^*_L}$ no longer asymptotically minimizes the first order approximation of $RE^2(\hat r^{(\phi)}_{opt})$ w.r.t. $T_\phi$. However, we find Algorithm \ref{algo_practical} with the hybrid objective works well in the numerical examples in Sec \ref{sec: mixture of rings}, \ref{sec: glmm} for any value of $\lambda_1, \lambda_2 \in (10^{-2},10^{-1})$. We want to keep $\lambda_1, \lambda_2$ small since we do not want the log likelihood terms to dominate $\hat G(\phi, \tilde r; s_2, \{\omega_{ij}\}_{j=1}^{n_i})$ in the hybrid objective $L_{\lambda_1, \lambda_2}(\phi, \tilde r; s_2, \{\omega_{ij}\}_{j=1}^{n_i})$. 
In addition, we would like to stress that even though the final $\phi_t$ in Algorithm \ref{algo_practical} does not asymptotically minimize the first order approximation of $RE^2(\hat r^{(\phi)}_{opt})$ w.r.t. $\phi$ when $\lambda_1,\lambda_2>0$, $\widehat{RE}^2(\hat r'^{(\phi_t)}_{opt})$ in Algorithm \ref{algo_practical} is still a consistent estimator of the first order approximation of $RE^2(\hat r'^{(\phi_t)}_{opt})$ by Proposition \ref{prop: estimating_harmonic_divergence} and the fact that $\hat r'^{(\phi_t)}_{opt}$ is the asymptotically optimal Bridge estimator based on the transformed $q_1^{(\phi_t)}$ and the original $q_2$.

\section{Example 1: Mixture of Rings} \label{sec: mixture of rings}

We first demonstrate the effectiveness of the $f$-GAN-Bridge estimator and Algorithm \ref{algo_practical} using a simulated example. 
Since this paper focuses on improving the original Bridge estimator \citep{meng1996simulating} rather than giving a new estimator of the normalizing constant or the ratio of normalizing constants, we will focus on comparing the performance of the proposed $f$-GAN-Bridge estimator to existing improvement strategies for Bridge sampling \citep{meng2002warp, wang2020warp, jia2020normalizing} in this and the following section. 
We do not include other classes of methods such as path sampling \citep{gelman1998simulating, lartillot2006computing}, nested sampling \citep{skilling2006nested}, variational approaches \citep{ranganath2014black}, etc. in the examples. Empirical study \citep{fourment202019} finds evidence that Bridge sampling was competitive with a wide range of methods, including the methods mentioned above, in the context of phylogenetics.

In this example, we set $q_1$, $q_2$ to be mixtures of ring-shaped distributions, and we would like to estimate the ratio of their normalizing constants. We choose this example because such mixture has a multi-modal structure, and its normalizing constant is available in closed form. Let $\boldsymbol{x} \in \mathbb{R}^2$. In order to define the pdf of $q_1,q_2$ for this example, we first define the pdf of a 2-d ring distribution as
\begin{equation}
    R(\boldsymbol{x}; \boldsymbol{\mu},b,\sigma) = \frac{1}{ \sqrt{2\pi^3\sigma^2}\Phi(b/\sigma)}\exp\left(-\frac{(\|\boldsymbol{x}-\boldsymbol{\mu}\|^2_2-b)^2}{2\sigma^2}\right); \quad \boldsymbol{\mu} \in \mathbb{R}^2, \ b,\sigma>0
\end{equation}
where $\Phi(\cdot)$ is the standard Normal CDF and $\boldsymbol{\mu}, b, \sigma$ controls the location, radius and thickness of the ring respectively. Let $\tilde R(\boldsymbol{x}; \boldsymbol{\mu},b,\sigma) = \exp\left(-\frac{(\|\boldsymbol{x}-\boldsymbol{\mu}\|^2_2-b)^2}{2\sigma^2}\right)$ be the corresponding unnormalized density. Let $\boldsymbol{\omega} \in \mathbb{R}^p$ where $p$ is an even integer. For $i=1,2$, let the unnormalized density $\tilde q_i$ be
\begin{multline}
        \tilde q_i(\boldsymbol{\omega}; \boldsymbol{\mu}_{i1},\boldsymbol{\mu}_{i2}, b_{i}, \sigma_{i}) = \prod_{j=1}^{p/2} \left(\frac{1}{2} \tilde R\left(\{\omega_{2j-1}, \omega_{2j}\}; \boldsymbol{\mu}_{i1},b_i,\sigma_i\right)\right. + \\ \left.\frac{1}{2} \tilde R\left(\{\omega_{2j-1}, \omega_{2j}\}; \boldsymbol{\mu}_{i2},b_i,\sigma_i\right) \right)
\end{multline}
where $\omega_j$ is the $j$th entry of $\boldsymbol{\omega}$. This means for $i=1,2$, if $\boldsymbol{\omega} \sim q_i$, then every two entries of $\boldsymbol{\omega}$ are independent and identically distributed, and follow an equally weighted mixture of 2-d ring distributions with different location parameters $\boldsymbol{\mu}_{i1}, \boldsymbol{\mu}_{i2}$ and the same radius and thickness parameter $b_i, \sigma_i$. It is straightforward to verify that $Z_i$, the normalizing constant of $\tilde q_i$ is $\left(\sqrt{2\pi^3\sigma_i^2}\Phi(b_i/\sigma_i)\right)^{p/2}$. In this example, we consider dimension $p=\{12,18,24,30,36,42,48\}$, and set $\boldsymbol{\mu}_{11} = (2,2), \boldsymbol{\mu}_{12} = (-2, -2), \boldsymbol{\mu}_{21}=(3,-3), \boldsymbol{\mu}_{22}=(-3,3)$, $b_1 =3 , b_2=6, \sigma_1=1, \sigma_2=2$.

In this example, we estimate $ \log r = \log Z_1 - \log Z_2$ using the  $f$-GAN-Bridge estimator ($f$-GB, Algorithm \ref{algo_practical}), Warp-III Bridge estimator \citep{meng2002warp}, Warp-U Bridge estimator \citep{wang2020warp} and Gaussianzed Bridge Sampling (GBS) \citep{jia2020normalizing}. We fix $N_i$, the number of samples from $q_i$, to be $2000$ for $i=1,2$, and compare the performance of these methods as we increase the dimension $p$. For each value of $p$, we run each methods 100 times. For Algorithm \ref{algo_practical}, we set $\lambda_1, \lambda_2 = 0.05$, and $\tilde q^{(\phi)}_{1}$ to be a Real-NVP with 4 coupling layers. For Warp-III and GBS, we use the recommended or default settings. For Warp-U, we adopt the cross splitting strategy suggested by the authors: We first estimate the Warp-U transformation using first half of the samples as the training set, and compute the Warp-U Bridge estimator using the second half as the estimating set. We then swap the role of the training and estimating set to compute another Warp-U Bridge estimator. The final output would then be the average of the two Warp-U Bridge estimators. This idea has also been discussed in \cite{wong2020properties}. Let $\hat r$ be a generic estimator of $r$. For each method and each value of $p$, we compute a MC estimate of the MSE of $\log \hat r$ based on the results from the repeated runs. We use it as the benchmark of performance. From Figure \ref{ring_all_methods} we see $f$-GB  outperforms all other methods for all choices of $p$. We also include a scatter plot of the first two dimensions of samples from $q_1,q_2$ and the transformed $q_1^{(\phi_t)}$ when $p=48$, where $q_1^{(\phi_t)}$ is estimated using Algorithm \ref{algo_practical} with $n_i=n'_i=N_i/2$ for $i=1,2$. We see the transformed $q_1^{(\phi_t)}$ captures the structure of $q_2$ accurately, and they share much greater overlap than the original $q_1,q_2$.

\begin{figure}[H]
    \centering
    \includegraphics[width=0.87\textwidth]{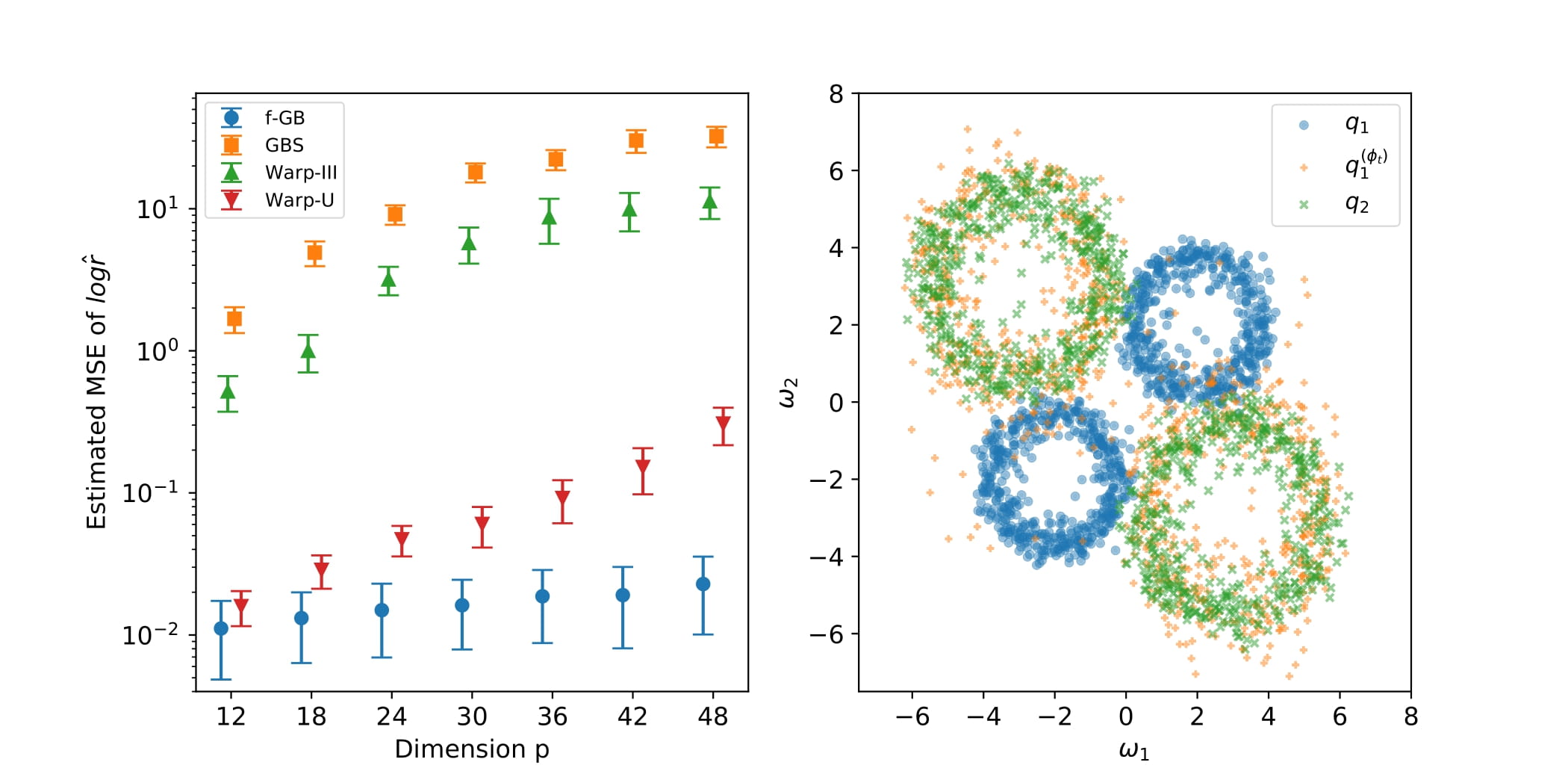}
    \caption{Left: MC estimates of MSE of $\log \hat r$ for each methods. Vertical segments are $2\sigma$ error bars. Note that the y-axis is on log scale. Right: Scatter plot of the first two dimensions of samples from $q_1,q_2$ and $q_1^{(\phi_t)}$ when $p=48$. $q_1^{(\phi_t)}$ is obtained from Algorithm \ref{algo_practical} with $n_i=n'_i=1000$ for $i=1,2$.}
    \label{ring_all_methods}
\end{figure}

We now compare the computational cost of these methods. Recall that our Algorithm \ref{algo_practical} utilizes GPU acceleration. Because of the difference in GPU and CPU computing, it is not straightforward to compare the computational cost of Algorithm \ref{algo_practical} with GBS, Warp-III and Warp-U, which are CPU based, using benchmarks such as CPU seconds or number of function calls. We simply report the averaged running time for each method on our machine in Figure \ref{ring_computational_cost}. Similar to \cite{wang2020warp}, we will also report the average ``precision per second", which is the reciprocal of the product of the running time and the estimated MSE of $\log \hat r$, for each method (higher precision per second means better efficiency). We see that for all methods, the computation time is approximately a linear function of the dimension $p$. Even though $f$-GB takes roughly twice longer to run compared to GBS and $30 \sim 40$ times longer compared to Warp-III, it achieves the highest precision per second for all dimension $p$ we consider. In addition, we also run further simulations with larger sample sizes. We find that when $p=48$, Warp-U needs around $N_1=N_2 =7500$ samples to reach a similar level of precision as $f$-GB based on $N_1=N_2=2000$ samples. In this case, Warp-U takes around $3 \sim 4$ times longer to run compared to $f$-GB. For Warp-III and GBS, we further increase the sample size to $N_1=N_2=5\times10^4$, but find that their performance is still worse than $f$-GB and Warp-U, and both take more than three times longer to run.  For Warp-III and Warp-U, it is not obvious how they would benefit from GPU computation. Although GBS may benefit from GPU acceleration in principle, it would require careful implementation and optimization. Therefore we compare our Algorithm \ref{algo_practical} to these methods based on their publicly available implementations.

\begin{figure}[H]
    \centering
    \includegraphics[width=0.87\textwidth]{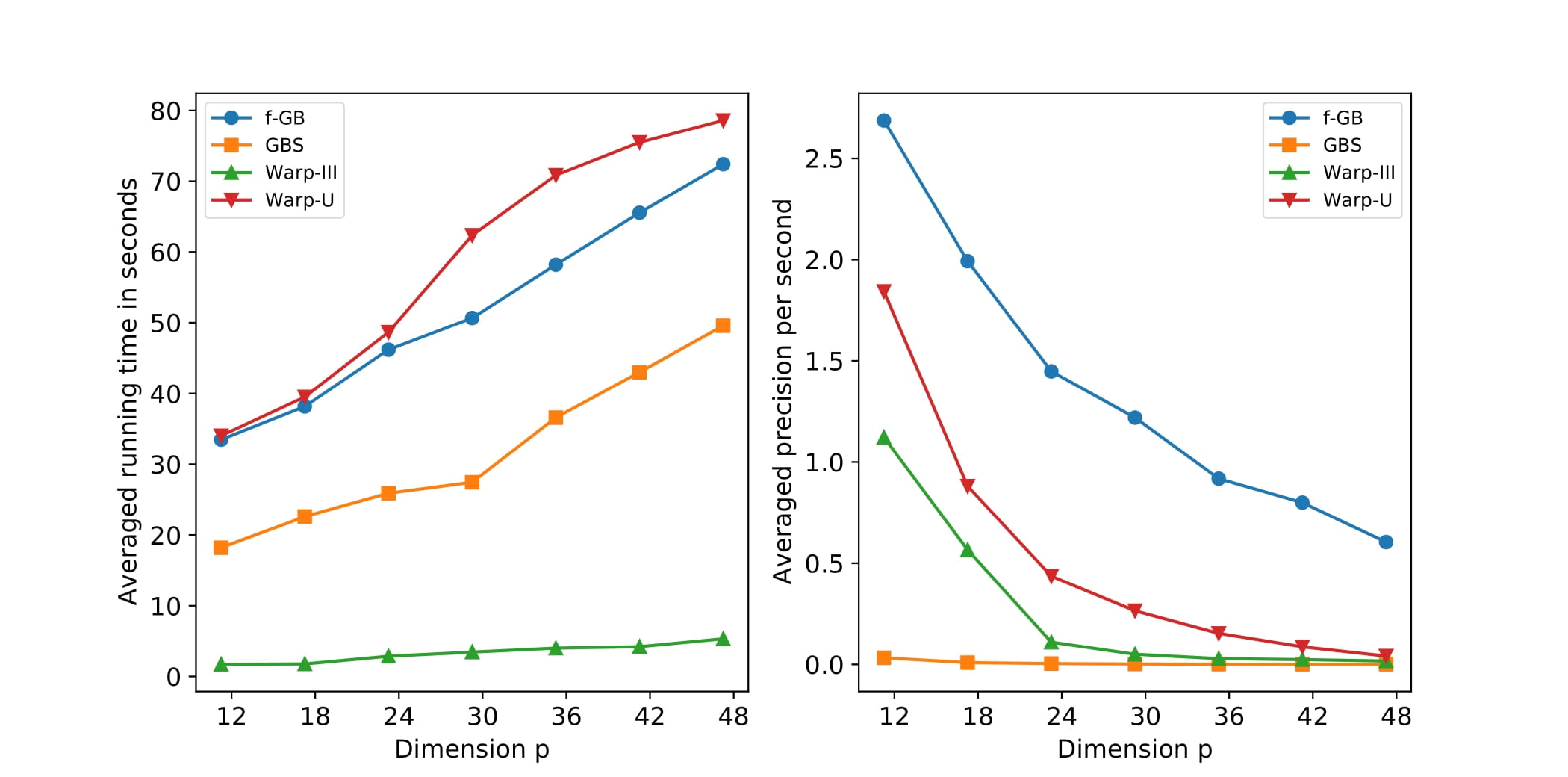}
    \caption{Left: Averaged running time for each method. Right: Averaged precision per second (i.e. reciprocal of the product of running time and the estimated MSE of $\log \hat r$) for each method.}
    \label{ring_computational_cost}
\end{figure}

Recall that $MSE(\log \hat r_{opt})$ is asymptotically equivalent to $RE^2(\hat r_{opt})$ \citep{meng1996simulating}. Therefore $\widehat{RE}^2(\hat r'^{(\phi_t)}_{opt})$ returned from Algorithm \ref{algo_practical} can also be viewed as an estimate of $MSE(\log \hat r'^{(\phi_t)}_{opt})$. In order to assess its accuracy, we compare it with both the error estimator given in \cite{fruhwirth2004estimating} and a direct MC estimator of $MSE(\log \hat r'^{(\phi_t)}_{opt})$: For each value of $p$, we first run Algorithm \ref{algo_practical} with $N_1=N_2=2000$ samples as before (i.e. we set $n_i=n'_i=1000$ for $i=1,2$). We then fix the transformed density $\tilde q^{(\phi_t)}_{1}$ obtained from Algorithm \ref{algo_practical}, repeatedly draw $n'_1 = n'_2 = 1000$ independent samples from $\tilde q^{(\phi_t)}_{1},q_2$ and record $\hat r'^{(\phi_t)}_{opt}$, $\widehat{RE}^2(\hat r'^{(\phi_t)}_{opt})$ and the error estimate given in \cite{fruhwirth2004estimating} (F-S) based on these new samples. We repeat this process 100 times, and report the box plots of $\widehat{RE}^2(\hat r'^{(\phi_t)}_{opt})$ and the error estimates given in \cite{fruhwirth2004estimating} (F-S) based on the repeated runs. We also compare the results with the direct MC estimate of $MSE(\log \hat r'^{(\phi_t)}_{opt})$ based on the repeated estimates $\log \hat r'^{(\phi_t)}_{opt}$ and the ground truth $\log r$. Note that here we fix the transformed $\tilde q^{(\phi_t)}_{1}$ and only repeat the Bridge sampling step in Algorithm \ref{algo_practical}. We summarize the results in Figure \ref{mix_of_ring_RMSE}. We see that $\widehat{RE}^2(\hat r'^{(\phi_t)}_{opt})$ returned from Algorithm \ref{algo_practical} agrees with the error estimator given in \cite{fruhwirth2004estimating} (F-S), and provides a sensible estimate of $MSE(\log \hat r'^{(\phi_t)}_{opt})$ for all choices of $p$.

\begin{figure}[H]
    \centering
    \includegraphics[width=0.87\textwidth]{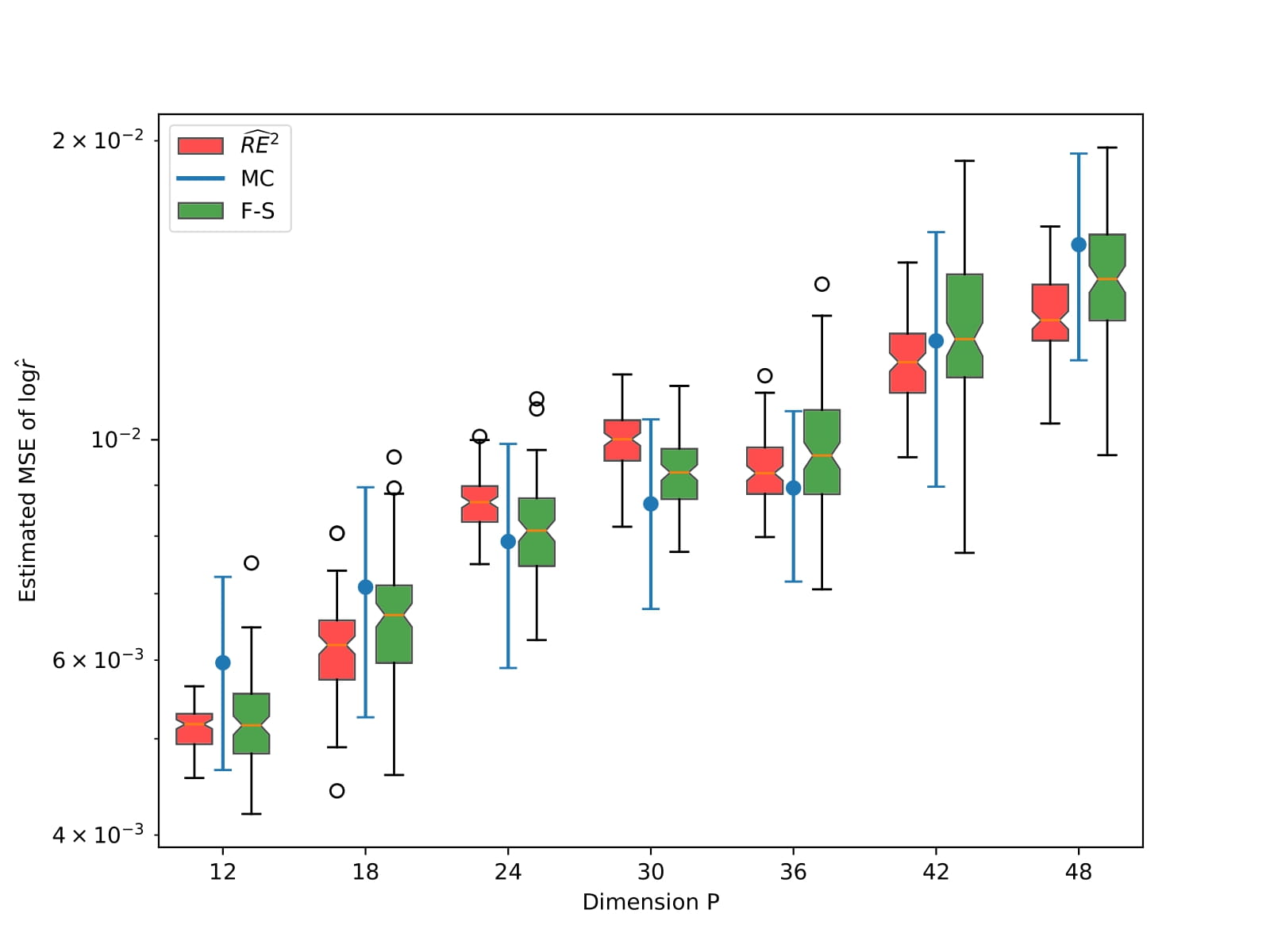}
    \caption{Box plots of 100 repetitions of $\widehat{RE}^2(\hat r'^{(\phi_t)}_{opt})$ based on Algorithm \ref{algo_practical} and the error estimator given in \cite{fruhwirth2004estimating} (F-S) for each dimension $P$. Blue vertical segments are the $2\sigma$ error bars of the corresponding MC estimates of $MSE(\log \hat r'^{(\phi_t)}_{opt})$ based on 100 repetitions.}
    \label{mix_of_ring_RMSE}
\end{figure}

\section{Example 2: Comparing two Bayesian GLMMs} \label{sec: glmm}

In this section we demonstrate the effectiveness of the $f$-GAN-Bridge estimator and Algorithm \ref{algo_practical} by considering a Bayesian model comparison problem based on the six cities dataset \citep{fitzmaurice1993likelihood}, where $q_1,q_2$ are the posterior densities of the parameters of two Bayesian GLMMs $M_1,M_2$. This example is adapted from \cite{overstall2010default}. We choose this example because it is based on real world dataset, and the posteriors $q_1,q_2$ are relatively high dimensional and are defined on disjoint supports with different dimensions.

The six cities dataset consists of the wheezing status $y_{ij}$ (1 = wheezing, 0 otherwise) of child $i$ at time $j$ for $i=1,...,n$, $n=537$ and $j=1,...4$. It also includes $x_{ij}$, the smoking status (1 = smoke, 0 otherwise) of the $i$-th child's mother at time-point $j$ as a covariate. We compare two mixed effects logistic regression models $M_1,M_2$ with different linear predictors. Define
\begin{align}
    M1 : \eta_{ij}^{(1)} &= \beta_0 + u_i; \quad u_i \stackrel{i.i.d.}{\sim} N(0,\sigma^2) \\
    M2: \eta_{ij}^{(2)} &= \beta_0 + \beta_1x_{ij} + u_i; \quad u_i\stackrel{i.i.d.}{\sim} N(0,\sigma^2)
\end{align}
where $\beta_0,\beta_1$ are regression parameters, $u_i$ is the random effect of the $i$-th child and $\sigma^2$ controls the variance of the random effects. We use the default prior given by \cite{overstall2010default} for both models, i.e. we take $\beta_0 \sim N(0, 4)$, $\sigma^{-2} \sim \Gamma(0.5,0.5)$ for $M_1$ and $(\beta_0,\beta_1)\sim N(0, 4n(\boldsymbol{X}^T \boldsymbol{X})^{-1})$, $\sigma^{-2} \sim \Gamma(0.5,0.5)$ for $M_2$ where $\boldsymbol{X} = [\boldsymbol{1}_{4n}^T, (\boldsymbol{x}_{1},...,\boldsymbol{x}_{n})^T]$, $\boldsymbol{x}_{i} = (x_{i1},...,x_{i4})$ for $i=1,...,n$. 

Let $\boldsymbol{y} = (\boldsymbol{y}_{1},... \boldsymbol{y}_{n})$ with $\boldsymbol{y}_{i} = (y_{i1},...,y_{i4})$. Let $\boldsymbol{u} = (u_1,...u_n)$ be the vector of random effects. Let $q_1(\beta_0, \boldsymbol{u}) = p( \beta_0, \boldsymbol{u} \vert \boldsymbol{X},\boldsymbol{y}, M_1)$ be the marginal posterior of $(\beta_0, \boldsymbol{u})$ under $M_1$, and $\tilde q_1(\beta_0, \boldsymbol{u})$ be the corresponding unnormalized density. Let $q_2(\beta_0,\beta_1 ,\boldsymbol{u})$, $\tilde q_2(\beta_0, \beta_1,\boldsymbol{u})$ be defined in a similar fashion under $M_2$. Samples of $q_1,q_2$ are obtained using MCMC package \texttt{R2WinBUGS} \citep{sturtz2005r2winbugs, lunn2000winbugs}. For $k=1,2$, the normalizing constant $Z_k$ of $\tilde q_k$ is the marginal likelihood under $M_k$. We first generate $2 \times 10^5$ MCMC samples from $q_1,q_2$ and estimate $\log Z_1, \log Z_2$ using the method described in \cite{overstall2010default}. The estimated log marginal likelihoods of $M_1, M_2$ based on $2 \times 10^5$ MCMC samples are $-808.139$ and $-809.818$ respectively. The results are consistent with the estimated log marginal likelihoods reported in \cite{overstall2010default} based on $5 \times 10^4$ MCMC samples. We take them as the baseline ``true values'' of $\log Z_1$ and $\log Z_2$. See \cite{overstall2010default} for R codes and technical details.

Similar to the previous example, we use $f$-GB to estimate the log Bayes factor $ \log r= \log Z_1 - \log Z_2$ between $M_1,M_2$. Note that $q_1,q_2$ are defined on disjoint support $\mathbb{R}^{n+1}, \mathbb{R}^{n+2}$ respectively. In order to apply our Algorithm \ref{algo_practical} to this problem, we first augment $q_1$ using a standard Normal to match up the difference in dimension between $q_1$ and $q_2$: Let $q_{1,aug}(\beta_0, \gamma, \boldsymbol{u}) = q_1(\beta_0, \boldsymbol{u})N(\gamma ; 0,1)$ be the augmented density where $N(\cdot ; 0,1)$ is the standard Normal pdf. Let $\tilde q_{1,aug}$ be the corresponding unnormalized augmented density. Note that $\tilde q_{1,aug}$ and $\tilde q_1$ have the same normalizing constant $Z_1$. We can then apply Algorithm \ref{algo_practical} to $q_{1,aug}$ and $q_2$ since $q_{1,aug}$ and $q_2$ are now defined on a common support $\mathbb{R}^{n+2}$. We can sample from $q_{1,aug}$ by simply concatenating a sample $(\beta_0, \boldsymbol{u}) \sim q_1$ and a sample $\gamma \sim N(0,1)$. 

Let $N_k$ be the number of MCMC samples drawn from $q_k$ for $k=1,2$. In this example, we compare the performance of the $f$-GAN-Bridge estimator with the Warp-III Bridge estimator and the Warp-U Bridge estimator as we increase the number of MCMC samples $N_1,N_2$. We consider sample size $N=\{1000,2000,3000,4000,5000\}$. This is a challenging task since the sample size $N$ is limited compared to the dimension of the problem (Recall that $q_1,q_2$ are defined on $\mathbb{R}^{n+1}, \mathbb{R}^{n+2}$ respectively with $n=537$). For each choice of $N$, we repeatedly draw $N_1=N_2=N$ MCMC samples from $q_1,q_2$ respectively and estimate the MSE of $\log \hat r$ for each method in the same way as in the previous example. For our Algorithm \ref{algo_practical}, we augment $q_1$ as described above, set $\lambda_1, \lambda_2 = 0.1$ and $q_{1,aug}^{(\phi)}$ to be a Real-NVP with 10 coupling layers. For the Warp-U and Warp-III Bridge estimator, we still use the recommended or default settings. We do not include GBS in this example since we find that for all values of $N$, it does not converge for most of the repetitions. From Figure \ref{fig: six_cities} we see our Algorithm \ref{algo_practical} outperforms the Warp-III and the Warp-U Bridge estimator for all sample size $N$. We also include a scatter plot of the first two dimensions of samples from $q_{1,aug},q_2$ and the transformed $ q^{(\phi_t)}_{1,aug}$, where $q^{(\phi_t)}_{1,aug}$ is obtained from Algorithm \ref{algo_practical} with $N=3000$. We see $q^{(\phi_t)}_{1,aug}$ an $q_2$ share much greater overlap than the original $q_{1,aug},q_2$. From Figure \ref{six_cities_computational_cost} we see for the same sample size $N$, the running time of $f$-GB is $4 \sim 6$ times as long as Warp-III, and roughly $30\% \sim 40\%$ shorter than Warp-U. On the other hand, $f$-GB achieves the highest precision per second for all sample size $N$ in this example. We further increase the sample size $N$, and find that Warp-U requires around $10^4$ MCMC samples to reach a similar level of precision achieved by $f$-GB with $N=5000$ samples, and takes around $2$ times longer to run. Similarly, Warp-III requires around $8\times10^4$ samples to get a similar level of precision, and takes around three times longer to run.

\begin{figure}[H]
    \centering
    \includegraphics[width=0.87\textwidth]{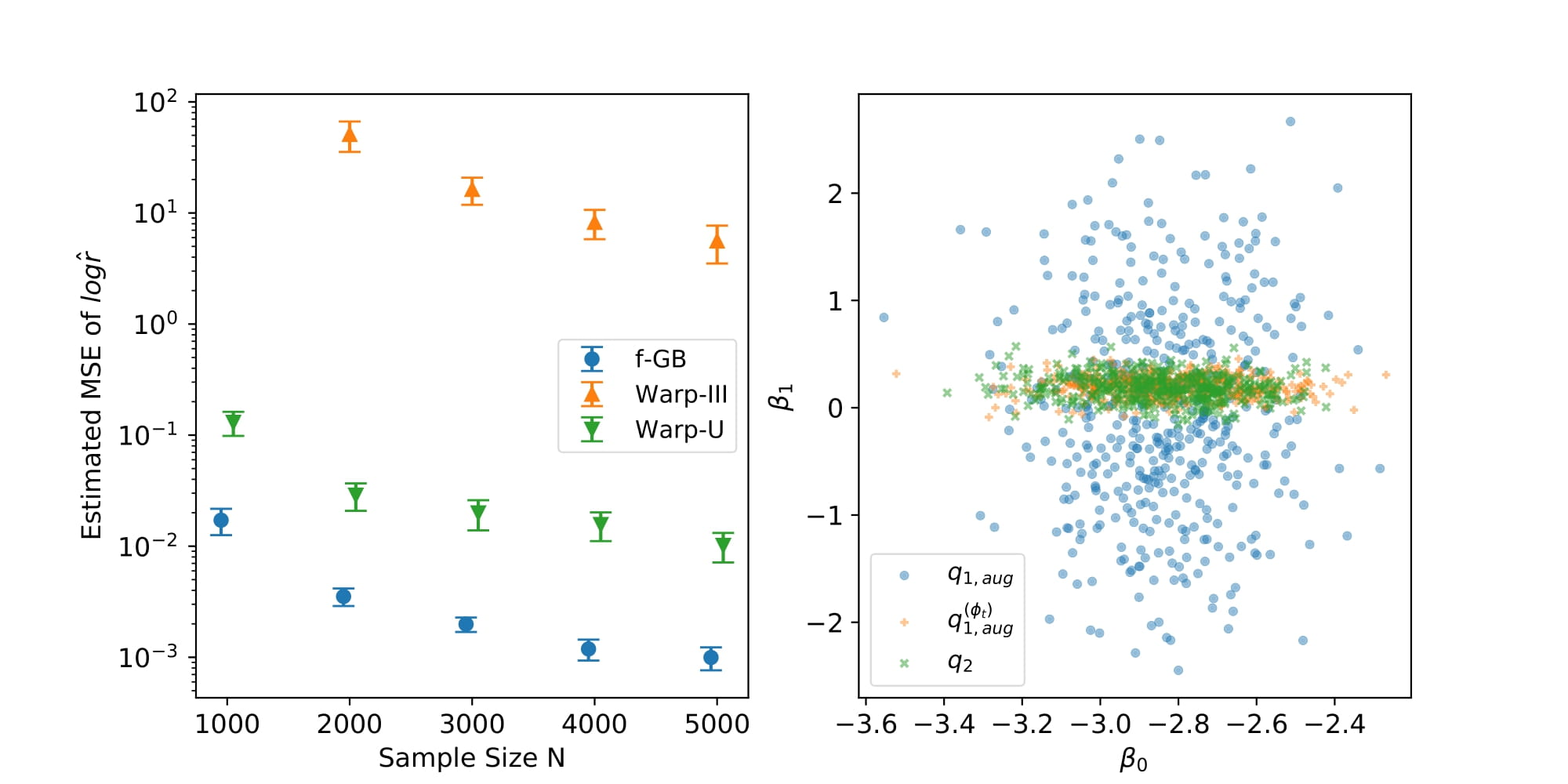}
    \caption{Left: MC estimates of MSE of $\log \hat r$ for each methods. Vertical segments are $2\sigma$ error bars. Note that the y-axis is on log scale. Warp-III does not converge for most of the repetitions when $N=1000$. Right: Scatter plot of the first two dimensions of samples from $q_{1,aug},q_2$ and $q_{1,aug}^{(\phi_t)}$, where $q_{1,aug}^{(\phi_t)}$ is obtained from Algorithm \ref{algo_practical} with $n_1=n'_i=1500$ for $i=1,2$. The first two dimensions of $q_{1,aug}$ and $q_2$ are $(\beta_0,\gamma), (\beta_0,\beta_1)$ respectively.}
    \label{fig: six_cities}
\end{figure}

\begin{figure}[H]
    \centering
    \includegraphics[width=0.87\textwidth]{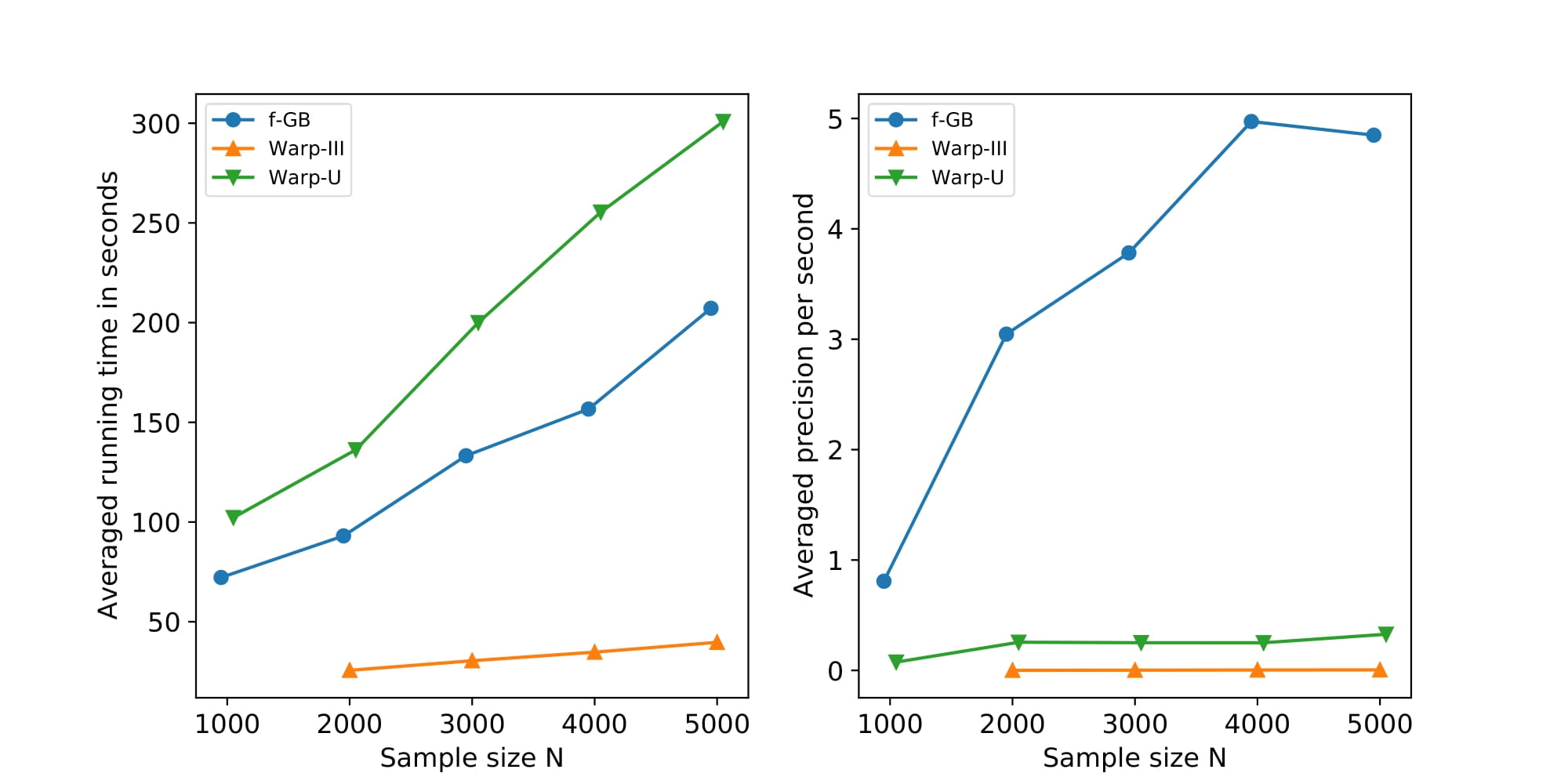}
    \caption{Left: Averaged running time for each method. Warp-III does not converge for most of the repetitions when $N=1000$. Right: Averaged precision per second (i.e. reciprocal of the product of running time and the estimated MSE of $\log \hat r$) for each method.}
    \label{six_cities_computational_cost}
\end{figure}

\begin{figure}[H]
    \centering
    \includegraphics[width=0.87\textwidth]{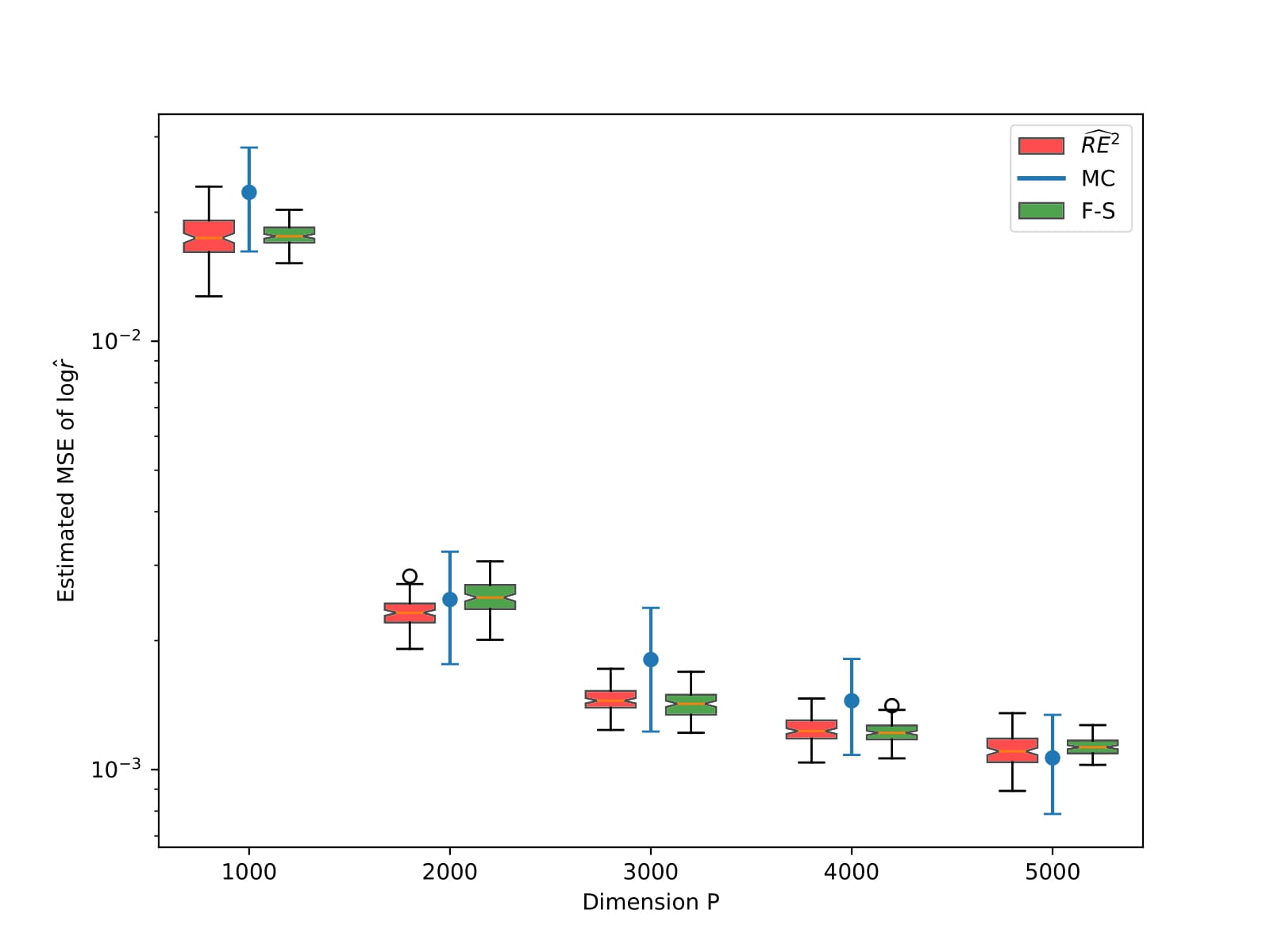}
    \caption{Box plots of 100 repetitions of $\widehat{RE}^2(\hat r'^{(\phi_t)}_{opt})$ based on Algorithm \ref{algo_practical} and the error estimator given in \cite{fruhwirth2004estimating} (F-S) for each sample size $N$. Blue vertical segments are the $2\sigma$ error bars of the corresponding MC estimates of $MSE(\log \hat r'^{(\phi_t)}_{opt})$ based on 100 repetitions.}
    \label{six_cities_RMSE}
\end{figure}

For each choice of $N$, we also  compare $\widehat{RE}^2(\hat r'^{(\phi_t)}_{opt})$ returned from Algorithm \ref{algo_practical} with the error estimator given in \cite{fruhwirth2004estimating} (F-S) and a direct MC estimator of $MSE(\log \hat r'^{(\phi_t)}_{opt})$ in the same way as in the last example. We summarize the results in Figure \ref{six_cities_RMSE}. In principle, it is not appropriate to use $\widehat{RE}^2(\hat r'^{(\phi_t)}_{opt})$ as an estimate of $MSE(\log \hat r'^{(\phi_t)}_{opt})$ in this example as the MCMC samples are correlated. However, from Figure \ref{six_cities_RMSE} we see it agrees with the erorr estimator given in \cite{fruhwirth2004estimating}, which does take autocorrelation into account, and still provides sensible estimate of $MSE(\log \hat r'^{(\phi_t)}_{opt})$ for all choices of $N$. This is likely due to the fact that the autocorrelation in our MCMC samples is weak, as we find that for all $N$, the effective sample sizes for all dimensions of the MCMC samples from $q_1,q_2$ are greater than $0.8N$. When working with weakly correlated MCMC samples, we recommend users to compute both our $\widehat{RE}^2(\hat r'^{(\phi_t)}_{opt})$ and the error estimator given in \cite{fruhwirth2004estimating}, which does take autocorrelation into account, and check if they agree with each other. When the MCMC samples are strongly correlated, we do not recommend using $\widehat{RE}^2(\hat r'^{(\phi_t)}_{opt})$ as the error estimate of $\hat r'^{(\phi_t)}_{opt}$.

\section{Conclusion} \label{sec: conclusion}
In this paper, we give a new estimator of $RE^2(\hat r_{opt})$ based on the variational lower bound of $f$-divergence proposed by \cite{nguyen2010estimating}, discuss the connection between Bridge estimators and the problem of $f$-divergence estimation, and give a computational framework to improve the optimal Bridge estimator using an $f$-GAN \citep{nowozin2016f}. We show that under the i.i.d. assumption, our $f$-GAN-Bridge estimator is optimal in the sense that it asymptotically minimizes the first order approximation of $RE^2(\hat r_{opt}^{(\phi)})$ with respect to the transformed density $q_1^{(\phi)}$. We see that in both simulated and real world examples, our $f$-GB estimator provides accurate estimate of $r$ and outperforms existing methods significantly. In addition, Algorithm \ref{algo_practical} also provides accurate estimates of $RE^2(\hat r_{opt}^{(\phi)})$ and $MSE(\log \hat r_{opt}^{(\phi)})$. In our experience, Algorithm \ref{algo_practical} ($f$-GB) is computationally more demanding than the existing methods. In the numerical examples, the running time of Algorithm \ref{algo_practical} is roughly 1 to 3 times as long as the existing methods such as Warp-U and GBS when the sample size are the same. We have not attempted to formalize the difference in computational cost because of the very different nature of GPU and CPU computing. Although in our examples, it is possible for a competing method to match the performance of the $f$-GB estimator by increasing the number of samples drawn from $q_1,q_2$, it takes longer to run, and can be inefficient or impractical when sampling from $q_1,q_2$ is computationally expensive. This also means the $f$-GB estimator is especially appealing when we only have a limited amount of samples from $q_1,q_2$. In summary, when $q_1,q_2$ are relatively simple-structured and low dimensional, the extra computational cost required by $f$-GB may not be worthwhile. However, when $q_1,q_2$ are high dimensional or have complicated multi-modal structure, we recommend the users to choose the more accurate $f$-GB estimator of $r$, given the key summary role it plays in many applications and publications.

\subsection{Limitations and future works}
One limitation of the $f$-GB estimator is the computational cost. In this paper we parameterize $q_1^{(\phi)}$ as a Normalizing flow. A possible direction of future work is to explore different choices of parameterizations of $q_1^{(\phi)}$. We expect that we can speed up our Algorithm \ref{algo_practical} by replacing a Normalizing flow by simpler transformations such as Warp-I and Warp-II transformation \citep{meng2002warp} at the expense of flexibility. Another limitation is that Algorithm \ref{algo_idealized} is only optimal when samples from $q_1,q_2$ are i.i.d. Recall that $RE^2(\hat r_{opt})$ in \eqref{rmseinH} is derived based on the i.i.d. assumption. Therefore if the samples from $q_1,q_2$ are correlated, then Proposition \ref{prop: minimizing_optimal_bridge_RMSE} no longer holds, and minimizing $H_{s_2}(q_1^{(\phi)},q_2)$ with respect to $q_1^{(\phi)}$ is no longer equivalent to minimizing the first order approximation of $RE^2(\hat r^{(\phi)}_{opt})$. Therefore it is of interest to see if it is possible to give an algorithm that minimizes the first order approximation of $RE^2(\hat r^{(\phi)}_{opt})$ when the samples are correlated. In addition, our approach only focuses on estimating the ratio of normalizing constants between two densities. When we have multiple unnormalized densities and would like to estimate the ratios between their normalizing constants, our approach needs to estimate these quantities separately in a pairwise fashion, which can be inefficient. \cite{meng1996fitting} and \citep{geyer1994estimating} show that one can estimate multiple normalizing constants simultaneously up to a common multiplicative constant. We are also interested in extending our improvement strategy to this multiple densities setup.

\newpage
\setcounter{page}{1}
\begin{center}
{\large\bf SUPPLEMENTARY MATERIAL}
\end{center}

\section{Proofs} \label{appendix:proofs}
Here we give proof of Proposition \ref{prop: estimating_harmonic_divergence}, \ref{prop: bridge_and_fdiv} and \ref{prop: minimizing_optimal_bridge_RMSE}.
\setcounter{theorem}{0}

\begin{proposition}[Estimating $H_{\pi}(q_1,q_2)$]  \label{prop: estimating_harmonic_divergence}
Let $q_1,q_2$ be continuous densities with respect to a base measure $\mu$ on the common support $\Omega$. Let $\{\omega_{ij}\}_{j=1}^{n_i}$ be samples from $q_i$ for $i=1,2$. Let $\pi \in (0,1)$ be the weight parameter. Let $r$ be the true ratio of normalizing constants between $q_1,q_2$, and $C_2 > C_1 > 0$ be constants such that $r \in [C_1, C_2]$.
For $\tilde r \in [C_1,C_2]$, define 
\begin{multline} \label{harmoniclowerbound1}
        G(\tilde r; \pi) = 1 - \frac{1}{\pi} E_{q_1}\left(\frac{\pi\tilde q_2(\omega) \tilde r}{(1-\pi) \tilde q_1(\omega) + \pi\tilde q_2(\omega) \tilde r} \right)^2 - \frac{1}{1-\pi} E_{q_2} \left( \frac{(1-\pi) \tilde q_1(\omega)}{(1-\pi)\tilde q_1(\omega) + \pi\tilde q_2(\omega) \tilde r }\right)^2.
\end{multline} 
Then $H_{\pi}(q_1,q_2)$ satisfies 
\begin{equation} 
    H_{\pi}(q_1,q_2) \geq \sup_{\tilde r \in [C_1,C_2]} G(\tilde r; \pi) ,
\end{equation}
and equality holds if and only if $\tilde r=r$. In addition, let $\hat G(\tilde r ; \pi, \{\omega_{ij}\}_{j=1}^{n_i})$ be an empirical estimate of $G(\tilde r; \pi)$ based on $\{\omega_{ij}\}_{j=1}^{n_i} \sim q_i$ for $i=1,2$. If $\hat r_{\pi} = \argmax_{\tilde r \in [C_1,C_2]} \hat G(\tilde r ; \pi, \{\omega_{ij}\}_{j=1}^{n_i})$, then $\hat r_{\pi}$ is a consistent estimator of $r$, and $\hat G(\hat r_{\pi} ; \pi, \{\omega_{ij}\}_{j=1}^{n_i})$ is a consistent estimator of $H_{\pi}(q_1,q_2)$ as $n_1,n_2 \rightarrow \infty$.
\end{proposition}

\begin{proof}
By definition, we know $0 \leq H_\pi(q_1,q_2) \leq 1$. And by setting $D_f(q_1,q_2) = H_\pi(q_1,q_2)$, $f(u) = 1-\frac{u}{\pi + (1-\pi)u}$ and variational function $V_{\tilde r}(\omega) = f'\left( \frac{\tilde q_1(\omega)}{\tilde q_2(\omega)\tilde r} \right)$ with $\mathcal{V} = \{V_{\tilde r}(\omega) \vert \tilde r \in [C_1,C_2]\}$, we see $G(\tilde r; \pi)$ exists for all $\tilde r \in [C_1,C_2]$ and is the variational lower bound of $H_\pi(q_1,q_2)$ in the form of \eqref{fdivlowerbound}. Then by \cite{nguyen2010estimating}, equality holds if and only if $V_{\tilde r}(\omega) = f'\left( \frac{\tilde q_1(\omega)}{\tilde q_2(\omega)r} \right)$. Since $f(u)$ is strictly convex, $f'(u)$ is monotonically increasing. By assumption, we also know $q_1(\omega),q_2(\omega)>0$ for all $\omega \in \Omega$. Therefore by applying the inverse of $f'$ to both side, we see $V_{\tilde r}(\omega) = f'\left( \frac{\tilde q_1(\omega)}{\tilde q_2(\omega)r} \right)$ if and only if $\tilde r = r$. Therefore $G(r; \pi) = H_\pi(q_1,q_2)$, and $\tilde r = r$ is the unique maximizer of $G(\tilde r; \pi)$.

Now we show the consistency of $\hat r_{\pi}$. It can be shown in a similar fashion to the proof of the consistency of an extremum estimator in e.g. \cite{newey1994large} Theorem 2.1.

We first check $\hat G(\tilde r ; \pi, \{\omega_{ij}\}_{j=1}^{n_i})$ satisfies the uniform law of large number (ULLN). Let $$g_1(\omega,\tilde r) =\frac{1}{\pi} \left( \frac{\pi\tilde q_2(\omega) \tilde r}{(1-\pi) \tilde q_1(\omega) + \pi\tilde q_2(\omega) \tilde r} \right)^2$$ and $$g_2(\omega,\tilde r) =\frac{1}{1-\pi} \left( \frac{(1-\pi) \tilde q_1(\omega)}{(1-\pi)\tilde q_1(\omega) + \pi\tilde q_2(\omega) \tilde r }\right)^2.$$ Since $0 < g_1(\omega,\tilde r), g_2(\omega,\tilde r) < \max(\frac{1}{\pi}, \frac{1}{1-\pi})$ for any $\omega \in \Omega$ and $\tilde r \in [C_1,C_2]$, by \cite{jennrich1969asymptotic} Theorem 2, we have $$\sup_{\tilde r \in [C_1,C_2]} \left\vert \frac{1}{n_1} \sum_{j=1}^{n_1} g_1(\omega_{1j}) - E_{q_1}g_1(\omega,\tilde r)\right\vert \rightarrow_p 0$$ and $$\sup_{\tilde r \in [C_1,C_2]} \left\vert \frac{1}{n_2} \sum_{j=1}^{n_2} g_2(\omega_{2j}) - E_{q_2}g_2(\omega,\tilde r)\right\vert \rightarrow_p 0$$ as $n_1,n_2 \rightarrow \infty$. Since $G(\tilde r; \pi) = 1 -E_{q_1}g_1(\omega,\tilde r)- E_{q_2}g_2(\omega,\tilde r)$, by triangle inequality, we have
 $$\sup_{\tilde r \in [C_1,C_2]} \left\vert \hat G(\tilde r ; \pi, \{\omega_{ij}\}_{j=1}^{n_i}) - G(\tilde r; \pi) \right\vert \rightarrow_p 0$$ as $n_1,n_2 \rightarrow \infty$. Hence $\hat G(\tilde r ; \pi, \{\omega_{ij}\}_{j=1}^{n_i})$ satisfies the uniform law of large number (ULLN).

We also need to check $G(\hat r_{\pi}; \pi) \rightarrow_p G(r; \pi)$:
\begin{flalign}
     G(r; \pi) &\geq G(\hat r_{\pi}; \pi) \quad \mbox{since $r$ is the unique maximizer of $G(\tilde r; \pi)$} \\
     & = \hat G(\hat r_{\pi} ; \pi, \{\omega_{ij}\}_{j=1}^{n_i}) + (G(\hat r_{\pi}; \pi) -  \hat G(\hat r_{\pi} ; \pi, \{\omega_{ij}\}_{j=1}^{n_i})) \\
     &\geq \hat G(r ; \pi, \{\omega_{ij}\}_{j=1}^{n_i}) + (G(\hat r_{\pi}; \pi) -  \hat G(\hat r_{\pi} ; \pi, \{\omega_{ij}\}_{j=1}^{n_i}) ) \\
     &= G(r; \pi) + (\hat G(r ; \pi, \{\omega_{ij}\}_{j=1}^{n_i}) - G(r; \pi)) + (G(\hat r_{\pi}; \pi) -  \hat G(\hat r_{\pi} ; \pi, \{\omega_{ij}\}_{j=1}^{n_i}))
\end{flalign}
Since the last two terms converge in probability to 0 by ULLN, we have $G(r; \pi) \geq G(\hat r_{\pi}; \pi) \geq G(r; \pi) + o_{p}(1)$. This implies $G(\hat r_{\pi}; \pi) \rightarrow_p G(r; \pi)$. 

Since $[C_1,C_2]$ is compact and $G(\tilde r; \pi)$ is continuous, for every open interval $A \subset [C_1,C_2]$ containing $r$, we have $\sup_{\tilde r \not\in A} G(\tilde r; \pi) < G(r; \pi)$. On the other hand, $G(\hat r_{\pi}; \pi) \rightarrow_p G(r; \pi)$ implies that $Pr(G(\hat r_{\pi}; \pi) > \sup_{\tilde r \not\in A} G(\tilde r; \pi))$ converges to 1. Therefore $Pr(\hat r_{\pi} \in A) $ also converges to 1, i.e. $\hat r_{\pi}$ is a consistent estimator of $r$.

Finally we show $\hat G(\hat r_{\pi} ; \pi, \{\omega_{ij}\}_{j=1}^{n_i})$ is a consistent estimator of $H_{\pi}(q_1,q_2)$. Recall that $G(r; \pi) = H_{\pi}(q_1,q_2)$. By triangle inequality, 
\begin{align}
  \left\vert \hat G(\hat r_{\pi} ; \pi, \{\omega_{ij}\}_{j=1}^{n_i}) - H_{\pi}(q_1,q_2) \right\vert &\leq \left\vert \hat G(\hat r_{\pi} ; \pi, \{\omega_{ij}\}_{j=1}^{n_i}) - G(\hat r_{\pi}; \pi) \right\vert + \left\vert G(\hat r_{\pi}; \pi) - G(r; \pi) \right\vert
\end{align}
The first term on the RHS converges to 0 in probability by ULLN. The second term on the RHS converges to 0 in probability by continuous mapping theorem and the fact that $\hat r_{\pi}$ is a consistent estimator of $r$. Hence $\hat G(\hat r_{\pi} ; \pi, \{\omega_{ij}\}_{j=1}^{n_i})$ is a consistent estimator of $H_{\pi}(q_1,q_2)$.
\end{proof}

\begin{proposition}[Connection between $\hat r^{(f)}$ and Bridge sampling] \label{prop: bridge_and_fdiv}
Suppose $f(u):\mathbb{R}^+ \to \mathbb{R}$ is strictly convex, twice differentiable and satisfies $f(1)=0$. Let $\{\omega_{ij}\}_{j=1}^{n_i}$ be samples from $q_i$ for $i=1,2$. If $\hat r^{(f)} = \argmax_{\tilde r \in \mathbb{R}^+} \hat G_f(\tilde r; \{\omega_{ij}\}_{j=1}^{n_i})$ is a stationary point of $\hat G_f(\tilde r; \{\omega_{ij}\}_{j=1}^{n_i})$ in \eqref{scalarobj}, then $\hat r^{(f)}$ satisfies the following equation
\begin{equation}
    \hat r^{(f)} = \frac{\frac{1}{n_2} \sum_{j=1}^{n_2} f''\left(\frac{\tilde q_1(\omega_{2j})}{\tilde q_2(\omega_{2j}) \hat r^{(f)}}\right)\frac{\tilde q_1(\omega_{2j})}{\tilde q_2(\omega_{2j})^2}\tilde q_1(\omega_{2j})}{\frac{1}{n_1} \sum_{j=1}^{n_1}  f''\left(\frac{\tilde q_1(\omega_{1j})}{\tilde q_2(\omega_{1j}) \hat r^{(f)}}\right) \frac{\tilde q_1(\omega_{1j})}{\tilde q_2(\omega_{1j})^2}\tilde q_2(\omega_{1j}) }. 
\end{equation}
where $f''$ is the second order derivative of $f$.
\end{proposition}

\begin{proof}
Note that the objective function can be written as 
\begin{align}
    \hat G_f(\tilde r; \{\omega_{ij}\}_{j=1}^{n_i}) 
    &= \frac{1}{n_1} \sum_{j=1}^{n_1} f'\left(\frac{\tilde q_1(\omega_{1j})}{\tilde q_2(\omega_{1j})\tilde r}\right) - \frac{1}{n_2} \sum_{j=1}^{n_2} f^* \circ f'\left(\frac{\tilde q_1(\omega_{2j})}{\tilde q_2(\omega_{2j})\tilde r}\right)\\ 
    &= \frac{1}{n_1} \sum_{j=1}^{n_1} f'\left(\frac{\tilde q_1(\omega_{1j})}{\tilde q_2(\omega_{1j})\tilde r}\right) - \frac{1}{n_2} \sum_{j=1}^{n_2} \frac{\tilde q_1(\omega_{2j})}{\tilde q_2(\omega_{2j})\tilde r}f'\left(\frac{\tilde q_1(\omega_{2j})}{\tilde q_2(\omega_{2j})\tilde r}\right) - f\left(\frac{\tilde q_1(\omega_{2j})}{\tilde q_2(\omega_{2j})\tilde r}\right)
\end{align}
using the equation $f^*\circ f'(u) = uf'(u)-f(u)$ \citep{uehara2016generative}.
Let $S(\tilde r; \{\omega_{ij}\}_{j=1}^{n_i}) = \frac{d}{d\tilde r} \hat G_f(\tilde r; \{\omega_{ij}\}_{j=1}^{n_i})$. If $\hat r^{(f)}$ is the stationary point, then it satisfies the ``score" equation
\begin{align}
    0&=S(\hat r^{(f)}; \{\omega_{ij}\}_{j=1}^{n_i})\\
    &= -\frac{1}{n_1} \sum_{j=1}^{n_1}  f''\left(\frac{\tilde q_1(\omega_{1j})}{\tilde q_2(\omega_{1j})\hat r^{(f)}}\right)\frac{\tilde q_1(\omega_{1j})}{\tilde q_2(\omega_{1j})(\hat r^{(f)})^2} + \frac{1}{n_2} \sum_{j=1}^{n_2} f''\left(\frac{\tilde q_1(\omega_{2j})}{\tilde q_2(\omega_{2j})\hat r^{(f)}}\right)\frac{\tilde q_1(\omega_{2j})^2}{\tilde q_2(\omega_{2j})^2(\hat r^{(f)})^3}
\end{align}
The above equation can be rearranged as
\begin{equation}
    \hat r^{(f)} = \frac{\frac{1}{n_2} \sum_{j=1}^{n_2} f''\left(\frac{\tilde q_1(\omega_{2j})}{\tilde q_2(\omega_{2j}) \hat r^{(f)}}\right)\frac{\tilde q_1(\omega_{2j})}{\tilde q_2(\omega_{2j})^2}\tilde q_1(\omega_{2j})}{\frac{1}{n_1} \sum_{j=1}^{n_1}  f''\left(\frac{\tilde q_1(\omega_{1j})}{\tilde q_2(\omega_{1j}) \hat r^{(f)}}\right) \frac{\tilde q_1(\omega_{1j})}{\tilde q_2(\omega_{1j})^2}\tilde q_2(\omega_{1j}) }
\end{equation}
\end{proof}

\begin{proposition}[Minimizing $RE^2(\hat r^{(\phi)}_{opt})$ using Algorithm \ref{algo_idealized}] \label{prop: minimizing_optimal_bridge_RMSE}
If $(\phi^*, \tilde r^*)$ is the solution of $\min_{\phi \in \mathbb{R}^l} \max_{\tilde r\in \mathbb{R}^+} G(\phi, \tilde r; s_2)$ defined in Algorithm \ref{algo_idealized}, then $G(\phi, \tilde r^*; s_2) = H_{s_2}(q_1^{(\phi)},q_2) $ for all $\phi \in \mathbb{R}^l$, $T_{\phi^*}$ minimizes  $H_{s_2}(q_1^{(\phi)},q_2)$ with respect to $T_\phi \in \mathcal{T}$. If the samples $\{\omega_{ij}\}_{j=1}^{n_i} \overset{i.i.d.}{\sim} q_i$ for $i=1,2$, then $T_{\phi^*}$ also minimizes $RE^2(\hat r^{(\phi)}_{opt})$ with respect to $T_\phi \in \mathcal{T}$ up to the first order. 
\end{proposition}
\begin{proof}
For every $\phi \in \mathbb{R}^l$, $G(\phi, \tilde r; s_2)$ is the variational lower bound of $H_{s_2}(q_1^{(\phi)},q_2)$ in the form of \eqref{fdivlowerbound}. By Proposition \ref{prop: estimating_harmonic_divergence}, we know $G(\phi, \tilde r; s_2)$ is uniquely maximized at $\tilde r = r$ w.r.t $\tilde r > 0$, and $G(\phi, r; s_2) = H_{s_2}(q_1^{(\phi)},q_2)$. Since $(\phi^*, \tilde r^*)$ is the solution of $\min_{\phi \in \mathbb{R}^l} \max_{\tilde r\in \mathbb{R}^+} G(\phi, \tilde r; s_2)$, it is straightforward to verify that $\tilde r^* = r$, and $H_{s_2}(q_1^{(\phi^*)},q_2) = G(\phi^*, \tilde r^*; s_2) \leq G(\phi, \tilde r^*; s_2)$ for any $\phi \in \mathbb{R}^l$. Hence $T_{\phi^*}$ minimizes $H_{s_2}(q_1^{(\phi)},q_2)$ with respect to $T_\phi \in \mathcal{T}$.

Since the leading term of $RE^2(\hat r^{(\phi)}_{opt})$ in \eqref{rmseinH} is a monotonically increasing function of $H_{s_2}(q_1^{(\phi)},q_2)$, $T_{\phi^*}$ minimizes $H_{s_2}(q_1^{(\phi)},q_2)$ w.r.t. $T_\phi \in \mathcal{T}$ implies $T_{\phi^*}$ minimizes the leading term of $RE^2(\hat r^{(\phi)}_{opt})$ w.r.t. $T_\phi \in \mathcal{T}$ under the assumption that samples $\{\omega_{ij}\}_{j=1}^{n_i} \sim q_i$ are i.i.d. for $i=1,2$.
\end{proof}

\section{Dimension matching} \label{appendix:dimensionmatching}

The standard Bridge estimator \eqref{bridgeestimate} can not be applied directly when $\Omega_1$,$\Omega_2$ have different dimensions. This is a common and important case. For example, if we would like to compare two models $M_1,M_2$ by estimating the Bayes factor between them, the standard Bridge estimator \eqref{bridgeestimate} is not directly applicable when $M_1,M_2$ are controlled by parameters that live in  different dimensions. 

Assume $\Omega_1 = \mathbb{R}^{d_1}$, $\Omega_2 = \mathbb{R}^{d_2}$ and $d_1<d_2$. Discrete cases work similarly. \cite{chen1997estimating} resolve the problem of unequal dimensions by first augmenting the lower dimensional density $q_1(\omega_1)$ by some completely known, normalized density $p(\theta\vert\omega_1)$ where $\theta \in \mathbb{R}^{d_2-d_1}$. This ensures the augmented density 
\begin{align}
    q_1^*(\omega_1,\theta) &= \tilde q_1^*(\omega_1,\theta)/Z_1 \\
    &= \tilde q_1(\omega_1) p(\theta\vert\omega_1)/Z_1
\end{align}
matches the dimension of the $q_2$, where $\tilde q_1^*(\omega_1,\theta)$ is the unnormalized augmented density. Let $\Omega_1^*$ be the augmented support of $q_1^*$. Since the augmented density $q_1^*(\omega_1,\theta)$ and the original $q_1(\omega_1)$ have the same normalizing constant, we can then treat $r=Z_1/Z_2$ as the ratio between the normalizing constants of $q_1^*(\omega_1,\theta)$ and $q_2(\omega_2)$, and form an ``augmented" Bridge estimator $\hat r^*_\alpha$ based on the augmented densities.  \cite{chen1997estimating} also show that when the free function $\alpha(\omega) = \alpha_{opt}(\omega)$, the optimal augmenting density $p_{opt}(\theta\vert\omega_1)$ which attains the minimal $RE^2(\hat r^*_{\alpha_{opt}})$ is
\begin{equation}
    p_{opt}(\theta\vert\omega_1) = q_2(\theta\vert\omega_1) 
\end{equation}
i.e. $p_{opt}(\theta\vert\omega_1)$ is the conditional distribution of the remaining $d_2-d_1$ entries of $\omega_2 \sim q_2$ given that the first $d_1$ entries are $\omega_1$. However, $q_2(\theta\vert\omega_1)$ is difficult to evaluate or sample from in general. One way to approximate the optimal augmenting distribution $q_2(\theta\vert\omega_1)$ is to incorporate the augmented density $q_1^*(\omega_1,\theta)$ with a Normalizing flow (see Sec \ref{sec: transform}). Assume we start with an arbitrary augmenting density $p(\theta\vert\omega_1)$, e.g. standard Normal  $N(0,I_{d_2-d_1})$. Consider a Normalizing flow with base density $q_1^*$ and a smooth and invertible transformation $T^*_1: \Omega_1^* \to \Omega_2$ that aims to map the augmented $q_1^*$ to the target $q_2$. Let $(\omega_1^{(T)}, \theta^{(T)}) = T^*_1(\omega_1, \theta)$. If $q_1^{*(T)}(\omega_1^{(T)},\theta^{(T)}) \approx q_2(\omega_1^{(T)}, \theta^{(T)})$ for all $(\omega_1^{(T)},\theta^{(T)}) \in \Omega_2$, i.e. $q_1^{*(T)}$ is a good approximation to $q_2$, then for the transformed augmenting density, we expect $q_1^{*(T)}(\theta^{(T)}\vert\omega_1^{(T)}) \approx q_2(\theta^{(T)}\vert\omega_1^{(T)})$ as well. This means the transformed $q_1^{*(T)}$ automatically learns the optimal augmenting density.

\section{Bias in the estimator of $H_\pi(q_1,q_2)$ given in Proposition \ref{prop: estimating_harmonic_divergence}} \label{appendix:positive_bias}

In Proposition \ref{prop: estimating_harmonic_divergence}, the estimator of $H_\pi(q_1,q_2)$ is given in the form of the maximum of the function $\hat G(\tilde r; \pi, \{\omega_{ij}\}_{j=1}^{n_i})$ w.r.t. $\tilde r$. Let $r=Z_1/Z_2$ be the true ratio of normalizing constants. Even though $\hat G(r; \pi, \{\omega_{ij}\}_{j=1}^{n_i})$ is an unbiased estimator of $H_\pi(q_1,q_2)$, our proposed estimator $\hat G(\hat r_\pi; \pi, \{\omega_{ij}\}_{j=1}^{n_i})$ sufferes from a positive bias. Intuitively speaking, this bias is analogous to the fact that the training error of a model is an underestimate of the true error. We use a toy example to illustrate this bias. Let $x\in\mathbb{R}^3$, $\sigma_1=1$ and $\sigma_2=3$. Let
\begin{equation}
    \tilde q_1 = \exp\left(-\frac{\left\lVert x \right\rVert^2_2}{2\sigma_1^2}\right) 
\end{equation}
\begin{equation}
    \tilde q_2 = \exp\left(-\frac{\left\lVert x \right\rVert^2_2}{2\sigma_2^2}\right).
\end{equation}

In other words, $\tilde q_1, \tilde q_2$ are the unnormalized pdf of two Gaussian distributions with zero mean and covaraince $\sigma_1I_3, \sigma_2I_3$ respectively, where $I_p$ is the $p\times p$ identity matrix. Let $q_1,q_2$ be the corresponding normalized densities. Let $\pi=0.5$. It is straightforward to form an unbiased MC estimate of $H_\pi(q_1,q_2)$ using \eqref{harmonic_def}. Let $N=\{10,20,30, ...,1000\}$. For each value of $N$, we repeatedly compute the proposed estimator $\hat G(\hat r_\pi; \pi, \{\omega_{ij}\}_{j=1}^{n_i})$ 1000 times based on $n_1=n_2=N$ i.i.d. samples from $q_1,q_2$ respectively. We then report the sample mean of the repeated estimates for each $N$, and compare it with a high precision unbiased MC estimator of $H_\pi(q_1,q_2)$.  From Figure \ref{fdiv_bias} we see $\hat G(\hat r_\pi; \pi, \{\omega_{ij}\}_{j=1}^{n_i})$ does exhibit a positive bias when $N < 500$, and this bias gradually vanishes as sample size increases.

\begin{figure}
    \centering
    \includegraphics[width=\textwidth]{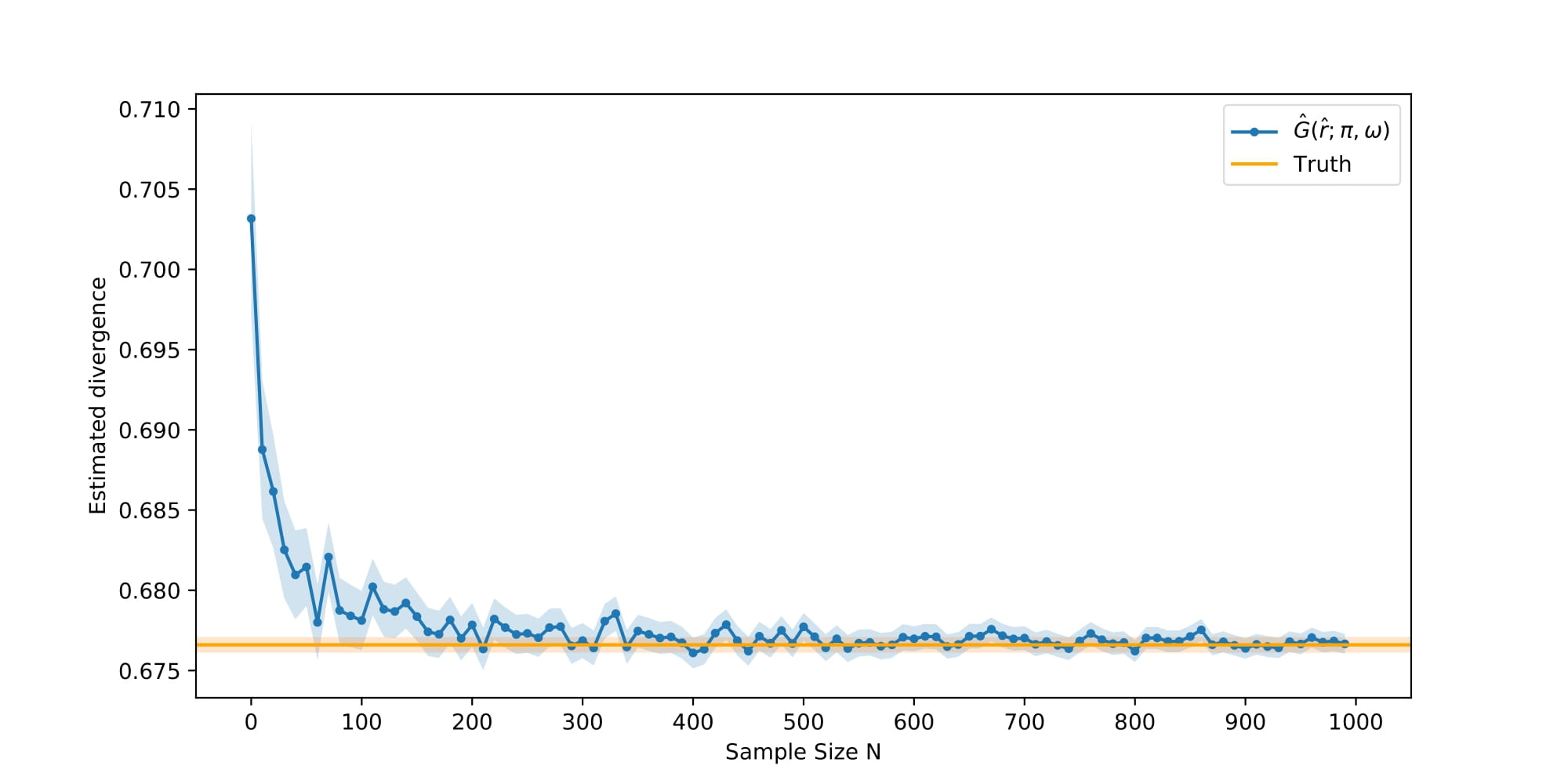}
    \caption{Sample mean of the estimated $H_\pi(q_1,q_2)$ for each sample size $N$. The blue band represents the $2\sigma$ error bars of the sample means. Orange line represents a high precision unbiased MC estimator of $H_\pi(q_1,q_2)$. Orange band represents the $2\sigma$ error bar of the MC estimate.}
    \label{fdiv_bias}
\end{figure}

Even though we have not found a practical strategy to correct this bias, we believe this bias does not prevent our proposed estimator from being useful in practice. Since our estimator of $RE^2(\hat r_{opt})$ in \eqref{minimalrmseestimate} is a monotonically increasing function of $\hat G(\hat r_\pi; \pi, \{\omega_{ij}\}_{j=1}^{n_i})$, the positive bias in $\hat G(\hat r_\pi; \pi, \{\omega_{ij}\}_{j=1}^{n_i})$ leads to a positive bias in $\widehat{RE}^2(\hat r_{opt})$. Therefore $\widehat{RE}^2(\hat r_{opt})$ will systemically overestimate the true error $RE^2(\hat r_{opt})$, which will lead to more conservative conclusions (e.g. wider error bars). This is certainly not ideal, but we believe in practice, it is less harmful than underestimating the variability in $\hat r_{opt}$. In addition, we see that our proposed error estimator provides accurate estimates of the MSE of $\log \hat r'^{(\phi^t)}_{opt}$ in both examples in Sec \ref{sec: mixture of rings} and \ref{sec: glmm}, indicating the effectiveness of it.

\section{$f$-divergence and Bridge estimators} \label{appendix:fdiv_and_Bridge}
Here we give some examples of Proposition \ref{prop: bridge_and_fdiv}. We demonstrate how the Bridge estimators with different choices of free function $\alpha(\omega)$ arise from estimating different $f$-divergences. 

\begin{example}[KL divergence and the Importance sampling estimator]
KL divergence 
\begin{equation}
    KL(q_1,q_2) = \int_\Omega \log\left(\frac{q_1(\omega)}{q_2(\omega)}\right)q_1(\omega)d\mu(\omega)
\end{equation}
is an $f$-divergence with $f(u) = u\log u$, $f'(u) = 1+\log u$ and $f^*(t) = \exp(t-1)$. This specification corresponds to $V_{\tilde r}(\omega) = 1+\log \frac{\tilde q_1(x)}{\tilde q_2(x) \tilde r}$. Suppose we have $\{\omega_{1j}\}_{j=1}^{n_1} \sim q_1$ and $\{\omega_{2j}\}_{j=1}^{n_2} \sim q_2$. The maximizer $\hat r_{KL}$ of equation \eqref{fdiv_empirical_general} under this specification is 

\begin{align}
    \hat r_{KL} &= \argmax_{\tilde r \in \mathbb{R}^+} \frac{1}{n_1} \sum_{j=1}^{n_1} \left( 1+\log \frac{\tilde q_1(\omega_{1j})}{\tilde q_2(\omega_{1j}) \tilde r} \right) - \frac{1}{n_2} \sum_{j=1}^{n_2} \frac{\tilde q_1(\omega_{2j})}{\tilde q_2(\omega_{2j}) \tilde r} \\
     &= \frac{1}{n_2} \sum_{j=1}^{n_2} \frac{\tilde q_1(\omega_{2j})}{\tilde q_2(\omega_{2j})}
\end{align}
Note that this is the Importance sampling estimator of $r$ using $q_2$ as the proposal, which is a special case of a Bridge estimator with free function $\alpha(\omega) = \tilde q_2(\omega)^{-1}$. Therefore we recover the Importance sampling estimator from the problem of estimating $KL(q_1,q_2)$. It is also straightforward to verify that estimating $KL(q_2,q_1)$ leads to $\hat r_{KL} = \left( \frac{1}{n_1} \sum_{j=1}^{n_1} \frac{\tilde q_2(\omega_{1j})}{\tilde q_1(\omega_{1j})} \right)^{-1}$, the Reciprocal Importance sampling estimator of $r$ based on a similar argument.
\end{example}

\begin{example}[Weighted Jensen-Shannon divergence and the optimal Bridge estimator]
 Weighted Jensen-Shannon divergence is defined as 
 \begin{equation}
     JS_{\pi} (q_1,q_2) = \pi KL(q_1,q_{\pi}) + (1-\pi) KL(q_2,q_{\pi})
 \end{equation}
where $\pi \in (0,1)$ is the weight parameter and $q_{\pi} = \pi q_1 + (1-\pi)q_2$ is a mixture of $q_1$ and $q_2$. Weighted Jensen-Shannon divergence is an $f$-divergence with $f(u) =  \pi u\log u-(1-\pi+\pi u)\log (1-\pi+\pi u)$, $f'(u) = \pi \log \frac{u}{1-\pi+\pi u}$ and $f^*(t) = (1-\pi) \log \frac{1-\pi}{1-\pi \exp(t/\pi)}$. This corresponds to $V_{\tilde r}(\omega) = \pi \log \frac{\tilde q_1(\omega)}{\pi \tilde q_1(\omega) + (1-\pi)\tilde q_2(\omega) \tilde r}$. Suppose we have $\{\omega_{1j}\}_{j=1}^{n_1}\sim q_1$ and $\{\omega_{2j}\}_{j=1}^{n_2} \sim q_2$. Let the weight $\pi = \frac{n_1}{n_1+n_2} = s_1$, then under this specification, the maximizer $\hat r_{JS}$ of Equation \eqref{fdiv_empirical_general} is defined as
 \begin{multline}
     \hat r_{JS} = \argmax_{\tilde r \in \mathbb{R}^+} \frac{\pi}{n_1} \sum_{j=1}^{n_1} \log \frac{\tilde q_1(\omega_{1j})}{\pi \tilde q_1(\omega_{1j}) + (1-\pi)\tilde q_2(\omega_{1j}) \tilde r} \\+ \frac{1-\pi}{n_2} \sum_{j=1}^{n_2} \log \frac{\tilde q_2(\omega_{2j})\tilde r}{\pi \tilde q_1(\omega_{2j}) + (1-\pi)\tilde q_2(\omega_{2j}) \tilde r} \label{JSobj}
 \end{multline}
 
It is straightforward to verify that $\hat r_{JS}$ satisfies
\begin{equation}
    \hat r_{JS} = \frac{\frac{1}{n_2}\sum_{j=1}^{n_2}\frac{\pi\tilde q_1(\omega_{2j})}{\pi \tilde q_1(\omega_{2j}) + (1-\pi)\tilde q_2(\omega_{2j})\hat r_{JS}}}{\frac{1}{n_1}\sum_{j=1}^{n_1}\frac{(1-\pi)\tilde q_2(\omega_{1j})}{\pi \tilde q_1(\omega_{1j}) + (1-\pi)\tilde q_2(\omega_{1j}) \hat r_{JS}}} \label{JS_r_estimate}
\end{equation}

On the other hand, recall that the asymptotically optimal Bridge estimator $\hat r_{opt}$ must be a fixed point of the iterative procedure  \eqref{iterativeoptimal}. Therefore $\hat r_{opt}$ satisfies the following ``score equation"  \citep{meng1996simulating}
\begin{align}
S(\hat r_{opt}) 
    &= -\sum_{j=1}^{n_1}\frac{s_2\tilde q_2(\omega_{1j})\hat r_{opt}}{s_1 \tilde q_1(\omega_{1j}) + s_2\tilde q_2(\omega_{1j}) \hat r_{opt}} + \sum_{j=1}^{n_2}\frac{s_1\tilde q_1(\omega_{2j})}{s_1 \tilde q_1(\omega_{2j}) + s_2\tilde q_2(\omega_{2j})\hat r_{opt}} \label{bridgescore} \\
    &= 0
\end{align}

When $\pi = s_1$, Equation \eqref{JS_r_estimate} is precisely the score equation \eqref{bridgescore} of $\hat r_{opt}$. This implies $\hat r_{JS}=\hat r_{opt}$ because the root of the score function $S(r)$ in \eqref{bridgescore} is unique \citep{meng1996simulating}. Therefore $\hat r_{JS}$ is equivalent to the asymptotically optimal Bridge estimator $\hat r_{opt}$, and we recover $\hat r_{opt}$ from the problem of estimating the weighted Jensen-Shannon divergence between $q_1,q_2$.
\end{example}

\begin{example}[Squared Hellinger distance and the geometric Bridge estimator]
 Squared Hellinger distance
 \begin{equation}
     H^2(q_1,q_2) = \int_\Omega \left(\sqrt{q_1(\omega)} - \sqrt{q_2(\omega)}\right)^2 d\mu(\omega)
 \end{equation}
is an $f$-divergence with $f(u) = (\sqrt{u}-1)^2$, $f'(u)=1-u^{-\frac{1}{2}}$ and $f^*(t) = \frac{t}{1-t}$. This specification corresponds to $V_{\tilde r}(\omega) = 1 - \sqrt{\frac{\tilde q_2(\omega)\tilde r}{\tilde q_1(\omega)}}$. Again suppose we have $\{\omega_{1j}\}_{j=1}^{n_1} \sim q_1$ and $\{\omega_{2j}\}_{j=1}^{n_2} \sim q_2$. The maximizer $\hat r_{H^2}$ of equation \eqref{fdiv_empirical_general} under this specification is 

\begin{align}
    \hat r_{H^2} &= \argmax_{\tilde r \in \mathbb{R}^+} \frac{1}{n_1} \sum_{j=1}^{n_1} \left( 1 - \sqrt{\frac{\tilde q_2(\omega_{1j})\tilde r}{\tilde q_1(\omega_{1j})}} \right) - \frac{1}{n_2} \sum_{j=1}^{n_2} \left( \sqrt{\frac{\tilde q_1(\omega_{2j})}{\tilde q_2(\omega_{2j}) \tilde r}} - 1\right) \\
    &= \frac{\frac{1}{n_2} \sum_{j=1}^{n_2}\sqrt{\tilde q_1(\omega_{2j})/\tilde q_2(\omega_{2j})}}{\frac{1}{n_1} \sum_{j=1}^{n_1}\sqrt{\tilde q_2(\omega_{1j})/\tilde q_1(\omega_{1j})}}
\end{align}
This is precisely the geometric Bridge estimator $\hat r_{geo}$ in \cite{meng1996simulating} with free function $\alpha(\omega) = (\tilde q_1(\omega) \tilde q_2(\omega))^{-\frac{1}{2}}$. 
\end{example}

In addition to the fact that Bridge estimators with different choices of free function $\alpha(\omega)$ can  arise from estimating different $f$-divergences, the asymptotic RMSE of $\hat r_{KL}, \hat r_{opt}$ and $\hat r_{geo}$ can also be written as functions of some $f$-divergences between the two distributions. For example, \cite{meng1996simulating} show that $RE^2(\hat r_{geo})$ is a function of the Hellinger distance between $q_1,q_2$, \cite{wang2020warp} show that $RE^2(\hat r_{opt})$ is a function of $H_\pi(q_1,q_2)$ in \eqref{rmseinH}. It is also straightforward to show $RE^2(\hat r_{KL})$ is a function of the R\'enyi's 2-divergence between $q_1,q_2$ using the formula of $RE^2(\hat r_{\alpha})$ given by (3.2) in \cite{meng1996simulating}. However, the general connection between $RE^2(\hat r_{\alpha})$ and the $f$-divergence between the two distributions is not obvious. For example, suppose we choose the constant free function $\alpha(\omega)= 1$ discussed in \cite{meng1996simulating}. Then we can work out the asymptotic RMSE of the corresponding Bridge estimator $\hat r_1$ using the formula of $RE^2(\hat r_{\alpha})$ in \cite{meng1996simulating}. Suppose $q_1,q_2$ are defined on a common support $\Omega$, the resulting $RE^2(\hat r_1)$ takes the form 
\begin{equation}
    RE^2(\hat r_1) = (s_1 s_2 n)^{-1} \frac{\int_\Omega q_1(\omega)q_2(\omega)(s_1 q_1(\omega) + s_2 q_2(\omega)) d\omega}{\left(\int_\Omega q_1(\omega)q_2(\omega) d\omega \right)^2} + o\left(\frac{1}{n}\right)
\end{equation}
It is not obvious how this expression can be rearranged into a function of some $f$-divergence between $q_1,q_2$, as the leading term of $RE^2(\hat r_1)$ is in the form of ratio of integrals, which is different from the general functional form of an $f$-divergence. This example suggests that there may not be a general connection between the $f$-divergence between two distributions and the RMSE of a Bridge estimator apart from common Bridge estimators such as the optimal Bridge estimator and the geometric Bridge estimator.  We have also tried the other direction. We started from an $f$-divergence. By Proposition \ref{prop: bridge_and_fdiv}, estimating the $f$-divergence leads to a Bridge estimator with a specific free function in the form of $\alpha_f(\omega) = f''\left(\frac{\tilde q_1(\omega)}{\tilde q_2(\omega) \hat r^{(f)}}\right)\frac{\tilde q_1(\omega)}{\tilde q_2(\omega)^2}$. We then substitute this $\alpha_f(\omega)$ into the formula of $RE^2(\hat r_{\alpha})$ in \cite{meng1996simulating}. The functional form of the resulting expression is still also very different from the functional form of an $f$-divergence in the general case, and it is not obvious to see how it can be rearranged into a function of some $f$-divergence between the two distributions. This also suggests that there may not be a general connection between $RE^2(\hat r_{\alpha})$ and the $f$-divergence between two distributions. 

\section{Other choices of $f$-divergence} \label{appendix:other_fdiv}

The weighted Harmonic divergence $H_{\pi}(q_1^{(\phi)},q_2)$ is not the only choice of $f$ divergence to minimize if our goal is to increase the overlap between $q_1^{(\phi)}$ and $q_2$. Recall that in Algorithm \ref{algo_practical} we parameterize $q_1^{(\phi)}$ as a Normalizing flow. Since both $\tilde q_1,\tilde q_2$ are available, it is also possible to estimate $q_1^{(\phi)}$ by maximizing the log likelihood $\log \tilde q_2(T_{\phi}(\omega_{1j}))$ or $\log \tilde q_1^{(\phi)}(\omega_{2j})$ without using the $f$-GAN framework. This is asymptotically equivalent to approximating $q_2$ using $q_1^{(\phi)}$ by minimizing $KL(q_1^{(\phi)},q_2)$ or $KL(q_2,q_1^{(\phi)})$. In addition to the KL divergence, other common $f$-divergences such as the Squared Hellinger distance and the weighted Jensen-Shannon divergence are also sensible measures of overlap between densities, and we can minimize these divergences using the $f$-GAN framework in a similar fashion to Algorithm \ref{algo_idealized}. However, $f$-divergences such as KL divergence, Squared Hellinger distance and the weighted Jensen-Shannon divergence are inefficient compared to the weighted Harmonic divergence $H_{s_2}(q_1^{(\phi)},q_2)$ if our goal is to minimize $RE^2(\hat r_{opt}^{(\phi)})$. In Proposition \ref{prop: minimizing_optimal_bridge_RMSE} we have shown that under the i.i.d. assumption, minimizing  $H_{s_2}(q_1^{(\phi)},q_2)$ with respect to $q_1^{(\phi)}$ is equivalent to minimizing the first order approximation of $RE^2(\hat r_{opt}^{(\phi)})$ directly. On the other hand, \cite{meng1996simulating} show that asymptotically, 
\begin{equation}
    RE^2(\hat r_{opt}) \leq (ns_1s_2)^{-1} \left(\left( 1 - \frac{1}{2} H^2(q_1,q_2) \right)^{-2} - 1\right) \label{rmse_Hellingerbound}
\end{equation} 
up to the first order, where $n = n_1+n_2$ and $s_i = n_i/n$ for $i=1,2$ under the same i.i.d. assumption. Note that $H^2(q_1^{(\phi)},q_2) \rightarrow 0$ also implies $RE^2(\hat r^{(\phi)}_{opt}) \rightarrow 0$, but minimizing $H^2(q_1^{(\phi)},q_2)$ with respect to the density $q_1^{(\phi)}$ can be viewed as minimizing an \emph{upper bound} of the first order approximation of $RE^2(\hat r_{opt}^{(\phi)})$, which is less efficient.
Here we show minimizing $JS_{\pi}(q_1^{(\phi)},q_2)$, $KL(q_1^{(\phi)},q_2)$ or $KL(q_2,q_1^{(\phi)})$ with respect to $q_1^{(\phi)}$ can also be viewed as  minimizing some upper bounds of the first order approximation of $RE^2(\hat r_{opt}^{(\phi)})$.

\begin{proposition}[Upper bounds of $RE^2(\hat r_{opt}^{(\phi)})$] \label{prop: other fdiv}
Let $q_1,q_2$ be continuous densities with respect to a base measure $\mu$ on the common support $\Omega$. If $\pi \in (0,1)$ is the weight parameter, then $JS_\pi(q_1,q_2) \to 0$, $KL(q_1,q_2) \to 0$ or $KL(q_2,q_1) \to 0$ implies $RE^2(\hat r_{opt}) \rightarrow 0$, and asymptotically, 
\begin{align}
        RE^2(\hat r_{opt}) &\leq \frac{1}{s_1s_2n} \left( \left(1-\min(1, \sqrt{JS_\pi(q_1,q_2)/\min(\pi,1-\pi)})\right)^{-2} - 1\right), \label{rmse_JSbound} \\
        RE^2(\hat r_{opt}) &\leq \frac{1}{s_1s_2n} \left( \left(1-\min(1, \sqrt{2KL(q_1,q_2)})\right)^{-2} - 1\right), \label{rmse_KLbound_forward} \\
        RE^2(\hat r_{opt}) &\leq \frac{1}{s_1s_2n} \left( \left(1-\min(1, \sqrt{2KL(q_2,q_1)})\right)^{-2} - 1\right). \label{rmse_KLbound_backward}     
\end{align}
up to the first order, where $n = n_1+n_2$ and $s_i = n_i/n$ for $i=1,2$.
\end{proposition}

\begin{proof}
 Recall that $JS_{\pi} (q_1,q_2) = \pi KL(q_1,q_{\pi}) + (1-\pi) KL(q_2,q_{\pi})$ where $q_\pi = \pi q_1 + (1-\pi) q_2$ is a mixture of $q_1,q_2$. Let $d_{TV}(q_1,q_2)$ be the total variation distance between $q_1$ and $q_2$. By Pinsker's inequality we have $KL(q_i,q_\pi) \geq 2d_{TV}^2(q_i,q_\pi)$ for $i=1,2$ \citep{pinsker1964information}. Then
\begin{align}
        JS_{\pi}(q_1,q_2) &= \pi KL(q_1,q_{\pi}) + (1-\pi) KL(q_2,q_{\pi}) \\
     &\geq 2\pi d_{TV}^2(q_1,q_\pi) + 2(1-\pi)d_{TV}^2(q_2,q_\pi)\\
     &\geq 2 \min(\pi,1-\pi) (d_{TV}^2(q_1,q_\pi) + d_{TV}^2(q_2,q_\pi))\\
     &\geq 2\min(\pi,1-\pi) \left( \frac{1}{2}(d_{TV}(q_1,q_\pi) + d_{TV}(q_2,q_\pi))^2 \right)\\
     &\geq \min(\pi,1-\pi) d_{TV}^2(q_1,q_2)
\end{align}
by the algebraic-geometric mean inequality and the triangle inequality. Since  $d_{TV}(q_1,q_2)\geq \frac{1}{2}H^2(q_1,q_2)$ \citep{le1969theorie}, we have $JS_{\pi}(q_1,q_2) \geq \min(\pi,1-\pi) \left(\frac{1}{2}H^2(q_1,q_2)\right)^2$. Since both $JS_\pi(q_1,q_2)$ and $H^2(q_1,q_2)$ are non-negative, $JS_\pi(q_1,q_2) \to 0$ implies $H^2(q_1,q_2) \to 0$ and $RE^2(\hat r_{opt}) \rightarrow 0$ by \eqref{rmse_Hellingerbound}.
On the other hand, since $H^2(q_1,q_2) \leq 2$, we have
\begin{equation}
    \frac{1}{2}H^2(q_1,q_2) \leq \min(1,\sqrt{JS_{\pi}(q_1,q_2)/\min(\pi,1-\pi)}). \footnote{Since $JS_{\pi}(q_1,q_2) \leq \log2$ for all $\pi \in (0,1)$ \citep{lin1991divergence}, $\sqrt{JS_{\pi}(q_1,q_2)/\min(\pi,1-\pi)}$ does not exceed 1 by large amount when $\pi$ is close to $1/2$. For example, when $\pi=1/2$, $\sqrt{JS_{\pi}(q_1,q_2)/\min(\pi,1-\pi)} < 1.18$.}
\end{equation}
Substituting it into the right hand side of \eqref{rmse_Hellingerbound} yields \eqref{rmse_JSbound}. 

 From the last paragraph, we have $KL(q_1,q_2) \geq 2d_{TV}^2(q_1,q_2) \geq \frac{1}{2}\left(H^2(q_1,q_2)\right)^2$. Therefore $KL(q_1,q_2) \to 0$ also implies $H^2(q_1,q_2) \to 0$ and $RE^2(\hat r_{opt}) \rightarrow 0$.
 We also have $\frac{1}{2}H^2(q_1,q_2) \leq \min(1,\sqrt{2KL(q_1,q_2)})$. Substituting it into the right hand side of \eqref{rmse_Hellingerbound} yields \eqref{rmse_KLbound_forward}. We can show \eqref{rmse_KLbound_backward} using the same argument.
\end{proof}
From Proposition \ref{prop: other fdiv} we see minimizing these choices of $f$-divergences are also effective for reducing the $RE^2(\hat r^{\phi)}_{opt})$. However, these choices of $f$-divergence are inefficient compared to $H_{s_2}(q_1^{(\phi)},q_2)$ since minimizing these $f$-divergences only correspond to minimizing some upper bounds of the first order approximation of $RE^2(\hat r^{\phi)}_{opt})$, while minimizing $H_{s_2}(q_1^{(\phi)}, q_2)$ is equivalent to minimizing the first order approximation of $RE^2(\hat r^{(\phi)}_{opt})$ directly.

\section{Implementation details of Algorithm \ref{algo_practical}}\label{appendix:algo_practical}

\subsubsection*{Choosing the transformation $T_\phi$}
We parameterize $\tilde q_1^{(\phi)}$ as a Real-NVP \citep{dinh2016density} with base density $\tilde q_1$ (See Sec \ref{sec: transform} for a brief description of Normalizing flow models and Real-NVP). As we have discussed before, this ensures that $\tilde q_1^{(\phi)}$ is both flexible and computationally tractable, and its normalizing constant is unchanged. It is possible to specify $\tilde q_1^{(\phi)}$ using a simpler parameterization, e.g. Warp-III transformation \citep{meng2002warp}. However, such parameterization is not as flexible comparing to a Normalizing flow.  It is also possible to replace a Real-NVP by more sophisticated Normalizing flow architectures e.g. Autoregressive flows \citep{papamakarios2017masked} or Neural Spline flows \citep{durkan2019neural}. But we find a Real-NVP is sufficient for us to illustrate our ideas and achieve satisfactory results in both simulated and real world examples. In addition, both the froward and inverse transformation of a Real-NVP can be computed efficiently. This is an appealing feature since we need both $T_\phi$ and $T_\phi^{-1}$ for evaluating $L_{\lambda_1, \lambda_2}(\phi, \tilde r; s_2, \{\omega_{ij}\}_{j=1}^{n_i})$. Therefore we choose to use a Real-NVP in Algorithm \ref{algo_practical}, as it has a good balance of flexibility and computational efficiency.

The number of coupling layers in a Real-NVP controls its flexibility. Choosing too few coupling layers restricts the flexibility of the Real-NVP, while choosing too many of them increases the computational cost and the risk of overfitting. Choosing the optimal number of coupling layers that balances computational cost and flexibility is problem-dependent. We demonstrate it using the mixture of rings example in Sec \ref{sec: mixture of rings} with $p=12$. Similar to Sec \ref{sec: mixture of rings}, we set $\beta_1=\beta_2=0.05$ and $N_1=N_2=2000$. We consider different number of coupling layers $K=\{2,4,6,8,10,12,14,16,18\}$. For each choice of $K$, we parameterize $q_1^{(\phi)}$ in Algorithm \ref{algo_practical} as a Real-NVP with $K$ coupling layers, then run Algorithm \ref{algo_practical} 50 times. We record the average running time and an MC estimate of $MSE(\log \hat r'^{(\phi_t)}_{opt})$ based on the repeated runs for each $K$. From Figure \ref{fig:coupling} we see the running time is roughly a linear function of the number of coupling layers $K$. As $K$ increases, the estimated $MSE(\log \hat r'^{(\phi_t)}_{opt})$ first decreases then starts increasing. This is likely due to overfitting. Similar to \cite{wang2020warp}, we also use precision per second, which is the reciprocal of the product of the average running time and the estimated mean square error, as a benchmark of efficiency. We see that the estimated precision per second is high when $K$ is between 2 and 6, and it starts decreasing rapidly when $K \geq 8$. Therefore for this example, we see the most efficient choice of $K$ is between 2 and 6. In practice, we recommend setting $q_1^{(\phi)}$ as a Real-NVP with 2 to 10 coupling layers in Algorithm \ref{algo_practical}. We also recommend running Algorithm \ref{algo_practical} multiple times with different number of coupling layers in $q_1^{(\phi)}$, and choose the one that achieves the lowest estimated RMSE $\widehat{RE}^2(\hat r'^{(\phi_t)}_{opt})$.

\begin{figure}
    \centering
    \includegraphics[width=1.1\textwidth]{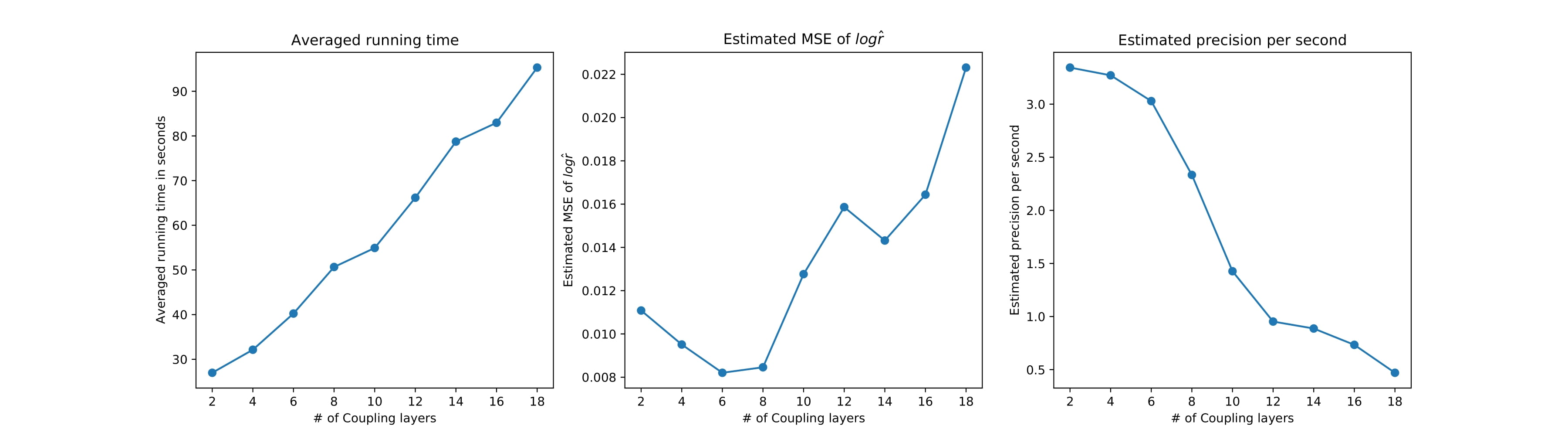}
    \caption{Left: Average running time for each $K$ based on 50 repetitions. Mid: MC estimate of $MSE(\log \hat r'^{(\phi_t)}_{opt})$ for each $K$. Right: Estimated precision per second for each $K$.}
    \label{fig:coupling}
\end{figure}

\subsubsection*{Splitting the samples from $q_1,q_2$}
In Algorithm \ref{algo_practical}, we first estimate $\{\phi_t, \tilde r_t\}$ using the training samples $\{\omega_{ij}\}_{j=1}^{n_i}$, then compute the optimal Bridge estimator based on the separate estimating samples $\{\omega'_{ij}\}_{j=1}^{n'_i}$.
We use separate samples for the Bridge sampling step because the estimated transformed density $q_1^{(\phi_t)}$ in Algorithm \ref{algo_practical} is chosen based on the training samples $\{\omega_{ij}\}_{j=1}^{n_i}$ for $i=1,2$. This means the density of the distribution of the transformed training samples $\{T_{\phi_t}(\omega_{1j})\}_{j=1}^{n_1}$ is no longer proportional to $\tilde q_1^{(\phi_t)}(T_{\phi_t}(\omega_{1j}))$ for $j=1,...,n_1$ as $\phi_t$ can be viewed as a function of  $\{\omega_{ij}\}_{j=1}^{n_i}$.  If we apply the iterative procedure \eqref{iterativeoptimal} to densities $q_1^{(\phi_t)}, q_2$ and the transformed training samples $\{T_{\phi_t}(\omega_{1j})\}_{j=1}^{n_1}, \{\omega_{2j}\}_{j=1}^{n_2}$, then the resulting $\hat r_{opt}^{(\phi_t)}$ will be a biased estimate of $r$. See also \cite{wong2020properties} for a detailed discussion under a similar setting. One way to correct this bias is to split the samples from $q_1,q_2$ into training samples $\{\omega_{ij}\}_{j=1}^{n_i}$ and estimating samples $\{\omega'_{ij}\}_{j=1}^{n'_i}$ for $i=1,2$. We first estimate the transformation $T_{\phi_t}$ using the training samples $\{\omega_{ij}\}_{j=1}^{n_i}$, $i=1,2$. Once we have obtained the estimated $\phi_t$, we apply the iterative procedure \eqref{iterativeoptimal} to $\tilde q_1^{(\phi_t)}, \tilde q_2$ and the transformed estimating samples $\{T_{\phi_t}(\omega'_{1j})\}_{j=1}^{n'_1}, \{\omega'_{2j}\}_{j=1}^{n'_2}$, $i=1,2$. Then the resulting estimate $\hat r'^{(\phi_t)}_{opt}$ will not suffer from this bias. The same approach is used in \cite{wang2020warp} and \cite{jia2020normalizing}. The idea of eliminating this bias by splitting the samples from $q_1,q_2$ is further discussed in  \cite{wong2020properties}. 
The above argument also applies to the estimation of $RE^2(\hat r'^{(\phi_t)}_{opt})$. We compute $\widehat{RE}^2(\hat r'^{(\phi_t)}_{opt})$ based on the independent estimating samples using \eqref{minimalrmseestimate}. Since finding $\widehat{RE}^2(\hat r'^{(\phi_t)}_{opt})$ is a 1-d optimization problem, the additional computational cost is negligible compared to the rest of Algorithm \ref{algo_practical}. In practice, we recommend setting $n_i=n'_i$ for $i=1,2$, i.e. splitting the samples from $q_1,q_2$ into equally sized training samples and estimating samples.

\subsubsection*{Finding the saddle point using alternating gradient method}

In Algorithm \ref{algo_practical}, we aim to find a saddle point of $L_{\lambda_1, \lambda_2}(\phi, \tilde r; s_2, \{\omega_{ij}\}_{j=1}^{n_i})$ using the alternating gradient method. This approach is adapted from the Algorithm 1 of \cite{nowozin2016f}. The authors show that their Algorithm 1 converges geometrically to a saddle point under suitable conditions.
In the alternating training process of Algorithm \ref{algo_practical}, updating $\tilde r_{t+1}$ is a 1-d optimization problem when $\phi_t$ is treated as fixed for any step $t$. Hence we can also directly find $\hat r_{\phi_t} = \argmax_{\tilde r \in \mathbb{R}^+}  L_{\lambda_1, \lambda_2}(\phi_t, \tilde r; s_2, \{\omega_{ij}\}_{j=1}^{n_i})$ instead of performing a single step gradient ascent on $\tilde r_t$. By Proposition \ref{prop: estimating_harmonic_divergence} and \ref{prop: bridge_and_fdiv}, $\hat r_{\phi_t}$ can be viewed as a (biased) Bridge estimator of $r$ given $\phi_t$. However, such estimator $\hat r_{\phi_t}$ is not reliable when $q_1^{(\phi_t)}$ and $q_2$ share little overlap. Therefore directly optimizing $\tilde r$ at each iteration $t$ is not always necessary in practice, especially at the early stage of training when $q_1^{(\phi_t)}$ is not yet a sensible approximation of $q_2$. In addition, the gradient ascent update of $\tilde r_t$ is computationally cheaper than finding the optimizer $\hat r_{\phi_t}$ directly. Therefore we follow \cite{nowozin2016f} and use the alternating gradient method to find the saddle point of $L_{\lambda_1, \lambda_2}(\phi, \tilde r; s_2, \{\omega_{ij}\}_{j=1}^{n_i})$. We only recommend optimizing $\tilde r_{t}$ directly in Algorithm \ref{algo_practical} when we know $q_1^{(\phi_t)}$ and $q_2$ have at least some degree of overlap.

Note that $\{\phi_t, \tilde r_t\}$ being approximately a saddle point of the objective function does not necessarily imply that it solves $\min_{\phi \in \mathbb{R}^l} \max_{\tilde r \in \mathbb{R}^+}  L_{\lambda_1, \lambda_2}(\phi, \tilde r; s_2, \{\omega_{ij}\}_{j=1}^{n_i})$. For $\tilde r_t$, it is easy to verify if $\tilde r_t$ is indeed the maximizer of $L_{\lambda_1, \lambda_2}(\phi_t, \tilde r; s_2, \{\omega_{ij}\}_{j=1}^{n_i})$ w.r.t. $\tilde r \in \mathbb{R}^+$ since it is a 1-d optimization problem. However, for $\phi_t$ there is no guarantee that it is the global minimizer of $L_{\lambda_1, \lambda_2}(\tilde r_t, \phi; s_2, \{\omega_{ij}\}_{j=1}^{n_i})$ w.r.t. $\phi \in \mathbb{R}^l$. One way to address this problem is to run Algorithm \ref{algo_practical} multiple times and choose the $q_1^{(\phi_t)}$ that attains the smallest objective function value. In the numerical examples, we find $q_1^{(\phi_t)}$ returned from Algorithm \ref{algo_practical} is almost always a good approximation of $q_2$. Therefore we do not worry about this problem in practice.

In the alternating training process, seeing the absolute difference between $L_{\lambda_1, \lambda_2}(\phi_t, \tilde r_t; s_2, \{\omega_{ij}\}_{j=1}^{n_i})$ and  $L_{\lambda_1, \lambda_2}(\phi_{t-1},\tilde r_{t-1}; s_2, \{\omega_{ij}\}_{j=1}^{n_i})$ being less than the tolerance level $\epsilon_1$ at an iteration $t$ does not necessarily imply that it has reached a saddle point. Therefore we also need to monitor the sequence $\{\tilde r_t\}$, $t=0,1,2,...$ in the training process. If $\vert\tilde r_t - \tilde r_{t-1}\vert > \epsilon_2$, then $\tilde r_t$ has not converged to a stationary point regardless of the value of the objective function. In other words, we know $\{\phi_t, \tilde r_t\}$ has approximately converged to a saddle point only if both the objective function $L_{\lambda_1, \lambda_2}(\phi_t, \tilde r_t; s_2, \{\omega_{ij}\}_{j=1}^{n_i})$ and $\tilde r_t$ have stopped changing. In practice, we recommend setting $\epsilon_1 \in (10^{-3},10^{-1})$ depending on the scale of the objective function, and $\epsilon_2 \in (10^{-3}, 10^{-2})$.

\subsubsection*{Effectiveness of the hybrid objective}
As we have discussed previously, we introduce the hybrid objective to stabilize the alternating training process and accelerate the convergence of Algorithm \ref{algo_practical}. Here we demonstrate the effectiveness of the hybrid objective in Algorithm \ref{algo_practical} using the mixture of rings example in Sec \ref{sec: mixture of rings} with $p=12$, $\boldsymbol{\mu}_{11} = (2,2), \boldsymbol{\mu}_{12} = (-2, -2), \boldsymbol{\mu}_{21}=(2,-2), \boldsymbol{\mu}_{21}=(-2,2)$, $b_1 =3 , b_2=6, \sigma_1=1, \sigma_2=2$. We set $q_1^{(\phi)}$ to be a Real-NVP with 5 coupling layers. We first run Algorithm \ref{algo_practical} 50 times with $n_i=n'_i=1000$ for $i=1,2$ and $\lambda_1=\lambda_2=0.05$. We record the values of the objective function and $\tilde r_t$ of the first 25 iterations. Then we run Algorithm \ref{algo_practical} 50 times with $n_i=n'_i=1000$ for $i=1,2$ and $\lambda_1=\lambda_2=0$, and record the same values. Recall that setting  $\lambda_1=\lambda_2=0$ is equivalent to using the original $f$-GAN objective \eqref{harmoniclowerbound_transformed_empirical}. From Figure \ref{fig:fgan_vs_hybrid} we see most of the hybrid objectives and the corresponding $\tilde r_t$ values have stabilized after 20 iterations. The stand alone $f$-GAN objective with $\lambda_1=\lambda_2=0$ also demonstrate a decreasing trend, but the objective values are much more wiggly compared to the hybrid objective due to the adversarial training process, and there is no sign of convergence in 25 iterations. Note that for both the hybrid objective and the original $f$-GAN objective, the corresponding $\tilde r_t$ tend to converge to a value slightly different from the true $r$ as the number of iteration increases. This is likely due to the bias we discussed previously.
\begin{figure}
    \centering
    \includegraphics[width=\textwidth]{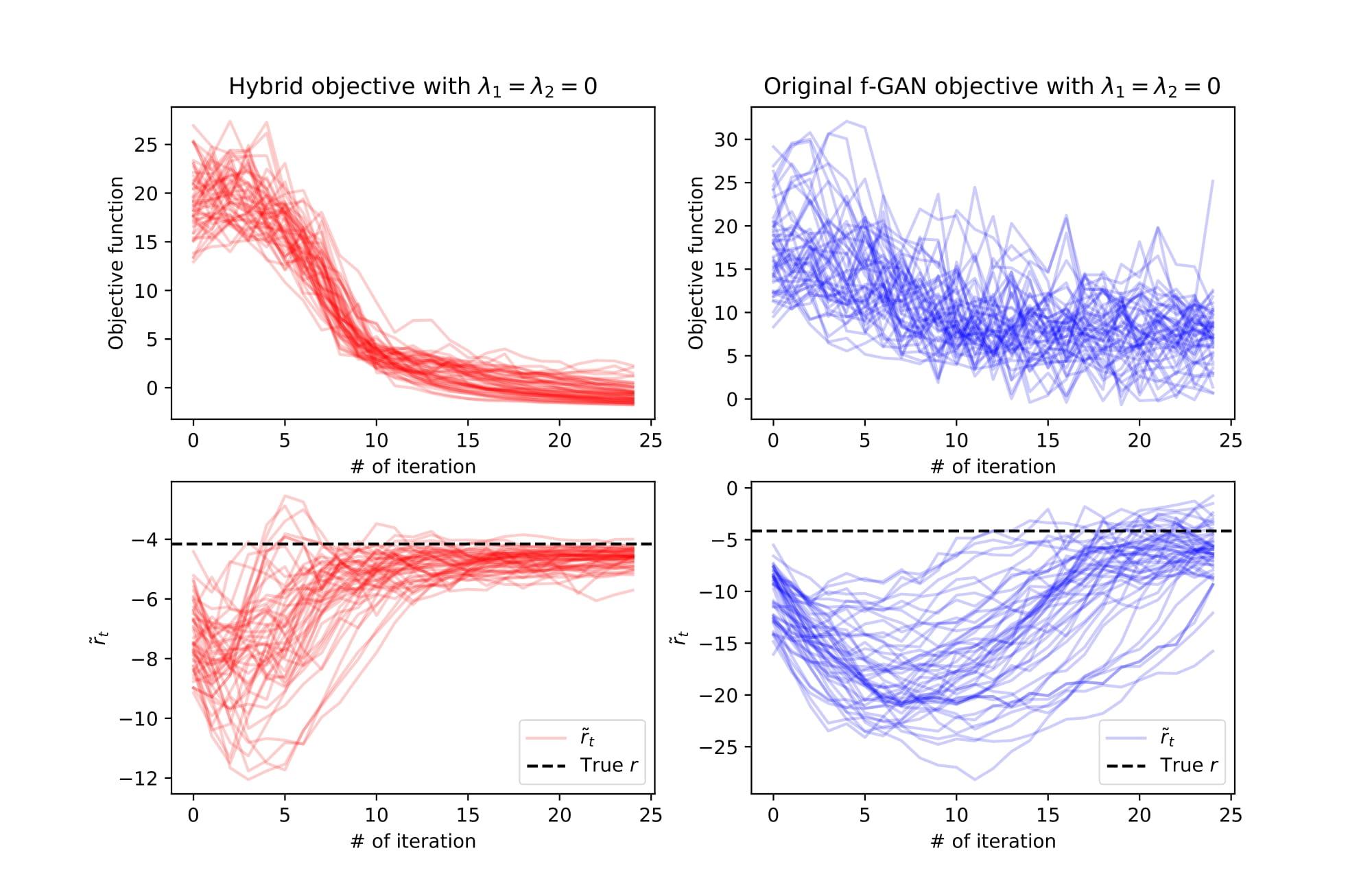}
    \caption{Left: The objective function and $\tilde r_t$ of the first 25 iterations of Algorithm \ref{algo_practical} with $\lambda_1=\lambda_2=0.05$. Right: The objective function and $\tilde r_t$ of the first 25 iterations of Algorithm \ref{algo_practical} with $\lambda_1=\lambda_2=0$.}
    \label{fig:fgan_vs_hybrid}
\end{figure}

\section{Additional simulations}\label{appendix:more_simulation}

\subsection*{Simulated example: Quantized Mixture of Gaussians} \label{appendix: discrete}

Here we illustrate how Normalizing flows can be used to increase the overlap between discrete random variables in the context of estimating a single normalizing constant (i.e. one of $q_1,q_2$ is completely known). We take the quantized Mixture of Gaussian in \cite{tran2019discrete} and \cite{MetzPPS17} as a toy example.

Following \cite{tran2019discrete}, we first define the completely known ``base" distribution $q_1$. Let $\omega^{(1)},\omega^{(2)}$ be two categorical variables each with 90 states. Let $\omega=(\omega^{(1)},\omega^{(2)})$. Let $q_1$ be a uniform distribution over all possible states of $\omega$. The probability mass function of $q_1$ is then
\begin{equation}
    q_1(\omega^{(1)}=u, \omega^{(2)}=v) = \frac{1}{90 \times 90}  \quad \quad u,v \in \{1,2,...,90\} \label{discrete_unif}
\end{equation}

We then define the quantized Mixture of Gaussian distribution as our ``target" distribution $q_2$.  In order to define the quantized Mixture of Gaussian, we first define $\tilde g(x)$, the unnormalized density of a mixture of 2D Gaussian distributions, to be
\begin{equation}
    \tilde g(x) = \sum_{k=1}^K \pi_{k} \tilde p(x; \mu_{k}, \sigma^2I_2)
\end{equation}
where $x\in\mathbb{R}^2$, $I_2$ is the $2\times 2$ identity matrix, $\tilde p(\cdot; \mu, \Sigma) = \exp(-\frac{1}{2}(x-\mu)^T\Sigma^{-1}(x-\mu))$ is the unnormalized 2D Gaussian density with mean $\mu$ and covariance $\Sigma$, $K=4$, $\sigma=0.1$, $\mu_1=(2,0)$, $\mu_2=(-2,0)$, $\mu_3=(0,2)$, $\mu_4=(0,-2)$ and $\pi_{k}=\frac{1}{K}$ for $k=1,...,K$. 
We then truncate the support of $\tilde g(x)$ to be $[ -2.25, 2.25]^2$. We now define $q_2(\omega)$, the quantized 2D Mixture of Gaussian distribution, by discretizing this square at the 0.05 level (i.e. forming a $90\times90$ equally spaced grid). This discretization step leads to two categorical variables $\omega^{(1)},\omega^{(2)}$ each with 90 states. For $u, v \in \{1,2,...,90\}$, let $B_{uv} \subseteq [ -2.25, 2.25 ]^2$ be the cell of the grid that corresponds to the state $\{\omega^{(1)}=u, \omega^{(2)}=v\}$. Then the unnormalized probability mass function of $q_2$ can be written as
\begin{equation}
    \tilde q_2(\omega^{(1)}=u, \omega^{(2)}=v) = \int_{x\in B_{uv}}  \tilde g(x) dx  \quad \quad u,v \in \{1,2,...,90\}.
\end{equation}
See Figure \ref{fig: qMoG_transformed} for a 2D histogram of samples from $q_2$.
Let 
\begin{equation}
    Z_2 = \sum_{u=1}^{90}\sum_{v=1}^{90} \int_{x\in B_{uv}}  \tilde g(x) dx 
\end{equation}
be the normalizing constant of $\tilde q_2(\omega)$, which can be computed easily. Let $q_2(\omega) = \tilde q_2(\omega)/Z_2$ be the corresponding normalized pmf. Since $q_1$ is completely known, its normalizing constant $Z_1$ is equal to $1$ and therefore $\tilde q_1(\omega)=q_1(\omega)$. 

Our goal is to estimate $\log r= \log Z_1 - \log Z_2=- \log Z_2$ by first increasing the overlap between $q_1$ and $q_2$ using a Normalizing flow, then compute the asymptotically optimal Bridge estimator of $r$ based of the transformed distributions. Let $N=\{500,1000,1500,2000,2500\}$. To demonstrate the effectiveness of this approach, for each value of $N$, we first draw $n_1=n_2=N$ training samples $\{\omega_{1j}\}_{j=1}^{n_1}$ and $\{\omega_{2j}\}_{j=1}^{n_2}$ from $q_1,q_2$ respectively, and use an autoregressive discrete flow \cite{tran2019discrete} to estimate a bijective transformation $T$ that maps $q_1$ to $q_2$ based on the training samples and the training procedure given by \cite{tran2019discrete}. One key distinction between discrete flows \cite{tran2019discrete} and their continuous counterparts (e.g. \cite{dinh2016density, kingma2016improved}) is that for discrete flows, the base distribution $q_1$ is treated as a model parameter and is estimated jointly with the transformation $T$. In our example, this means the parameterization (i.e. the $90\times 90$ probability table) we chose for $q_1$ in \eqref{discrete_unif} is only treated as the ``initial values" of the model parameters, and is updated alongside with the transformation $T$. (Note that when $q_1,q_2$ have a large number of discrete states, storing or updating the probability table of the base $q_1$ is computationally infeasible. To alleviate this problem, \cite{tran2019discrete} also considered more sophisticated parameterizations of the ``trainable base" $q_1$ such as the autoregressive Categorical distribution.) Let $T_1$ be the estimated transformation, $\bar{q}_1$ be the updated base distribution (which is also completely known and easy to sample from). Let $\bar{q}_1^{(T)}$ be the transformed distribution obtained by applying $T_1$ to the samples from the updated $\bar{q}_1$. We then draw $n'_1=n'_2=N$ estimating samples $\{\bar{\omega}'_{1j}\}_{j=1}^{n'_1}$ and $\{\omega'_{2j}\}_{j=1}^{n'_2}$ from $\bar{q}_1, q_2$ respectively, and compute $\hat r^{(T)}_{opt}$ in \eqref{transformedbridge} based on the transformed $\{T_1(\bar{\omega}'_{1j})\}_{j=1}^{n'_1}$ and the original $\{\omega'_{2j}\}_{j=1}^{n'_2}$. For each value of $N$, we repeat this process 100 times, and report the MC estimate of $MSE(\hat r^{(T)}_{opt})$ based on the repeated $\hat r^{(T)}_{opt}$s and the ground truth $r$. Let $\hat r_{opt}$ be the optimal Bridge estimator based on the original $\tilde q_1, \tilde q_2$. For each value of $N$, we compare $MSE(\log \hat r^{(T)}_{opt})$ with $MSE(\log \hat r_{opt})$, which is also estimated based on 100 repetitions in a similar fashion. From Figure \ref{fig:discrete error} we see $\hat r^{(T)}_{opt}$ is a reliable estimator of $r$ for all choice of $N$ and is much more accurate than the optimal Bridge estimator based on the original $\tilde q_1, \tilde q_2$. From Figure \ref{fig: qMoG_transformed} we also see that the transformed $\bar q_1^{(T)}$ accurately captures the multimodal structure of $q_2$.

\begin{figure}
    \centering
    \includegraphics[width=\textwidth]{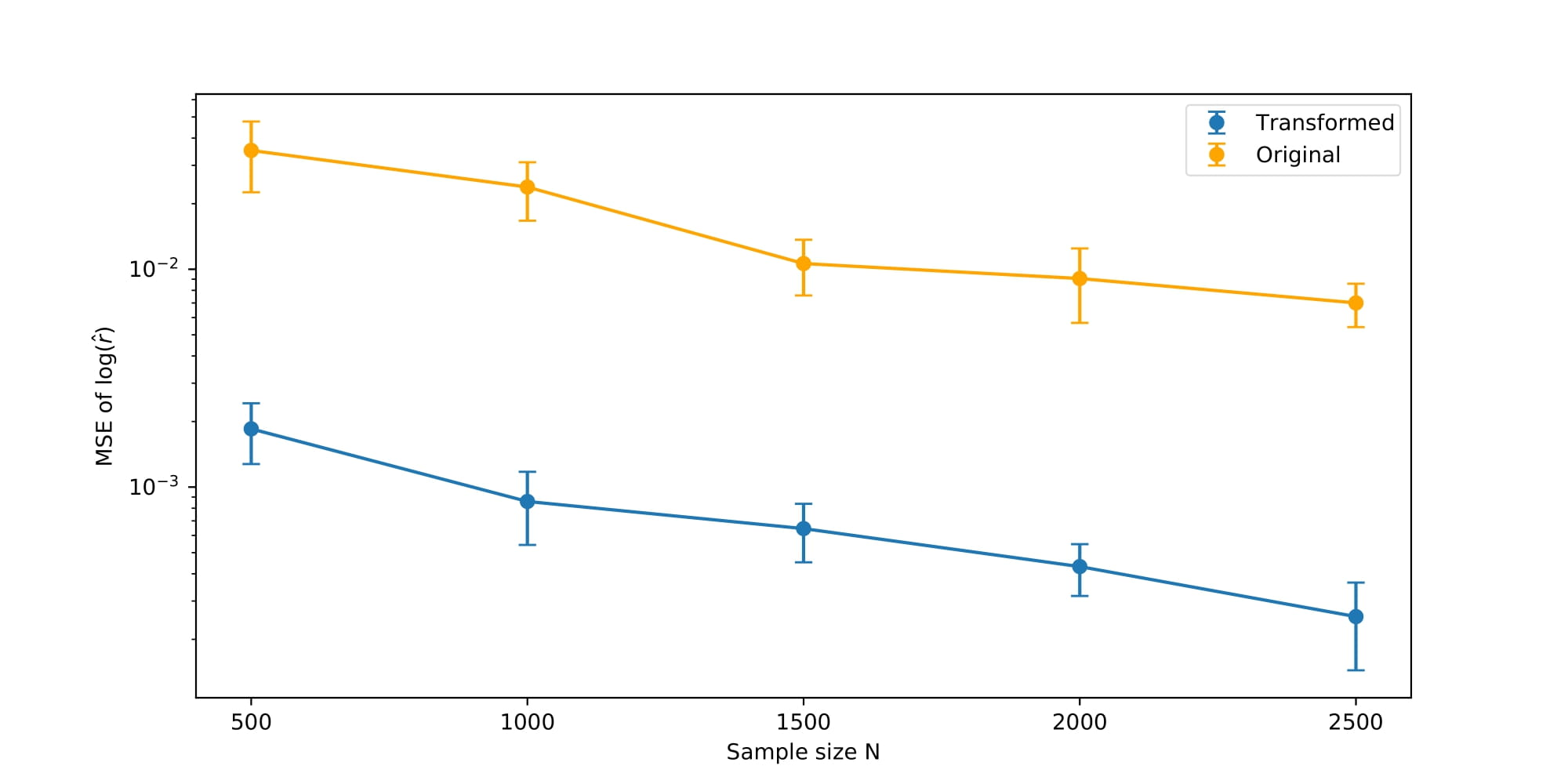}
    \caption{MC estimate of MSE of $\log \hat r^{(T)}_{opt}$ and $\log \hat r_{opt}$ for each value of $N$. Vertical segments represent the $2\sigma$ error bars.}
    \label{fig:discrete error}
\end{figure}

\begin{figure}
    \centering
    \includegraphics[width=\textwidth]{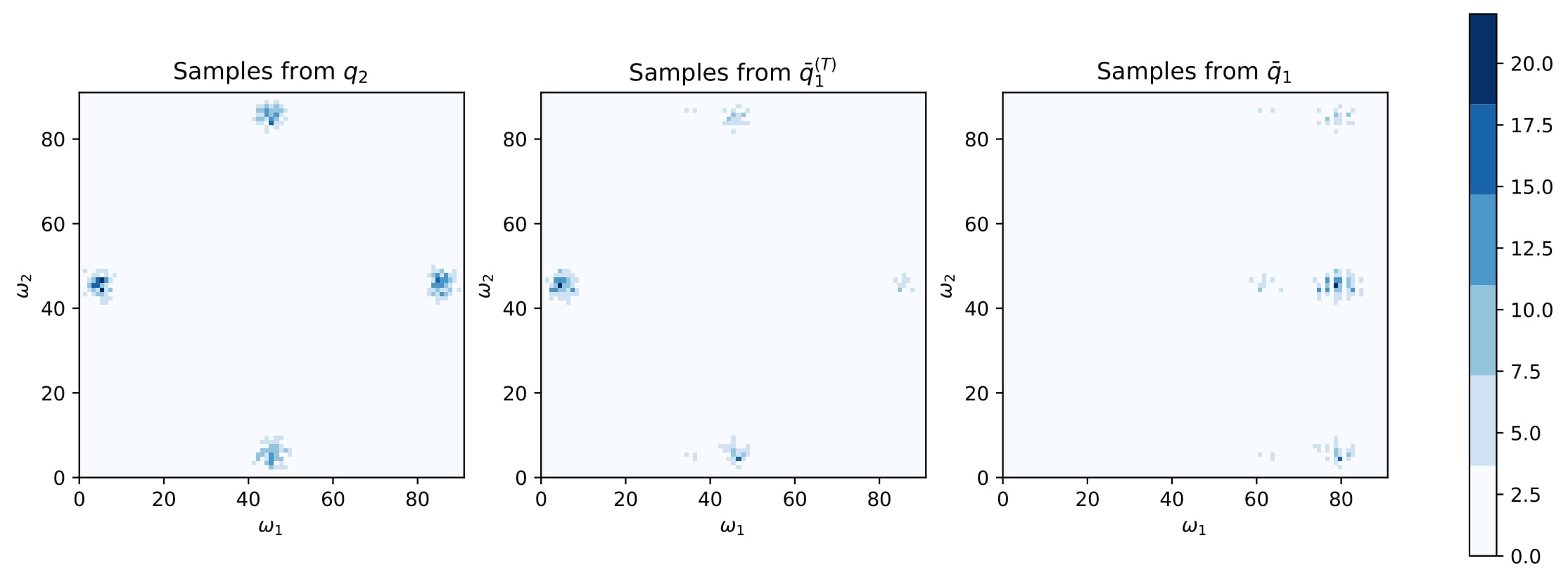}
    \caption{Left: 2D histogram of $10^3$ samples from $q_2$. Mid: 2D histogram of $10^3$ samples from the transformed $\bar{q}_1^{(T)}$, which is estimated using $N=10^3$ training samples. Right: 2D histogram of $10^3$ samples from the corresponding updated base distribution $\bar{q}_1$ of the transformed $\bar{q}_1^{(T)}$.}
    \label{fig: qMoG_transformed}
\end{figure}

In addition to the quantized mixture of Gaussian example, more substantial applications of discrete flows can also be found in \cite{tran2019discrete}. However, the discrete flows in \cite{tran2019discrete} are in general not directly applicable to our proposed Algorithm \ref{algo_practical}. This is because in our Algorithm \ref{algo_practical}, the unnormalized densities $\tilde q_1$ and $\tilde q_2$ are specified by the users and therefore can be arbitrary. However, for discrete flows, the ``base" distribution has to be completely known, and is treated as a trainable model parameter (as in this example).  This means we are not able to use it to directly estimate the ratio of normalizing constants between two arbitrary unnormalized probability mass functions in the same way as in Algorithm \ref{algo_practical}. Nevertheless, one may obtain an estimator of the ratio of normalizing constants between two discrete distributions by estimating their normalizing constants separately using discrete flows and the procedure used in this example. For future work, we are interested in extending Algorithm \ref{algo_practical} so that it is able to handle arbitrary unnormalized pmfs using e.g. more sophisticated Normalizing flow architectures.

\subsection*{Simulated example: Mixture of $t$-distributions}
 In this example, we let $q_1$ and $q_2$ be two mixtures of $p$ dimensional $t$-distributions. We are interested in this example because both $q_1,q_2$ are multimodal and have heavy tails. For $i=1,2$, let 
\begin{equation}
    q_i(\omega) = \sum_{k=1}^K \pi_{ik} p_t(\omega; \mu_{ik}, \Sigma_i, \nu_i),
\end{equation}
where $K$ is the number of components, $\pi_i = \{\pi_{ik}\}_{k=1}^K$ are the mixing weights and  $p_t(\cdot; \mu_{ik}, \Sigma_i, \nu_i)$, the $k$th component of $q_i$ is the pdf of a multivariate $t$-distribution with mean $\mu_{ik} \in \mathbb{R}^p$, positive-definite scale matrix $\Sigma_i \in \mathbb{R}^{p\times p}$ and degree of freedom $\nu_i \in \mathbb{R}^+$. Note that all $K$ components of $q_i$ have the same covariance structure and degree of freedom. Let 
\begin{equation}
    Z_i = \frac{\Gamma((\nu_i+p)/2)}{\Gamma(\nu_i/2) \nu_i^{p/2} \pi^{p/2} \vert\Sigma_i\vert^{1/2}}
\end{equation}
be the normalizing constant of $p_t(\cdot; \mu_{ik}, \Sigma_i, \nu_i)$. Note that this quantity does not depend on $\mu_{ik}$. Let $\tilde p_t(\cdot; \mu_{ik}, \Sigma_i, \nu_i) = Z_i p_t(\cdot; \mu_{ik}, \Sigma_i, \nu_i)$ be the unnormalized density of each component $p_t(\cdot; \mu_{ik}, \Sigma_i, \nu_i)$. Let $\tilde q_i(\omega) = \sum_{k=1}^K \pi_{ik} \tilde p_t(\omega; \mu_{ik}, \Sigma_i, \nu_i)$ be the unnormalized density of $q_i(\omega)$. It is easy to verify that $\tilde q_i(\omega) = Z_i q_i(\omega)$, i.e. the normalizing constant of $\tilde q_i(\omega)$ is $Z_i$.

For this example, we consider $p=\{5,10,20,40,60,80,100\}$. For each choice of $p$, the parameters of $q_1,q_2$ are chosen in the following way: We fix the degree of freedom $\nu_1 = 1, \nu_2 = 4$, and number of component $K=7$. The mixing weights $\pi_i$ for $i=1,2$ are sampled independently from a $Dir(\alpha_1,...,\alpha_K)$ where $\alpha_k=1$ for $k=1,...,K$. The location parameters $\mu_{ik}$ for $i=1,2$, $k=1,...,K$ are sampled from a standard Normal $N(0,I_p)$ independently. For the scale matrices $\Sigma_i$, we first sample $\Sigma_1, \Sigma_2$ independently from a inverse Wishart distribution $\mathcal{W}^{-1}(I_p, p)$, then rescale $\Sigma_1 ,\Sigma_2$ so that $\left\vert\Sigma_1\right\vert = 1$ and $\left\vert\Sigma_2\right\vert\ = 1000$. This ensures the components of $q_1$ are more concentrated than the components in $q_2$.

We estimate $ \log r= \log Z_1 - \log Z_2$ in a similar fashion to the previous examples. For each choice of $p$, we run each method 30 times. Let $\hat r$ be a generic estimate of $r$. We use the MC estimate of RMSE of $\log \hat r$, i.e. $E((\log \hat r - r)^2)/(\log r)^2$, based on the repeated runs as the benchmark of performance for this example. For each repetition, we run each method with $N_1=N_2=6000$ independent samples from $q_1,q_2$ respectively. For our Algorithm \ref{algo_practical}, we parameterize $q_1^{(\phi)}$ as a Real-NVP with 20 coupling layers, and set $\lambda_1=\lambda_2=0.01$. For the rest of the methods, we use the deafult or recommended settings. The results are summarized in Table \ref{table: mixture_of_t}. We see our Algorithm \ref{algo_practical} outperforms all methods when $p \geq 40$.

\begin{table}[ht]
\centering
\begin{tabular}{ccccccc}
  \hline
$p$ & $f$-GAN &  GBS & Warp-III & Warp-U \\ 
  \hline
$5$ &  3.69\text{e-}{5} &  8.22\text{e-}{4} & 1.14\text{e-}{3} &  \bf{3.54\text{e-}{5}} \\ 
$10$ &   \bf{6.21\text{e-}{5}} &  1.74\text{e-}{3} & 5.15\text{e-}{3} & 6.42\text{e-}{5} \\ 
$20$ & 4.12\text{e-}{3} & 4.96\text{e-}{3} & 8.87\text{e-}{3} & \bf{1.69\text{e-}{3}}\\
$40$ & \bf{1.23\text{e-}{2}} &  4.05\text{e-}{2} & 9.01\text{e-}{2} & 5.75\text{e-}{2}\\
$60$ & \bf{1.21\text{e-}{2}} &  3.88\text{e-}{2} & 9.26\text{e-}{2} & 7.64\text{e-}{2}\\ 
$80$ & \bf{1.81\text{e-}{2}} & 5.20\text{e-}{2} & 1.59\text{e-}{1} & 6.05\text{e-}{2}\\
$100$& \bf{2.46\text{e-}{2}} &  8.14\text{e-}{2} & - & 4.78\text{e-}{1}\\
   \hline
\end{tabular}
\caption{The estimated RMSE of the $\log \hat r$ of each methods based on 30 repeated runs. The lowest estimated RMSE for each $p$ is in boldface. Warp-III does not converge for most of the repeated runs when $p = 100$ so we are not able to estimate its RMSE.}
\label{table: mixture_of_t}
\end{table}

\section{Computational cost of Algorithm \ref{algo_practical}}\label{appendix:computational_cost}

Comparing the computational cost of our Algorithm \ref{algo_practical} with existing methods and their existing implementations is not straightforward because of the very different
nature of GPU and CPU computing. In both examples, we compare the existing CPU implementations of Warp-III, Warp-U, GBS and a GPU implementation of our Algorithm 2 in term of wall clock time. We think this comparison is not unfair because there is no simple way to accelerate existing algorithms with a GPU, while training neural nets on GPU was a design element in implementing Algorithm \ref{algo_practical} using deep learning frameworks such as Torch \citep{paszke2017automatic}. If a user have access to both CPU and GPU, then we believe the wall clock time to some extent can be viewed as a natural metric of the time cost a user has to pay for the estimator. And the Precision per Second benchmark can be viewed as the cost-performance ratio of these methods. This measure is not a rigorous metric for comparing computation costs, but we believe it is at least an intuitive one for the users to get a rough idea of the time cost and efficiency of these algorithms. 

From the numerical examples we see $f$-GB scaled better with dimension than its competitors. 
For example, from Example 1 we see that even though Warp-III can be $30\sim 40$ faster to compute than our proposed method given the same amount of samples from $q_1,q_2$, its accuracy (measured in $MSE(\log \hat r)$) is orders of magnitude greater (worse) than our approach. When $p=48$, we find that Warp-III is not able to return a sensible estimate of $r$ even with 25 times more samples from $q_1,q_2$ than $f$-GB. In Example 2 we also find that Warp-III requires around 18 times more samples to achieve a similar level of accuracy as $f$-GB, and takes around 3 times longer to run. Therefore we believe the extra computational cost of our $f$-GB estimator is ``well-spent" as the numerical examples show that our Algorithm \ref{algo_practical} is able to return an estimate of $r$ with much higher precision than GBS, Warp-III and Warp-U and scales better with the dimension of the distributions.

As we acknowledge in Sec \ref{sec: conclusion}, if $q_1,q_2$ are simple structured and low dimensional, then users can get adequate precision more quickly using Warp-III or Warp-U. On the other hand, in speaking to users who report Bayes factors in applied Bayesian work, the overwhelming requirement was that the estimate be reliable, as the Bayes Factor value is sometimes the crux. In all the numerical examples we considered, $f$-GB never broke and produced accurate estimates. Warp III and GBS did break on larger problems. Warp-U took a similar amount of time to run compared with $f$-GB, but was less accurate. Therefore users may still prefer a method which is ``over-powered'' but more reliable.

\bibliographystyle{apalike}

\begin{thebibliography}{}

\bibitem[Ali and Silvey, 1966]{ali1966general}
Ali, S.~M. and Silvey, S.~D. (1966).
\newblock A general class of coefficients of divergence of one distribution
  from another.
\newblock {\em Journal of the Royal Statistical Society: Series B
  (Methodological)}, 28(1):131--142.

\bibitem[Arjovsky and Bottou, 2017]{arjovsky2017towards}
Arjovsky, M. and Bottou, L. (2017).
\newblock Towards principled methods for training generative adversarial
  networks.
\newblock {\em arXiv preprint arXiv:1701.04862}.

\bibitem[Bennett, 1976]{bennett1976efficient}
Bennett, C.~H. (1976).
\newblock Efficient estimation of free energy differences from monte carlo
  data.
\newblock {\em Journal of Computational Physics}, 22(2):245--268.

\bibitem[Bridges et~al., 2009]{bridges2009bayesian}
Bridges, M., Feroz, F., Hobson, M., and Lasenby, A. (2009).
\newblock Bayesian optimal reconstruction of the primordial power spectrum.
\newblock {\em Monthly Notices of the Royal Astronomical Society},
  400(2):1075--1084.

\bibitem[Burda et~al., 2015]{burda2015accurate}
Burda, Y., Grosse, R., and Salakhutdinov, R. (2015).
\newblock Accurate and conservative estimates of mrf log-likelihood using
  reverse annealing.
\newblock In {\em Artificial Intelligence and Statistics}, pages 102--110.
  PMLR.

\bibitem[Chen and Shao, 1997]{chen1997estimating}
Chen, M.-H. and Shao, Q.-M. (1997).
\newblock Estimating ratios of normalizing constants for densities with
  different dimensions.
\newblock {\em Statistica Sinica}, pages 607--630.

\bibitem[Creswell et~al., 2018]{creswell2018generative}
Creswell, A., White, T., Dumoulin, V., Arulkumaran, K., Sengupta, B., and
  Bharath, A.~A. (2018).
\newblock Generative adversarial networks: An overview.
\newblock {\em IEEE Signal Processing Magazine}, 35(1):53--65.

\bibitem[Dinh et~al., 2016]{dinh2016density}
Dinh, L., Sohl-Dickstein, J., and Bengio, S. (2016).
\newblock Density estimation using real nvp.
\newblock {\em arXiv preprint arXiv:1605.08803}.

\bibitem[Durkan et~al., 2019]{durkan2019neural}
Durkan, C., Bekasov, A., Murray, I., and Papamakarios, G. (2019).
\newblock Neural spline flows.
\newblock {\em arXiv preprint arXiv:1906.04032}.

\bibitem[Fitzmaurice and Laird, 1993]{fitzmaurice1993likelihood}
Fitzmaurice, G.~M. and Laird, N.~M. (1993).
\newblock A likelihood-based method for analysing longitudinal binary
  responses.
\newblock {\em Biometrika}, 80(1):141--151.

\bibitem[Fourment et~al., 2020]{fourment202019}
Fourment, M., Magee, A.~F., Whidden, C., Bilge, A., Matsen~IV, F.~A., and
  Minin, V.~N. (2020).
\newblock 19 dubious ways to compute the marginal likelihood of a phylogenetic
  tree topology.
\newblock {\em Systematic biology}, 69(2):209--220.

\bibitem[Friel and Wyse, 2012]{friel2012estimating}
Friel, N. and Wyse, J. (2012).
\newblock Estimating the evidence--a review.
\newblock {\em Statistica Neerlandica}, 66(3):288--308.

\bibitem[Fr{\"u}hwirth-Schnatter, 2004]{fruhwirth2004estimating}
Fr{\"u}hwirth-Schnatter, S. (2004).
\newblock Estimating marginal likelihoods for mixture and markov switching
  models using bridge sampling techniques.
\newblock {\em The Econometrics Journal}, 7(1):143--167.

\bibitem[Gelman and Meng, 1998]{gelman1998simulating}
Gelman, A. and Meng, X.-L. (1998).
\newblock Simulating normalizing constants: From importance sampling to bridge
  sampling to path sampling.
\newblock {\em Statistical science}, pages 163--185.

\bibitem[Geweke, 1999]{geweke1999using}
Geweke, J. (1999).
\newblock Using simulation methods for bayesian econometric models: inference,
  development, and communication.
\newblock {\em Econometric reviews}, 18(1):1--73.

\bibitem[Geyer, 1994]{geyer1994estimating}
Geyer, C.~J. (1994).
\newblock Estimating normalizing constants and reweighting mixtures.
\newblock {\em Technical Report 568, School of Statistics, University of
  Minnesota}.

\bibitem[Goodfellow et~al., 2014]{goodfellow2014generative}
Goodfellow, I., Pouget-Abadie, J., Mirza, M., Xu, B., Warde-Farley, D., Ozair,
  S., Courville, A., and Bengio, Y. (2014).
\newblock Generative adversarial nets.
\newblock In {\em Advances in neural information processing systems}, pages
  2672--2680.

\bibitem[Grover et~al., 2018]{grover2018flowgan}
Grover, A., Dhar, M., and Ermon, S. (2018).
\newblock Flow-gan: Combining maximum likelihood and adversarial learning in
  generative models.

\bibitem[Gutmann and Hyv{\"a}rinen, 2010]{gutmann2010noise}
Gutmann, M. and Hyv{\"a}rinen, A. (2010).
\newblock Noise-contrastive estimation: A new estimation principle for
  unnormalized statistical models.
\newblock In {\em Proceedings of the Thirteenth International Conference on
  Artificial Intelligence and Statistics}, pages 297--304.

\bibitem[Jennrich, 1969]{jennrich1969asymptotic}
Jennrich, R.~I. (1969).
\newblock Asymptotic properties of non-linear least squares estimators.
\newblock {\em The Annals of Mathematical Statistics}, 40(2):633--643.

\bibitem[Jia and Seljak, 2020]{jia2020normalizing}
Jia, H. and Seljak, U. (2020).
\newblock Normalizing constant estimation with gaussianized bridge sampling.
\newblock In {\em Symposium on Advances in Approximate Bayesian Inference},
  pages 1--14. PMLR.

\bibitem[Kingma et~al., 2016]{kingma2016improved}
Kingma, D.~P., Salimans, T., Jozefowicz, R., Chen, X., Sutskever, I., and
  Welling, M. (2016).
\newblock Improved variational inference with inverse autoregressive flow.
\newblock {\em Advances in neural information processing systems},
  29:4743--4751.

\bibitem[Kong et~al., 2003]{kong2003theory}
Kong, A., McCullagh, P., Meng, X.-L., Nicolae, D., and Tan, Z. (2003).
\newblock A theory of statistical models for monte carlo integration.
\newblock {\em Journal of the Royal Statistical Society: Series B (Statistical
  Methodology)}, 65(3):585--604.

\bibitem[Lartillot and Philippe, 2006]{lartillot2006computing}
Lartillot, N. and Philippe, H. (2006).
\newblock Computing bayes factors using thermodynamic integration.
\newblock {\em Systematic biology}, 55(2):195--207.

\bibitem[Le~Cam, 1969]{le1969theorie}
Le~Cam, L.~M. (1969).
\newblock Th{\'e}orie asymptotique de la d{\'e}cision statistique.
\newblock {\em Presses de l'Universit{\'e} de Montr{\'e}al}.

\bibitem[Lin, 1991]{lin1991divergence}
Lin, J. (1991).
\newblock Divergence measures based on the shannon entropy.
\newblock {\em IEEE Transactions on Information theory}, 37(1):145--151.

\bibitem[Lunn et~al., 2000]{lunn2000winbugs}
Lunn, D.~J., Thomas, A., Best, N., and Spiegelhalter, D. (2000).
\newblock Winbugs-a bayesian modelling framework: concepts, structure, and
  extensibility.
\newblock {\em Statistics and computing}, 10(4):325--337.

\bibitem[Meng and Schilling, 1996]{meng1996fitting}
Meng, X.-L. and Schilling, S. (1996).
\newblock Fitting full-information item factor models and an empirical
  investigation of bridge sampling.
\newblock {\em Journal of the American Statistical Association},
  91(435):1254--1267.

\bibitem[Meng and Schilling, 2002]{meng2002warp}
Meng, X.-L. and Schilling, S. (2002).
\newblock Warp bridge sampling.
\newblock {\em Journal of Computational and Graphical Statistics},
  11(3):552--586.

\bibitem[Meng and Wong, 1996]{meng1996simulating}
Meng, X.-L. and Wong, W.~H. (1996).
\newblock Simulating ratios of normalizing constants via a simple identity: a
  theoretical exploration.
\newblock {\em Statistica Sinica}, pages 831--860.

\bibitem[Metz et~al., 2017]{MetzPPS17}
Metz, L., Poole, B., Pfau, D., and Sohl{-}Dickstein, J. (2017).
\newblock Unrolled generative adversarial networks.
\newblock In {\em 5th International Conference on Learning Representations,
  {ICLR}, Toulon, France}.

\bibitem[Newey and McFadden, 1994]{newey1994large}
Newey, W.~K. and McFadden, D. (1994).
\newblock Large sample estimation and hypothesis testing.
\newblock {\em Handbook of econometrics}, 4:2111--2245.

\bibitem[Nguyen et~al., 2010]{nguyen2010estimating}
Nguyen, X., Wainwright, M.~J., and Jordan, M.~I. (2010).
\newblock Estimating divergence functionals and the likelihood ratio by convex
  risk minimization.
\newblock {\em IEEE Transactions on Information Theory}, 56(11):5847--5861.

\bibitem[Nowozin et~al., 2016]{nowozin2016f}
Nowozin, S., Cseke, B., and Tomioka, R. (2016).
\newblock f-gan: Training generative neural samplers using variational
  divergence minimization.
\newblock In {\em Advances in neural information processing systems}, pages
  271--279.

\bibitem[NVIDIA et~al., 2020]{cuda}
NVIDIA, Vingelmann, P., and Fitzek, F.~H. (2020).
\newblock Cuda, release: 10.2.89.

\bibitem[Overstall and Forster, 2010]{overstall2010default}
Overstall, A.~M. and Forster, J.~J. (2010).
\newblock Default bayesian model determination methods for generalised linear
  mixed models.
\newblock {\em Computational Statistics \& Data Analysis}, 54(12):3269--3288.

\bibitem[Papamakarios et~al., 2017]{papamakarios2017masked}
Papamakarios, G., Pavlakou, T., and Murray, I. (2017).
\newblock Masked autoregressive flow for density estimation.
\newblock In {\em Advances in Neural Information Processing Systems}, pages
  2338--2347.

\bibitem[Paszke et~al., 2017]{paszke2017automatic}
Paszke, A., Gross, S., Chintala, S., Chanan, G., Yang, E., DeVito, Z., Lin, Z.,
  Desmaison, A., Antiga, L., and Lerer, A. (2017).
\newblock Automatic differentiation in pytorch.
\newblock In {\em NIPS 2017 Workshop on Autodiff}.

\bibitem[Pinsker, 1964]{pinsker1964information}
Pinsker, M.~S. (1964).
\newblock Information and information stability of random variables and
  processes.
\newblock {\em Holden-Day}.

\bibitem[Ranganath et~al., 2014]{ranganath2014black}
Ranganath, R., Gerrish, S., and Blei, D. (2014).
\newblock Black box variational inference.
\newblock In {\em Artificial intelligence and statistics}, pages 814--822.
  PMLR.

\bibitem[Rezende and Mohamed, 2015]{rezende2015variational}
Rezende, D.~J. and Mohamed, S. (2015).
\newblock Variational inference with normalizing flows.
\newblock {\em arXiv preprint arXiv:1505.05770}.

\bibitem[Skilling et~al., 2006]{skilling2006nested}
Skilling, J. et~al. (2006).
\newblock Nested sampling for general bayesian computation.
\newblock {\em Bayesian analysis}, 1(4):833--859.

\bibitem[Sturtz et~al., 2005]{sturtz2005r2winbugs}
Sturtz, S., Ligges, U., and Gelman, A.~E. (2005).
\newblock R2winbugs: a package for running winbugs from r.

\bibitem[Tran et~al., 2019]{tran2019discrete}
Tran, D., Vafa, K., Agrawal, K., Dinh, L., and Poole, B. (2019).
\newblock Discrete flows: Invertible generative models of discrete data.
\newblock In {\em Advances in Neural Information Processing Systems}, pages
  14719--14728.

\bibitem[Uehara et~al., 2016]{uehara2016generative}
Uehara, M., Sato, I., Suzuki, M., Nakayama, K., and Matsuo, Y. (2016).
\newblock Generative adversarial nets from a density ratio estimation
  perspective.
\newblock {\em arXiv preprint arXiv:1610.02920}.

\bibitem[Voter, 1985]{voter1985monte}
Voter, A.~F. (1985).
\newblock A monte carlo method for determining free-energy differences and
  transition state theory rate constants.
\newblock {\em The Journal of chemical physics}, 82(4):1890--1899.

\bibitem[Wang et~al., 2020]{wang2020warp}
Wang, L., Jones, D.~E., and Meng, X.-L. (2020).
\newblock Warp bridge sampling: The next generation.
\newblock {\em Journal of the American Statistical Association},
  (just-accepted):1--31.

\bibitem[Wong et~al., 2020]{wong2020properties}
Wong, J.~S., Forster, J.~J., and Smith, P.~W. (2020).
\newblock Properties of the bridge sampler with a focus on splitting the mcmc
  sample.
\newblock {\em Statistics and Computing}, pages 1--18.

\end{thebibliography}

\end{document}